\newif\iffocs

\iffocs
  \documentclass[10pt, conference, compsocconf]{IEEEtran}
\else
  \documentclass[11pt]{article}
  \usepackage[margin=1in]{geometry}
\fi

\usepackage{amsmath,amsfonts,amssymb,amsthm}
\usepackage{mathtools}
\usepackage{enumerate}
\usepackage{thmtools,thm-restate}
\usepackage{braket}
\usepackage{comment}

\usepackage{mmap}
\usepackage[T1]{fontenc}

\usepackage[utf8]{inputenc}
\usepackage{csquotes}
\usepackage[english]{babel}

\iffocs
\usepackage[backend=biber,
              style=alphabetic,
        maxbibnames=99,
      minalphanames=3,
      maxalphanames=4,
            backref=false]{biblatex}
\else
\usepackage[backend=biber,
              style=alphabetic,
        maxbibnames=99,
      minalphanames=3,
      maxalphanames=4,
            backref=true]{biblatex}
\fi

\usepackage{hyperref}
\hypersetup{
    pdftitle={}, 
    pdfauthor={}, 
    colorlinks=true, 
    linkcolor=blue, 
    citecolor=blue, 
    urlcolor=green 
}

\AtEveryBibitem{
 \clearlist{address}
 \clearfield{date}
 \clearlist{location}
 \clearfield{month}
 \clearfield{series}
 \clearfield{pages}
 \clearlist{organization}
 \clearfield{number}
 \clearlist{language}

 \ifentrytype{book}{}{
  \clearlist{publisher}
  \clearname{editor}
  \clearfield{issn}
  \clearfield{isbn}
  \clearfield{volume}
 }
}

\DeclareFieldFormat{eprint:eccc}{
\mkbibacro{ECCC}\addcolon\space\ifhyperref
    {\href{http://eccc.hpi-web.de/report/#1/}{\nolinkurl{#1}}}
    {\nolinkurl{#1}}}
\DeclareFieldAlias{eprint:ECCC}{eprint:eccc}
\DeclareFieldFormat{eprint:iacr}{Cryptology ePrint\addcolon\space\ifhyperref
    {\href{https://eprint.iacr.org/#1}{\nolinkurl{#1}}}
    {\nolinkurl{#1}}}
\DeclareFieldAlias{eprint:IACR}{eprint:iacr}
\DeclareFieldFormat{eprint:arxiv}{arXiv\addcolon\space\ifhyperref
    {\href{https://arxiv.org/abs/#1}{\nolinkurl{#1}}}
    {\nolinkurl{#1}}}

\DeclareFieldFormat[article]{title}{#1}
\DeclareFieldFormat[inproceedings]{title}{#1}

\DeclareFieldFormat[online]{title}{#1}

\renewbibmacro{in:}{}

\newcommand{\lName}{1}

\newcommand{\donothing}[1]{#1}

\newcommand{\JACM}{\if\lName1\donothing{Journal of the {ACM}}\else{JACM}\fi}
\newcommand{\SICOMP}{\if\lName1\donothing{{SIAM} Journal on Computing}\else{SICOMP}\fi}
\newcommand{\ToC}{\if\lName1\donothing{Theory of Computing}\else{ToC}\fi}
\newcommand{\ToCGS}{\if\lName1\donothing{Theory of Computing Graduate Surveys}\else{ToC}\fi}
\newcommand{\TOCT}{\if\lName1\donothing{{ACM} Transactions on Computation Theory}\else{TOCT}\fi}
\newcommand{\CCjournal}{\if\lName1\donothing{Computational Complexity}\else{CC}\fi}
\newcommand{\CJTCS}{\if\lName1\donothing{Chicago Journal of Theoretical Computer Science}\else{CJTCS}\fi}

\newcommand{\TCS}{\if\lName1\donothing{Theoretical Computer Science}\else{TCS}\fi}
\newcommand{\IPL}{\if\lName1\donothing{Information Processing Letters}\else{IPL}\fi}
\newcommand{\JCSS}{\if\lName1\donothing{Journal of Computer and System Sciences}\else{JCSS}\fi}

\newcommand{\RSA}{\if\lName1\donothing{Random Structures and Algorithms}\else{RSA}\fi}
\newcommand{\JCTA}{\if\lName1\donothing{Journal of Combinatorial Theory, Series A}\else{JCTA}\fi}
\newcommand{\JCTB}{\if\lName1\donothing{Journal of Combinatorial Theory, Series B}\else{JCTB}\fi}
\newcommand{\PJM}{\if\lName1\donothing{Pacific Journal of Mathematics}\else{PJM}\fi}
\newcommand{\QIC}{\if\lName1\donothing{Quantum Information and Computation}\else{QIC}\fi}
\newcommand{\IJQI}{\if\lName1\donothing{International Journal of Quantum Information}\else{IJQI}\fi}

\DefineBibliographyStrings{english}{%
  backrefpage = {p{.}},
  backrefpages = {pp{.}},
}

\newtheorem{theorem}{Theorem}
\numberwithin{theorem}{section}

\newtheorem{lemma}[theorem]{Lemma}
\newtheorem{proposition}[theorem]{Proposition}
\newtheorem{corollary}[theorem]{Corollary}
\newtheorem{definition}[theorem]{Definition}

\newtheorem{question}[theorem]{Question}

\theoremstyle{definition}

\newcommand{\be}{\begin{equation}}
\newcommand{\ee}{\end{equation}}
\newcommand{\eq}[1]{\hyperref[eq:#1]{(\ref*{eq:#1})}}
\renewcommand{\sec}[1]{\hyperref[sec:#1]{Section~\ref*{sec:#1}}}
\newcommand{\thm}[1]{\hyperref[thm:#1]{Theorem~\ref*{thm:#1}}}
\newcommand{\lem}[1]{\hyperref[lem:#1]{Lemma~\ref*{lem:#1}}}
\newcommand{\defn}[1]{\hyperref[def:#1]{Definition~\ref*{def:#1}}}
\newcommand{\prop}[1]{\hyperref[prop:#1]{Proposition~\ref*{prop:#1}}}
\newcommand{\cor}[1]{\hyperref[cor:#1]{Corollary~\ref*{cor:#1}}}
\newcommand{\fig}[1]{\hyperref[fig:#1]{Figure~\ref*{fig:#1}}}
\newcommand{\tab}[1]{\hyperref[tab:#1]{Table~\ref*{tab:#1}}}
\newcommand{\alg}[1]{\hyperref[alg:#1]{Algorithm~\ref*{alg:#1}}}
\newcommand{\app}[1]{\hyperref[app:#1]{Appendix~\ref*{app:#1}}}
\newcommand{\conj}[1]{\hyperref[conj:#1]{Conjecture~\ref*{conj:#1}}}
\newcommand{\chap}[1]{\hyperref[chap:#1]{Chapter~\ref*{chap:#1}}}
\newcommand{\clm}[1]{\hyperref[clm:#1]{Claim~\ref*{clm:#1}}}
\newcommand{\fct}[1]{\hyperref[fct:#1]{Fact~\ref*{fct:#1}}}
\newcommand{\itm}[1]{\hyperref[itm:#1]{(\ref*{itm:#1})}}
\newcommand{\qstn}[1]{\hyperref[qstn:#1]{Question~\ref*{qstn:#1}}}
\makeatletter
\newcommand*\rel@kern[1]{\kern#1\dimexpr\macc@kerna}
\newcommand*\widebar[1]{%
  \begingroup
  \def\mathaccent##1##2{%
    \rel@kern{0.8}%
    \overline{\rel@kern{-0.8}\macc@nucleus\rel@kern{0.2}}%
    \rel@kern{-0.2}%
  }%
  \macc@depth\@ne
  \let\math@bgroup\@empty \let\math@egroup\macc@set@skewchar
  \mathsurround\z@ \frozen@everymath{\mathgroup\macc@group\relax}%
  \macc@set@skewchar\relax
  \let\mathaccentV\macc@nested@a
  \macc@nested@a\relax111{#1}%
  \endgroup
}
\makeatother
\renewcommand{\bar}{\widebar}
\newcommand{\eprint}[1]{{\small \upshape \tt \href{http://arxiv.org/abs/#1}{#1}}}

\newcommand{\B}{\{0,1\}}

\newcommand{\tOmega}{\tilde{\Omega}}

\newcommand{\tO}{\tilde{O}}

\DeclareMathOperator{\R}{R}
\DeclareMathOperator{\Q}{Q}

\DeclareMathOperator{\adeg}{adeg}

\DeclareMathOperator{\Dom}{Dom}
\DeclareMathOperator{\bias}{bias}

\DeclareMathOperator{\polylog}{polylog}
\DeclareMathOperator{\poly}{poly}
\DeclareMathOperator*{\E}{\mathbb{E}}
\DeclareMathOperator{\bR}{\mathbb{R}}

\DeclareMathOperator{\cost}{cost}
\DeclareMathOperator{\outpt}{output}
\DeclareMathOperator{\pred}{pred}

\DeclareMathOperator{\tran}{tran}

\DeclareMathOperator{\Bernoulli}{Bernoulli}

\newcommand{\triv}{\textsc{Trivial}}
\DeclareMathOperator{\score}{score}

\DeclareMathOperator{\hs}{hs}
\DeclareMathOperator{\ls}{ls}
\DeclareMathOperator{\Brier}{Brier}

\DeclareMathOperator{\JS}{JS}
\DeclareMathOperator{\h}{h}
\DeclareMathOperator{\Ess}{S}
\DeclareMathOperator{\Conv}{Conv}
\DeclareMathOperator{\RCC}{RCC}
\DeclareMathOperator{\QCC}{QCC}
\DeclareMathOperator{\QDT}{QDT}
\DeclareMathOperator{\RDT}{RDT}
\DeclareMathOperator{\Rcirc}{Rcirc}
\DeclareMathOperator{\RNC}{RNC1}
\DeclareMathOperator{\RTC}{RTC0}

\DeclareMathOperator{\PQ}{PQ}

\DeclareMathOperator{\rank}{rank}

\newcommand{\tTheta}{\widetilde{\Theta}}

\newcommand{\bN}{\mathbb{N}}
\newcommand{\bE}{\mathbb{E}}
\newcommand{\Var}{\mathrm{Var}}

\iffocs
  \excludecomment{fulltext}
\else
  \includecomment{fulltext}
\fi

\begin{document}

\title{A New Minimax Theorem for Randomized Algorithms
}

\iffocs
\else
    \author{
    Shalev Ben{-}David\\
    \small University of Waterloo\\
    \small \texttt{shalev.b@uwaterloo.ca}
    \and
    Eric Blais\\
    \small University of Waterloo\\
    \small \texttt{eric.blais@uwaterloo.ca}
    }
\fi

\date{}
\maketitle

\begin{abstract}
The celebrated minimax principle of Yao (1977)
says that for any Boolean-valued
function $f$ with finite domain,
there is a distribution $\mu$ over the domain of $f$
such that computing $f$ to error $\epsilon$
against inputs from $\mu$ is just as hard
as computing $f$ to error $\epsilon$ on worst-case inputs.
Notably, however, the distribution $\mu$
depends on the target error level $\epsilon$:
the hard distribution which is tight
for bounded error might be trivial to solve
to small bias, and the hard distribution
which is tight for
a small bias level might be far from tight
for bounded error levels.

In this work, we introduce a new type of minimax
theorem which can provide a hard distribution
$\mu$ that works for all bias levels at once.
We show that this works for randomized query complexity,
randomized communication complexity,
some randomized circuit models, quantum query
and communication complexities, approximate
polynomial degree, and approximate logrank.
We also prove an improved version of Impagliazzo's
hardcore lemma.

Our proofs rely on two innovations over
the classical approach of using Von Neumann's
minimax theorem or linear programming duality.
First, we use Sion's minimax theorem to
prove a minimax theorem for ratios
of bilinear functions representing the cost
and score of algorithms.

Second, we introduce
a new way to analyze low-bias randomized algorithms
by viewing them as ``forecasting algorithms''
evaluated by a certain proper scoring rule.
The expected score of the forecasting
version of a randomized algorithm
appears to be a more fine-grained
way of analyzing the bias of the algorithm.
We show that such expected scores have many
elegant mathematical properties:
for example, they can be amplified linearly
instead of quadratically.
We anticipate forecasting algorithms will
find use in future work in which a fine-grained
analysis of small-bias algorithms is required.
\end{abstract}

\iffocs
\else
    \clearpage
    {\small\tableofcontents}
    \clearpage
\fi

\section{Introduction}
\label{sec:introduction}

Yao's minimax principle~\cite{Yao77} is a central tool
in the analysis of randomized algorithms in many
different models of computation. 
In its most commonly-used form,
it states that for every Boolean-valued function
$f$ with a finite domain,
if $\mathcal{R}^{(c)}$ denotes the set of randomized
algorithms with worst-case cost at most $c$ and
$\Delta$ denotes the set of distributions
over the domain of $f$, then
{
\[
\min_{R \in \mathcal{R}^{(c)}} \max_{\mu \in \Delta} \Pr[ R(x) \neq f(x) ]
= \max_{\mu \in \Delta} \min_{R \in \mathcal{R}^{(c)}} \Pr[ R(x) \neq f(x) ]
\]
}
with both probabilities being over the choice of $x$ drawn from $\mu$ and the internal randomness of $R$. This identity implies that there exists a distribution $\mu$ for which any algorithm that computes $f$ with bounded error over inputs drawn from $\mu$ must have cost at least $\R(f)$,
the cost of computing $f$ to worst-case bounded error. 
But it does not say anything else about $\mu$ itself. Notably, 
\begin{enumerate}[I.]
\item The minimax principle does not guarantee that the resulting distribution $\mu$ must be balanced on the sets $f^{-1}(0)$ and $f^{-1}(1)$. 
\item More generally, it does not rule out the possibility that $f$ is very easy to compute by randomized algorithms that are only required to output the correct value with probability at least $\frac{1+\gamma}2$ for some small bias measure $\gamma > 0$ over inputs drawn from the distribution $\mu$.
\end{enumerate}
A separate application of the minimax principle can be used to show that there is a distribution $\mu'$ for which all randomized algorithms computing $f$ with bias $\gamma$ over $\mu'$ have cost at least $\R_{\frac{1-\gamma}2}(f)$
(the cost of computing $f$ to worst-case error $(1-\gamma)/2$), but then there is no guarantee that randomized algorithms with bounded error over $\mu'$ must have cost anywhere close to $\R(f)$.

Intuitively, it seems reasonable to expect that for every
function $f$, there is a distribution $\mu$ for $f$
that addresses issues I and II: a distribution
that is balanced on $f^{-1}(0)$ and $f^{-1}(1)$,
and which is at least slightly hard even to solve
to a small bias level $\gamma$.

\begin{question}[Informal]\label{qstn:informal}
Is there a distribution $\mu$ which certifies
the hardness of $f$ for all
bias levels $\gamma>0$ at the same time?
\end{question}

More formally, observe that the cost of computing $f$
to worst-case bias $\gamma$ cannot be smaller than
$\gamma^2\R(f)$. This is because randomized algorithms
can be \emph{amplified}: by repeating an algorithm
$O(1/\gamma^2)$ times and outputting the majority
vote of the runs, we can increase its bias from
$\gamma^2$ to $\Omega(1)$. Therefore, a natural
refinement of \qstn{informal} is as follows.

\begin{question}[Refinement of \qstn{informal}]
\label{qstn:refined}
Is there a distribution $\mu$ such that for all
bias levels $\gamma>0$, any algorithm computing
$f$ to bias $\gamma$ against $\mu$ must have
cost at least $\Omega(\gamma^2\R(f))$?
\end{question}

\qstn{refined} is the primary focus of this work.
We answer it affirmatively in a variety of computational
models (we can handle most models in which amplification
and Yao's minimax principle both apply).
We note that the distribution satisfying
the conditions of \qstn{refined} is hard
for bounded error in Yao's sense, since each algorithm
solving $f$ to bounded error against $\mu$ must
have cost at least $\Omega(\R(f))$. In addition to this,
such $\mu$ must also be perfectly balanced between
$0$- and $1$-inputs of $f$ (by considering the limit
as $\gamma\to 0$), and must remain somewhat hard
to solve even to small bias levels.

The study of \qstn{refined} has led us to consider
randomized forecasting algorithms which output
\emph{probabilistic confidence predictions}
about the value of $f(x)$, instead of a Boolean guess
for $f(x)$. When evaluated using a certain proper
scoring rule, the best possible score of a forecasting
algorithm is intimately related to the best possible
bias of a randomized algorithm; in fact, the score
appears to be a more fine-grained way of measuring
the bias. Scores of forecasting algorithms
appear to be the ``right'' way of measuring
the success of randomized algorithms, as such scores
satisfy elegant mathematical properties.
The following question, which we answer affirmatively,
turns out to be a strengthening of \qstn{refined}.

\begin{question}\label{qstn:score}
Is there a distribution $\mu$ such that for all
$\eta>0$, any forecasting algorithm which
achieves expected score at least $\eta$ against
$\mu$ must have cost at least $\Omega(\eta\R(f))$?
\end{question}

\subsection{Motivation from joint computation}

The answers to \qstn{refined} and \qstn{score} have
a direct impact on the study of
composition theorems and joint computation problems in
randomized computational models: a natural approach for such problems
involves first applying a minimax theorem and then establishing the
required inequalities in the deterministic distributional setting.
However, as observed by Shaltiel~\cite{Sha03}
this approach runs into trouble if the hard distribution
is easy to solve to small bias.
Specifically, Shaltiel considered distributions $\mu$
which are hard to solve most of the time, but which
give a completely trivial input with small probability $\gamma$.
Then computing $n$ independent copies from $\mu$
is a little easier than $n$ times the cost of computing $f$,
because on average, $\gamma n$ of the copies are trivial;
the cost of computing $n$ independent inputs
from $\mu$ is at most $(1-\gamma)n$ times the cost of solving $f$.

Things get even worse when the inputs have a promised
correlation, as can happen when proving composition theorems.
For a concrete example, consider the partial function
$\triv_n$, which is defined on domain $\{0^n,1^n\}$
and maps $0^n\to 0$ and $1^n\to 1$. Suppose we want
to prove a composition lower bound with $\triv_n$
on the outside: that is, we want to show that
for every function $f$, computing $\triv_n$
composed with $n$ copies of $f$ requires $\Omega(\R(f))$
cost. In other words, we want to lower bound the cost
of an algorithm which outputs $0$ when given $n$
$0$-inputs to $f$, outputs $1$ when given $n$
$1$-inputs to $f$, and outputs arbitrarily when given
some other type of input.

Now, if we try to lower bound this using the hard distribution
from Yao's minimax principle, then the distribution
might give a trivial input
with small probability $\gamma$,
as Shaltiel observed; but then so long as $n=\Omega(1/\gamma)$,
one of the inputs to $f$ will be trivial with high probability,
and we can solve this ``all-$0$s vs all-$1$s'' problem
simply by searching for the trivial copy -- potentially
much faster than the worst-case cost of computing a single
copy of $f$!

The hard distributions we give in this work solve this issue
by being hard for all bias levels. In our companion manuscript
\cite{BB20b}, we use one of the query versions of our
minimax theorem 
\iffocs
\else
  (\thm{query_shaltiel_free})
\fi
to prove a new composition theorem
for randomized query complexity.

\subsection{Main tools}

\iffocs
\else
  \paragraph{Minimax theorem for cost/score ratios.}
\fi
The first main result is a new minimax theorem for the ratio of the  cost and score of randomized algorithms. A special case of the theorem with a simple formulation is as follows.

\begin{theorem}
\label{thm:specialcase_minimax}
\iffocs
\else
  [Special case of \thm{main_minimax}]
\fi
Let $\mathcal{R}$ be a set of randomized algorithms that can be expressed as a convex subset of a real topological vector space. Let $S$
be a nonempty finite set, and let $\Delta$ be the
set of all probability distributions over $S$,
viewed as a subset of $\bR^{|S|}$. Let
$\cost\colon\mathcal{R}\times\Delta\to(0,\infty)$
and $\score\colon\mathcal{R}\times\Delta\to[-\infty,\infty)$
be continuous bilinear functions.
Then using the convention $r/0=\infty$ for all $r\in(0,\infty)$
and the notation $r^+=\max\{r,0\}$
for all $r\in[-\infty,\infty]$, we have
\[\adjustlimits\inf_{R\in\mathcal{R}}\max_{x\in S}
\frac{\cost(R,x)}{\score(R,x)^{+}}
=\adjustlimits\max_{\mu\in\Delta}\inf_{R\in\mathcal{R}}
\frac{\cost(R,\mu)}{\score(R,\mu)^{+}}.\]
Further, all of the above maximums are attained.
\end{theorem}

The general version of the minimax theorem in 
\iffocs
\else
  \thm{main_minimax} 
\fi
shows that the same identity holds even when the $\cost$ and $\score$ functions are semicontinuous and saddle (but not necessarily linear) under some mild additional restrictions. Furthermore, a variant of the theorem also holds when we consider convex and compact subsets of distributions over the finite set $S$ instead of the set $\Delta$ of all distributions over that set.

Minimax theorems for ratios of semicontinuous and saddle functions as in 
\iffocs
\else
  \thm{main_minimax} 
\fi
do not seem to have appeared in the literature previously in the precise form we need, but as we show in 
\iffocs
\else
  \sec{minimax}, 
\fi
they can be obtained by extending Sion's minimax theorem~\cite{Sio58} with standard arguments. 
We believe that the main contribution of
\iffocs
\else
  \thm{main_minimax}
\fi
is in its interpretation for randomized algorithms. Various extensions and variations of Yao's minimax theorem have been considered in the computer science literature previously~\cite{Yao77,Imp95,Ver98,Bra15,BGK+18,BB19}, but all of them appear to consider the cost of an algorithm (with the minimax theorem applied to algorithms with a fixed worst-case score), the score of an algorithm (with the cost being fixed), or a linear combination of the two. None of those variants suffice to answer the questions raised at the beginning of the introduction or to establish the results in the following subsections; what was needed in those cases was a minimax theorem for the \emph{ratio} of the cost/score of randomized algorithms, and we suspect that this ratio minimax theorem will find further applications in computer science in the future as well.

\iffocs
\else
  \paragraph{Forecasting algorithms and linear amplification.}
\fi
To convert the statements obtained from 
\iffocs
\else
  \thm{main_minimax} 
\fi
regarding the cost/score ratios of randomized algorithms under some distribution $\mu$ into more familiar lower bounds on the cost of randomized algorithms that achieve some bias on $\mu$, we need a \emph{linear amplification} theorem. Ideally, we would like to argue that if there exists a randomized algorithm $R$ with bias $\gamma$ on $\mu$, then by combining $O(1/\gamma)$ instances of $R$ we can obtain a randomized algorithm $R'$ with
$\cost(R',\mu) = O\left(\frac1\gamma \cdot \cost(R,\mu)\right)
= O\left(\frac{\cost(R,\mu)}{\bias_f(R,\mu)}\right)$ and constant bias. Unfortunately, such a linear amplification property does not hold for most models of randomized algorithms, where amplification from bias $\gamma$ to bounded error requires combining $O(1/\gamma^2)$ instances of the original algorithm. To obtain a linear amplification result, we must turn our attention away from bias and error and consider other score functions instead.\footnote{The astute reader may have noticed that we obtain linear amplification if we simply set the score to be the squared bias of the randomized algorithm. That is true, but this approach does not work in conjunction with the ratio minimax theorem since this score function no longer satisfies the appropriate saddle property requirements of that theorem; this is why we instead consider forecasting algorithms as described below.}

To describe our score function, we first generalize
our computational model from randomized algorithms that
output $0$ or $1$ to \emph{forecasting algorithms}, which are randomized algorithms that output a \emph{confidence value} in $[0,1]$ for the value $f(x)$ of the function $f$ on their
given input $x$. 
A ``low'' confidence prediction is a value close to $\frac12$ whereas a ``high'' confidence prediction would be a value very close to $0$ or to $1$.
There are many natural ways to assign a score to a confidence
value for $f(x)$.
The study of such scoring rules and their properties has a rich history in the statistics and decision theory communities (see for instance~\cite{BSS05,GR07} and references therein); we discuss some fundamental scoring rules and give relations between them in 
\iffocs
\else
  \sec{forecasting}. 
\fi
Of particular importance to our main purpose is the scoring rule $\hs \colon [0,1] \to [-\infty,1]$ defined by
\[
\hs_f(p) = 
\begin{cases}
1 - \sqrt{\frac{1-p}{p}} & \mbox{when } f(x) = 1 \\
1 - \sqrt{\frac{p}{1-p}} & \mbox{when } f(x) = 0.
\end{cases}
\]
Define the score of a forecasting algorithm $R$ on an input $x$ in the domain of $f$ to be $\score_{\hs,f}(R,x) = \E[ \hs_f(R(x)) ]$, the expectation of the hs score of the output of $R$ over the internal randomness of $R$. Then linear amplification does hold for this score function.

\begin{lemma}
\label{lem:linear-amplification}
For any Boolean-valued function $f$, any forecasting algorithm $R$, and any $k \ge 1$, there is a forecasting algorithm $R'$ that combines the outputs of $k$ instances of $R$ and satisfies
\[
\score_{\hs,f}(R',x) \ge 1 - (1 - \score_{\hs,f}(R,x))^k
\]
for every $x$ in the domain of $f$. In particular, when $k = \max_x \frac{2}{\score_{\hs,f}(R,x)}$ then 
for each $x \in \Dom(f)$,
$\score_{\hs,f}(R',x) \ge 1 - e^{-2} > 0.85$.
\end{lemma}

To the best of our knowledge, \lem{linear-amplification}
has not previously appeared in the literature.
This lemma is sensitive to the precise definition
of $\hs_f$; other scoring rules do not appear to satisfy
this amplification property, which is crucial for
the proof of our main results. Additionally,
the scoring rule $\hs_f$ is special because
there is a close connection between $\hs$ score of forecasting
algorithms and the bias of randomized algorithms.

\begin{lemma}
\label{lem:bias-hs-conversion}
For any Boolean-valued function $f$, any distribution $\mu$ on $\Dom(f)$, and any parameter $\gamma > 0$,
\begin{itemize}
\item If there exists a randomized algorithm $R$ with $\bias_f(R,\mu) =1- 2\Pr[ R(x) \ne f(x) ] \ge \gamma$,
then there is a forecasting algorithm $R'$ with $\score_{\hs,f}(R',\mu) \ge 1 - \sqrt{1-\gamma^2} \ge \gamma^2/2$, and
\item If there exists a forecasting algorithm $R$ with $\score_{\hs,f}(R,\mu) \ge \gamma$ then there is a randomized algorithm $R'$ with $\bias_f(R',\mu) \ge \gamma$.
\end{itemize}
Moreover, in both cases $R'$ can be explicitly constructed
from $R$ by modifying its output.
\end{lemma}

\lem{linear-amplification} and \lem{bias-hs-conversion} can be used to reprove the fact that $O(1/\gamma^2)$ instances of a bias-$\gamma$ randomized algorithms can be combined to obtain a bounded-error algorithm; combining those lemmas (or, more precisely, specific instantiations of these lemmas that account for the explicit constructions of the relevant algorithms and their costs) with the minimax theorem also leads to new results as described in the next section.

\subsection{Main results}
\label{sec:applications}

\iffocs
\else
  \paragraph{Hard distributions for bounded error and small bias.}
\fi
The minimax theorem for cost/score ratios and linear amplification of forecasting algorithms can be combined to show that for many measures of randomized complexity, for every Boolean-valued function 
$f$ with finite domain there exists a single distribution $\mu$ on which it is hard to compute $f$ with bounded error \emph{and} with (any) small bias. For example, letting $\RDT(f)$ denote the minimum (worst-case) query complexity of a randomized algorithm computing $f$ (or equivalently the minimum worst-case depth of a decision tree computing $f$) with error at most $\frac13$ on every input in $\Dom(f)$ and $\RDT_{\dot{\gamma}}^\mu$ denote the minimum query complexity of a randomized algorithm that has error probability at most $\dot{\gamma} := \frac{1-\gamma}{2}$ when inputs are drawn from $\mu$, we obtain the following result.

\begin{theorem}
\label{thm:query-avg-worst}
For any non-constant partial function $f\colon \B^n \to\B$, 
there exists a distribution $\mu$ on $\Dom(f)$ such that for every $\gamma \in [0,1]$,
\[
\RDT_{\dot{\gamma}}^\mu(f) = \Omega\big( \gamma^2 \RDT(f)\big).
\]
\end{theorem}

We establish analogous theorems for multiple other computational models 
\iffocs
\else
as well:
\begin{center}
\begin{tabular}{lll}
Randomized communication complexity & $\RCC_{\dot{\gamma}}^\mu(f) = \Omega\big( \gamma^2 \RCC(f)\big)$ & \cor{communication} \\

Quantum query complexity & $\QDT_{\dot{\gamma}}^\mu(f) = \gamma \cdot \tOmega\big(  \QDT(f)\big)$ & \thm{minimax_quantum} \\

Quantum communication complexity & $\QCC_{\dot{\gamma}}^\mu(f) = \gamma \cdot \tOmega\big(  \QCC(f)\big)$ & \thm{quantum-communication} \\

Polynomial degree & $\deg_{\dot{\gamma}}^\mu(f)=\gamma\cdot\tOmega(\adeg(f))$ & \thm{minimax_polynomials} \\

Log-rank complexity & $\log\rank_{\dot{\gamma}}^\mu(f)=\gamma\cdot\tOmega(\log\rank_{1/3}(f))$ & \thm{minimax_logrank} \\

Circuit complexity & $\Rcirc_{\dot{\gamma}}^\mu(f) = \gamma^2 \cdot \tOmega\big(\Rcirc(f)\big)$ & \thm{main-circuit} \\

Log-depth circuit complexity & $\RNC_{\dot{\gamma}}^\mu(f) = \gamma^2 \cdot \tOmega\big(  \RNC(f)\big)$ & \thm{nc1-circuit} \\

Threshold circuit complexity & $\RTC_{\dot{\gamma}}^\mu(f) = \gamma^2 \cdot \tOmega\big(  \RTC(f)\big)$ & \thm{tc0-circuit} \\

\end{tabular}
\end{center}
\fi
(Note that as in \thm{query-avg-worst}, the novel aspect of all these results is that they guarantee that for each of the stated inequalities, there exists a \emph{single} distribution $\mu$ that satisfies the inequality for \emph{every} value of $\gamma$ simultaneously.)

\iffocs
\else
  \paragraph{Hard distributions for forecasting algorithms.} 
\fi
The theorems listed above settle
\qstn{refined} in the affirmative for the specified
models. For the models with quadratic dependence
on $\gamma$
(i.e.\ randomized query complexity, randomized communication
complexity, and the various randomized circuit models),
we also get hard distributions which lower bound
the expected score of a forecasting algorithm,
settling \qstn{score} affirmatively.

\iffocs
\else
  \paragraph{Distinguishing power of randomized algorithms and protocols.}
\fi
In the communication complexity setting, we can also analyze how well a randomized communication protocol computes a function $f \colon \mathcal{X} \times \mathcal{Y} \to \B$ via its communication transcripts. Let $\tran(R,\mu_0)$ denote the distribution on communication transcripts of the randomized protocol $R$ on inputs drawn from $\mu$. Then one way to measure how well $R$ is able to distinguish $0$- and $1$-inputs of $f$ is to measure the Hellinger distance between the distributions $\tran(R,\mu_0)$ and $\tran(R,\mu_1)$ of transcripts of $R$ on some distributions $\mu_0$ over $f^{-1}(0)$ and $\mu_1$ over $f^{-1}(1)$. 
We can use the minimax and linear amplification theorems to give a strong upper bound on this Hellinger distance as a measure of the cost of the protocol.

\begin{theorem}
\label{thm:cc-h}
For any non-constant partial function $f\colon \mathcal{X} \times \mathcal{Y} \to \B$ over finite sets $\mathcal{X}$ and $\mathcal{Y}$,
there is a pair of distributions $\mu_0$ on $f^{-1}(0)$ and $\mu_1$ on $f^{-1}(1)$
such that for any randomized communication protocol $R$, the squared Hellinger distance between 
the distribution of its transcripts on $\mu_0$ and $\mu_1$ is bounded above by
\begin{align*}
\h^2\big(\tran(R,\mu_0),&\tran(R,\mu_1)\big) \\
 &= O\left( \frac{\min\{\cost(R,\mu_0),\cost(R,\mu_1)\}}{\RCC(f)} \right).
\end{align*}
Here $\cost(R,\mu)$ denotes the expected amount
of communication the protocol $R$ transmits
when given inputs from $\mu$.
\end{theorem}

\iffocs
\else
  \thm{query_shaltiel_free} establishes 
\fi
an analogous result for query complexity.
In our companion paper~\cite{BB20b}, that theorem is one of the ingredients that enables us to establish a new composition theory for query complexity. 

\iffocs
\else
  \paragraph{Hardcore lemma.}
\fi
Impagliazzo's Hardcore Lemma~\cite{Imp95} states that for every $\epsilon, \delta > 0$, if every circuit $C$ of size at most $s$ computes $f$ with error at least $\delta$ on the uniform distribution, then there is a $\delta$-regular distribution $\mu = \mu(\delta,\epsilon)$ for which every circuit that computes $f$ with bias at least $\epsilon$ on the distribution $\mu$ must have size $\Omega(\epsilon^2 s)$. Informally, the lemma shows that if a function $f$ is mildly hard on average, it is because it is ``very'' hard to compute on a fairly large subset of its inputs. But, interestingly, this version of the hardcore lemma leaves open the possibility that the hard core might be different for various levels $\epsilon$ of hardness. Using our main theorems, we can show that this is not the case.

\begin{theorem}
\label{thm:univ-hardcore}
There exists a universal constant $c > 0$ such that for any $\delta > 0$ and function $f : \{0,1\}^n \to \{0,1\}$, if every circuit $C$ of size at most $s$ satisfies
$\Pr[ C(x) = f(x) ] \le 1 - \delta$ when the probability is taken over the uniform distribution of $x$ in $\{0,1\}^n$,
then there is a distribution $\mu$ with min-entropy $\delta$ such that for every $\epsilon > 0$, any circuit $C'$ of size at most $c \cdot \epsilon^2/\log(1/\delta) \cdot s$ has success probability bounded by
\[
\Pr[ C'(x) = f(x) ] \le \frac{1+\epsilon}2.
\]
\end{theorem}

The proof of \thm{univ-hardcore} follows closely the original argument of Nisan in~\cite{Imp95} that established the hardcore lemma via a minimax theorem. Since that original work, many extensions and different proofs of the hardcore lemma have been established (e.g.,~\cite{Imp95,KS03,BHK09,TTV09}), but to the best of our knowledge \thm{univ-hardcore} represents the first version of the lemma which gives a single distribution $\mu$ which is hard for all values of $\epsilon > 0$ simultaneously.

\subsection{Recent independent work}

In independent work concurrent with this one,
Bassilakis, Drucker, G\"{o}\"{o}s, Hu, Ma, and Tan \cite{BDG+20}
showed the existence of a certain hard distribution
for randomized query complexity. They showed
every Boolean function $f$ has hard distributions
$\mu_0$ and $\mu_1$ (on $0$- and $1$-inputs respectively)
such that given query access to $k$ independent
samples from $\mu_b$, it is still necessary to use
$\Omega(\R(f))$ queries to the bits of the samples
in order to decide the value of $b\in\B$ to bounded error.

The guarantee on the hard distribution provided by \cite{BDG+20}
is formally stronger than the one we provide in
\iffocs
\else
  \thm{query_shaltiel_free} 
\fi
(though in our companion manuscript \cite{BB20b}, we prove a new composition theorem for randomized query complexity,
and use it to conclude that
the guarantee of \cite{BDG+20} turns out to be
equivalent to the guarantee of
\iffocs
\else
  \thm{query_shaltiel_free} in our current work).
\fi
The tools used by \cite{BDG+20} are also completely different:
they use arguments specific to query complexity
that construct the hard distribution more explicitly,
but their arguments do not generalize to other models
such as communication complexity or circuit complexity.

\iffocs
  \subsection{Overview of the remaining sections of the full version}
\else
  \subsection{Organization and overview of the remaining sections}

  \begin{description}
\fi

\iffocs
\else
  \item[\sec{minimax}] 
\fi
is devoted to proving the main minimax theorem for the cost/score ratio of randomized algorithms. The main result of that section is 
\iffocs
\else
  \thm{main_minimax}; 
\fi
the rest of the section is devoted to introducing the mathematical notions and preliminaries required to obtain a proof of that theorem from Sion's minimax theorem.

\iffocs
\else
  \item[\sec{forecasting}] 
\fi
introduces the basic definitions and some basic scoring rules for forecasting algorithms. The section establishes some of the core properties of scoring functions, including notably connections between the best score achievable by forecasting algorithms on distributions over inputs and various distance measures on those distributions. The final portions of this section then establish the main linear amplification theorem in general form in \lem{linear-amplification} and the general form of the conversion between randomized and forecasting 
\iffocs
\else
  algorithms in \lem{conversion}.
\fi

\iffocs
\else
  \item[\sec{query}] 
\fi
focuses on the query and communication complexity settings. Conversions between randomized and forecasting algorithms in the query complexity setting are straightforward, but there is one significant challenge in applying the linear amplification theorem to obtain the results in \thm{query-avg-worst} and
\iffocs
\else
  \thm{query_shaltiel_free}: 
\fi
the cost and score of a randomized algorithm $R$ on an input $x$ can both depend on $x$ itself. This is a problem because to obtain a constant score (and after the final conversion, a bounded-error randomized algorithm), we want to amplify $R$ with a number $k$ of copies that depends on the score of $R$ on $x$---but since we don't know $x$ we don't know what $\score(R,x)$ is either. We get around this problem with odometer arguments:
by empirically estimating the expected number of queries $R$
makes on $x$, we can obtain effective bounds on the number $k$ of copies of $R$ that we need to obtain successful amplification.

As we show in the section, the communication complexity results 
\iffocs
\else
  \cor{communication} and \thm{cc-h} 
\fi
follow immediately from their query complexity analogues.

\iffocs
\else
  \item[\sec{quantum}] 
\fi
establishes the results in the quantum query and communication complexity settings. Unlike in the classical setting, amplification that is linear in the bias of an algorithm \emph{does} hold in the quantum query complexity setting. However, the proof of 
\iffocs
\else
  \thm{minimax_quantum}
\fi
requires that the set of algorithms
must be representable as a convex subset of a real topological space, and that the cost of an algorithm
is a convex function on this set. It is not immediately clear how quantum query algorithms can satisfy this condition,
because in the usual definition, the cost of a mixture
of two quantum algorithms would be the \emph{maximum}
of the costs of the algorithms rather than the average.
To overcome this issue, we instead establish 
\iffocs
\else
  \thm{minimax_quantum} 
\fi
via consideration of what we call \emph{probabilistic} quantum algorithms, which correspond to probability distributions over quantum algorithms and do easily satisfy the appropriate convexity requirements.
Probabilistic quantum algorithms are harder to amplify
than regular quantum algorithms (due to their
lack of coherence), but we show that a linear
amplification theorem still holds.

Another important difference between the quantum and the classical setting is that the communication complexity 
\iffocs
\else
  result, \thm{quantum-communication}, 
\fi
is not implied by the analogous query complexity result. Nonetheless, the same argument used for quantum query algorithms also holds for quantum communication protocols as well. We complete the proof of
\iffocs
\else
  \thm{quantum-communication} 
\fi
by first providing an abstraction of the query complexity argument 
\iffocs
\else
  in \thm{quantum_abstraction} 
\fi
and then showing how communication protocols satisfy the conditions of this abstract theorem.

\iffocs
\else
  \item[\sec{polynomials}] 
\fi
considers the approximate polynomial degree and the logrank complexity of functions.
As with quantum query complexity, approximate polynomial
degree satisfies an amplification theorem that is linear
in the bias, meaning that we do not need to use
forecasting algorithms or scoring rules. However,
also as with quantum query complexity, polynomials
and their cost do not satisfy the right convexity
requirements, as the degree of a mixture of two
polynomials is not the average of their degrees.
We overcome this by considering probabilistic
polynomials. Proving an amplification theorem
for probabilistic polynomials turns out to
be somewhat tricky, and requires tools from
approximation theory such as Jackson's theorem.

Approximate logrank inherits all of the problems of
approximate polynomial degree, and adds a few more.
To handle approximate logrank, we switch over
to the nearly-equivalent model of the logarithm
of the approximate gamma $2$ norm, and then
use the previous trick of considering the
\emph{probabilistic} approximate gamma $2$ norm.
To prove an amplification theorem
for probabilistic gamma $2$ norm we apply
the same tools as for probabilistic polynomials.

\iffocs
\else
  \item[\sec{circuits}] 
\fi
establishes the circuit complexity results. There are two main hurdles in establishing 
\iffocs
\else
  \thm{main-circuit}. 
\fi
The first is that the notion of randomized circuits is not as trivially extendable to forecasting circuits as in other computational models.  We show that this conversion can be done efficiently when we discretize the set of confidence values that can be returned by forecasting circuits, and that this discretization does not affect the guaranteed relations between score and bias. The second is that the overhead required to combine the output of multiple instances of a randomized circuit during linear amplification is not trivial. This second hurdle can be overcome with the use of efficient circuit constructions for elementary arithmetic operations and the iterated addition problem.

The proof of the universal hardcore lemma
\iffocs
\else
  in \thm{univ-hardcore} 
\fi
is obtained via a slight generalization of the ratio minimax theorem. 
\iffocs
\else
  This variant of the minimax theorem is stated in \lem{minimax-circuits-regular} and the rest of the proof of \thm{univ-hardcore} is presented in \sec{hardcore}.
\fi

\iffocs
\else
\end{description}
\fi

\subsection{Further remarks and open problems}

We make a few remarks regarding other possible generalizations
of Yao's original minimax theorem. First, one may wonder
why we provide a hard distribution $\mu$ satisfying
$\R^\mu_{\dot{\gamma}}(f)=\Omega(\gamma^2\R(f))$
for all $\gamma$, rather than the stronger statement
$\R^\mu_{\dot{\gamma}}(f)=\Omega(\R_{\dot{\gamma}}(f))$
for all $\gamma$. In other words, we've stated our lower bounds
in terms of the bounded-error randomized cost $\R(f)$,
which required amplification; why not directly compare
the average-case complexity to bias $\gamma$,
denoted $\R^\mu_{\dot{\gamma}}(f)$, to the worst-case complexity
to bias $\gamma$, denoted $\R_{\dot{\gamma}}(f)$?

The reason is that this stronger version of the minimax
is actually false: that is, there need not be a distribution
$\mu$ for which
$\R^\mu_{\dot{\gamma}}(f)=\Omega(\R_{\dot{\gamma}}(f))$
for all $\gamma$ (even though for every given $\gamma$,
such a distribution $\mu$ that depends on $\gamma$ does exist,
by Yao's minimax theorem).
For a counterexample, consider the query complexity model.
Let $f$ be the Boolean function on $n+m+1$ bits,
where if the first bit is $0$ the function $f$ evaluates
to the parity of the next $m$ bits, whereas if the first bit is
$1$ the function $f$ evaluates to the majority of the last $n$ bits.
Say we take $n=m^2$. Then, since parity is hard to compute
even to small bias, we have $\R_{\dot{\gamma}}(f)\ge m$
for all $\gamma$. We also have $\R_{1/3}(f)=\Omega(m^2)$,
since majority on $m^2$ bits requires $\Omega(m^2)$ queries.
Now, consider any distribution $\mu$ over the domain of $f$.
If $\mu$ places nonzero probability mass on inputs with
first bit $1$, then $\mu$ can necessarily be solved to
some sufficiently small bias using at most $2$ queries
(one query to the first bit of the input, and one to a random
position in the input to majority). In this case, we would
have $\R_{\dot{\gamma}}^\mu(f)=O(1)$ and
$\R_{\dot{\gamma}}(f)=\Omega(\sqrt{n})$ for this sufficiently
small $\gamma$. Alternatively, if $\mu$ places zero probability
mass on inputs with first bit $1$, then solving $f$
against $\mu$ is solving parity on $m=O(\sqrt{n})$ bits;
hence $\R_{1/3}^\mu(f)=O(\sqrt{n})$, even though
$\R_{1/3}(f)=\Omega(n)$. Similar counterexamples can be
constructed in other computational models.

Another possible generalization of Yao's minimax
is to a distribution $\mu$ for which $\R^\mu(f)$
is large even when the both the error of the algorithm
and the expected cost are measured against $\mu$.
That is, in a normal application of Yao's minimax, we either
consider randomized algorithms which only ever make at most $T$
queries (against any input)
and measure their expected error against $\mu$,
or else we consider randomized algorithms which only ever
make error at most $\epsilon$ (against any input) and measure
their expected cost against $\mu$. One may wonder if it is possible
for one distribution to certify the hardness of $f$
in both ways at once, with both the cost and the error
measured in expectation against $\mu$.

The answer turns out to be yes, as first observed by
Vereshchagin for query complexity \cite{Ver98}.
Vereshchagin stated his theorem
for bounded error, but in the case of small bias $\gamma$,
his techniques appear to give a distribution $\mu$ (which
depends on $\gamma$) such that
$\R_{\dot{\gamma}}^\mu(f)=\Omega(\gamma\R_{\dot{\gamma}}(f))$
even where the left-hand side is defined as the \emph{expected}
query complexity against $\mu$ to bias at least $\gamma$
(also against $\mu$). This is in contrast to Yao-style
minimax theorems, which are stronger in that they
lack the $\gamma$ factor on the right hand side, but
weaker in that the left-hand side has either the cost
or the error being worst-case (rather than both being
average-case against $\mu$).

Our results in this work are ``Vereshchagin-like'' in that
they hold even when $\R_{\dot{\gamma}}^\mu(f)$ has both
the cost and the bias defined in expectation against $\mu$.
We prove such results for randomized query complexity
and randomized communication complexity, showing a single
$\mu$ satisfies $\R_{\dot{\gamma}}^\mu(f)=\Omega(\gamma^2\R(f))$
for all $\gamma>0$, even when both the error and the cost
in the definition of $\R_{\dot{\gamma}}^\mu(f)$
are average-case against $\mu$. (For models such as quantum query
complexity or circuit complexity, the expected cost of an
algorithm does not have an obvious interpretation, since
the algorithms generally have the same cost for all inputs;
therefore, for those models we do not give a theorem
in which the cost is measured in expectation against $\mu$.)

Note that our minimax theorem is not directly comparable
to Vereshagin, because we state our lower bounds in
an ``amplified'' form -- that is, the lower bounds
are with respect to $\R(f)$ rather than $\R_{\dot{\gamma}}(f)$.
As previously mentioned, this is necessary when proving
that a single distribution works for all $\gamma$,
and our theorems appear to be tight in that setting.
Moreover, Vereshchagin's theorem is tight in its setting:
the factor of $\gamma$ is necessary, because average-case
query complexity can be smaller than worst-case query complexity
(for example, consider the parity function on $n$ bits,
which has $\R_{\dot{\gamma}}(f)=n$ for all $\gamma$;
if we design a randomized algorithm which queries all the
bits with probability $\gamma$ and queries no bits with
probability $1-\gamma$, it will use only $\gamma n$ expected
queries, and it will solve $f$ to bias $\gamma$).%
\footnote{We thank an anonymous reviewer for this example.}

A remaining open problem is as follows:
can Vereshchagin's theorem be modified to show
\begin{equation}\label{eq:open}
\R_{\dot{\gamma}}^\mu(f)=\Omega(\bar{\R}_{\dot{\gamma}}(f)),
\end{equation}
where both cost and bias on the left are measured in
expectation against $\mu$, and where $\bar{\R}_{\dot{\gamma}}(f)$
denotes the worst-case (over the inputs of $f$) expected 
(over the internal randomness of the algorithm)
query complexity of $f$ to bias $\gamma$?
Note that in the bounded-error setting,
$\bar{\R}(f)=\Theta(\R(f))$, so
for bounded $\gamma$ this result follows from both Vereshchagin's
theorem and from our work here.
For small $\gamma$,
we leave this question as an intriguing open problem.

We also note that we cannot hope that a single
distribution $\mu$ satisfies \eq{open} for all $\gamma$,
because one can construct a counterexample via
a modification of our earlier function:
we let $f$ be defined on $1+m+n$ bits, where
if $x_1=0$ the function evaluates to the parity of the
next $m$ bits, and if $x_1=1$ the function evaluates
to the majority of the last $n$ bits, as before;
this time we will have $n=m^{4/3}$. We also add a promise:
we require that the input always has Hamming weight
either at most $n/2-\sqrt{n}$ or at least $n/2+\sqrt{n}$
on the last $n$ bits, turning the majority part of the function
into a $\sqrt{n}$-gap majority function.
Now, to compute $f$ to worst-case bias $\gamma$
requires at least $\gamma m$ expected queries
on inputs $x$ with $x_1=0$, and requires at least
$\gamma^2 n$ expected queries on inputs with $x_1=1$,
so at least $\Omega(\max\{\gamma m,\gamma^2 n\})$ expected
queries in the worst case. This is $\Omega(n^{1/4})$
when $\gamma=n^{-1/2}$ and $\Omega(n)$ when $\gamma$ is constant.
Now fix a distribution $\mu$, let
$p$ be the probability that $\mu$ assigns to inputs with $x_1=1$.
If $p\le 1/2$, then we can compute $f$ to constant bias
simply by querying the first bit, guessing randomly if $x_1=1$,
and querying $m$ bits to compute $f$ exactly when $x_1=0$;
this uses $O(n^{3/4})$ queries to achieve constant bias, instead
of the $\Omega(n)$ which were required in the worst case.
On the other hand, if $p\ge 1/2$,
then we can compute $f$ against $\mu$ by
querying the first bit and nothing else when $x_1=0$
(guessing the answer randomly),
and otherwise making one additional query
to estimate the gap majority function to bias $1/\sqrt{n}$.
This uses $2$ queries and achieves bias
$n^{-1/2}$ against $\mu$, instead of the $\Omega(n^{1/4})$
queries required in the worst case.

\begin{fulltext}

\section{Minimax theorem for the ratio of saddle functions}
\label{sec:minimax}

Minimax theorems take the form
\[\adjustlimits \inf_{x\in X}\sup_{y\in Y}\alpha(x,y)
=\adjustlimits\sup_{y\in Y}\inf_{x\in X}\alpha(x,y).\]
For any function $\alpha$,
the left-hand side above is always at least the right
hand side, but equality only holds under certain conditions;
when equality does hold, we call it a minimax theorem.

Broadly speaking, the following conditions are required
to ensure that a minimax theorem holds. First,
$X$ and $Y$ must be convex sets (and they must be subsets of some
real vector spaces). Second, $\alpha$ must be
\emph{saddle} -- or at least quasisaddle -- meaning
that it is convex as a function of $x$ and concave as a function
of $y$ (or at least quasiconvex and quasiconcave).
Third, $\alpha$ must satisfy some continuity conditions.
And finally, one of $X$ or $Y$ must be compact (importantly,
it's not necessary for both to be compact).

In this section, we show that under certain conditions,
minimax theorems also hold for \emph{ratios} of
positive saddle functions.
Such a ratio of saddle functions is not necessarily saddle,
but the important insight is that it is still quasisaddle.

\subsection{Background definitions}

In order to formally state the conditions in which
minimax theorems hold, we will need a few definitions.
We assume the reader is familiar with vector spaces
and topological spaces, including standard terminology
such as compact sets and neighborhoods.

\begin{definition}[Real topological vector space]
A \emph{real topological vector space} is a tuple
$(V,+,\cdot,\tau)$, where $V$ is a set, $+$ is a function
$V\times V\to V$, $\cdot$ is a function $V\times\bR\to V$,
and $\tau\subseteq 2^V$, such that
\begin{itemize}
    \item $(V,+,\cdot)$ is a vector space over $\bR$,
    \item $(V,\tau)$ is a topological space,
    \item $+$ is continuous under the topology $\tau$, and
    \item $\cdot$ is continuous under the topology $\tau$ and the
    standard topology of $\bR$.
\end{itemize}
\end{definition}

We note that any normed real vector space is a real topological
space, as the norm induces a topology. We will primarily
focus on the real topological vector spaces $\bR^n$ for $n\in\bN$,
which have a standard topology.

\begin{definition}[Extended reals]
The \emph{extended reals} is the set
$\bar{\bR}\coloneqq\bR\cup\{-\infty,\infty\}$.
We use the extended interval notation
$(r,\infty]\coloneqq (r,\infty)\cup\{\infty\}$ for $r\in\bR$,
and similarly for $[-\infty,r)$ and $[-\infty,\infty]$.
We associate with $\bar{\bR}$ the following topology. A set
$S\subseteq\bar{\bR}$ is a neighborhood of $x\in\bR$
if it contains an open interval $(x-\epsilon,x+\epsilon)$ for some
$\epsilon\in(0,\infty)$, it is a neighborhood of $\infty$ if it contains
the interval $(r,\infty]$ for some $r\in \bR$, and it is a neighborhood
of $-\infty$ if it contains the interval $[-\infty, r)$ for some
$r\in \bR$.

We define addition, subtraction, multiplication, and division
of extended reals in the intuitive way, with $\infty-\infty$,
$0\cdot\infty$, $\infty/\infty$, and $x/0$ for $x\in\bar{\bR}$
all undefined. Note also that the extended reals are ordered
(for each $x,y\in\bar{\bR}$, we have either $x=y$, $x<y$, or $x>y$).
\end{definition}

Note that while we define the extended reals and will often talk
about extended-real-valued functions, our vector spaces
will always be over the reals, not over the extended reals.
In particular, the extended reals are not a field.

\begin{definition}[Convexity of sets]
We say a subset $X$ of a real vector space $V$ is \emph{convex}
if
\[\forall x,y\in X,\;\forall\lambda\in(0,1)\quad
\lambda x+(1-\lambda)y\in X.\]
\end{definition}

\begin{definition}[Convex hull]
Let $V$ be a real vector space and let $X\subseteq V$.
The \emph{convex hull} of $X$, denoted $\Conv(X)$,
is the intersection of all convex subsets of $V$ that
contain $X$ as a subset.
\end{definition}

Note that it is easy to verify that an arbitrary intersection
of convex sets is convex, which means that the convex hull
of any set is always convex.

\begin{definition}
[(quasi)convexity and (quasi)concavity of functions]
Let $V$ be a real vector space, let $X\subseteq V$ be convex,
and let $\phi:X\to\bar{\bR}$. We say that $\phi$ is \emph{convex}
if for all $x,y\in X$ and $\lambda\in(0,1)$, we have
$\phi(\lambda x+(1-\lambda)y)\le \lambda\phi(x)+(1-\lambda)\phi(y)$.
We say $\phi$ is \emph{quasiconvex} if for all $x,y\in X$ and
$\lambda\in(0,1)$, we have
$\phi(\lambda x+(1-\lambda)y)\le\max\{\phi(x),\phi(y)\}$.
We say that $\phi$ is \emph{concave} if $-\phi$ is convex,
and we say $\phi$ is \emph{quasiconcave} if $-\phi$
is quasiconvex. If $\phi$ is both convex and concave,
we say it is \emph{linear}.
\end{definition}

Note that if $\infty$ and $-\infty$ are both in the range
of $\phi$, then $\lambda\phi(x)+(1-\lambda)\phi(y)$
may be $\infty-\infty$, which is undefined; in this case
we say $\phi$ is neither convex nor concave. A function with both
$\infty$ and $-\infty$ in its range may still be quasiconcave or
quasiconvex.

\begin{definition}[Saddle and quasisaddle]
Let $V_1$ and $V_2$ be real vector spaces, let $X\subseteq V_1$
and $Y\subseteq V_2$, and let $\alpha:X\times Y\to\bar{\bR}$.
We say that $\alpha$ is \emph{saddle} if for all $x\in X$
the function $\alpha(x,\cdot)$ is concave and for all $y\in Y$
the function $\alpha(\cdot, y)$ is convex. We say that $\alpha$
is \emph{quasisaddle} if for all $x\in X$ the function
$\alpha(x,\cdot)$ is quasiconcave and for all $y\in Y$ the function
$\alpha(\cdot,y)$ is quasiconvex.
\end{definition}

\begin{definition}[Semicontinuity]
Let $X$ be a topological space and let $\phi:X\to\bar{\bR}$.
We say that $\phi$ is \emph{upper semicontinuous} at $x\in X$
if for all $y\in(\phi(x),\infty]$
there exists some neighborhood $U$ of $x$ on which the value of
$\phi(x')$ for $x'\in U$ is less than $y$. We say that $\phi$
is \emph{lower semicontinuous} at $x$ if $-\phi$ is upper semicontinuous
at $x$.

Let $Y$ be another topological space
and let $\alpha:X\times Y\to\bar{\bR}$ be a function. We say
that $\alpha$ is \emph{semicontinuous} if for all $x\in X$
the function $\alpha(x,\cdot)$ is upper semicontinuous over all of $Y$,
and for all $y\in Y$ the function $\alpha(\cdot,y)$ is lower
semicontinuous over all of $X$.
\end{definition}

We note the following two useful lemmas about upper and lower
semicontinuous functions. These lemmas are standard, but for
completeness we reprove them in \app{minimax}.

\begin{restatable}
[An upper semicontinuous function on a compact set attains its max]
{lemma}{attain}
\label{lem:attain}
Let $X$ be a nonempty compact topological space, and let
$\phi:X\to\bar{\bR}$ be a function. Then if $\phi$ is
upper semicontinuous, it attains its maximum,
meaning there is some $x\in X$ such that for all $x'\in X$,
$\phi(x')\le\phi(x)$. Similarly, if $\phi$ is lower semicontinuous,
it attains its minimum.
\end{restatable}

\begin{restatable}[A pointwise infimum of upper semicontinuous functions
is upper semicontinuous]{lemma}{semicontinuousinf}
\label{lem:semicontinuous_inf}
Let $X$ be a topological space, let $I$ be a set, and let
$\{\phi_i\}_{i\in I}$ be a collection of functions $\phi_i:X\to\bar{\bR}$.
Then if each $\phi_i$ is upper semicontinuous, the function
$\phi(x)=\inf_{i\in I}\phi_i(x)$ is also upper semicontinuous.
Similarly, if each $\phi_i$ is lower semicontinuous,
the pointwise supremum is lower semicontinuous.
\end{restatable}

From these lemmas, it follows that if
$\alpha:X\times Y\to\bar{\bR}$ is semicontinuous, the expressions
\[\adjustlimits\inf_{x\in X}\sup_{y\in Y}\alpha(x,y)\]
\[\adjustlimits\sup_{y\in Y}\inf_{x\in X}\alpha(x,y)\]
have all the infimums attained if $X$ is nonempty and compact,
and all the supremums attained if $Y$ is nonempty and compact.
Hence on compact sets, inf-sup theorems become min-max theorems.

The following lemma will also come in useful. We also prove it
in \app{minimax}.

\begin{restatable}[Quasiconvex functions on convex hulls]
{lemma}{quasiconvexhull}\label{lem:quasiconvex_hull}
Let $V$ be a real vector space, let $X\subseteq V$,
and let $\phi\colon\Conv(X)\to\bar{\bR}$ be a function.
If $\phi$ is quasiconvex, then
\[\sup_{x\in \Conv(X)}\phi(x)=\sup_{x\in X}\phi(x).\]
Similarly, if $\phi$ is quasiconcave, then
\[\inf_{x\in \Conv(X)}\phi(x)=\inf_{x\in X}\phi(x).\]
\end{restatable}

\subsection{Minimax theorems}

We are now ready to state Sion's minimax theorem. Actually,
we will need a version of Sion's minimax for extended-real-valued
functions, while Sion \cite{Sio58} originally only dealt with
real-valued functions; luckily, proving this extension is not hard
given Sion's original theorem, and we do so in \app{minimax}.

\begin{restatable}[Sion's minimax for extended reals]{theorem}
{sionextended}\label{thm:sion_extended}
Let $V_1$ and $V_2$ be real topological vector spaces, and let
$X\subseteq V_1$ and $Y\subseteq V_2$ be convex. Let
$\alpha:X\times Y\to\bar{\bR}$ be semicontinuous and quasisaddle.
If either $X$ or $Y$ is compact, then
\[\adjustlimits\inf_{x\in X}\sup_{y\in Y}\alpha(x,y)
=\adjustlimits\sup_{y\in Y}\inf_{x\in X}\alpha(x,y).\]
\end{restatable}

Next, we use Sion's minimax theorem to show a minimax theorem
for the ratio of positive saddle functions. To do so,
we will need the following lemma.

\begin{lemma}\label{lem:ratio}
Let $a,b,c,d\in(0,\infty)$, and let $\lambda\in(0,1)$. Then
\[\min\left\{\frac{a}{b},\frac{c}{d}\right\}
\le \frac{\lambda a+(1-\lambda)c}{\lambda b+(1-\lambda)d}
\le \max\left\{\frac{a}{b},\frac{c}{d}\right\}.\]
This still holds if any of $a,b,c,d$ are $0$, or if
$a$ or $c$ are $\infty$, so long as we interpret
$x/0=\infty$ for $x\in[0,\infty]$.
\end{lemma}

\begin{proof}
When $a,c\in[0,\infty)$ and  $b,d\in(0,\infty)$, it's easy to check that
\[\frac{\lambda a+(1-\lambda)c}{\lambda b+(1-\lambda)d}
=\frac{a}{b}\cdot\frac{1}{1+z}+\frac{c}{d}\cdot\frac{z}{1+z},\]
where $z=(1-\lambda)d/\lambda b$. Since $z>0$, this
is a convex combination of $a/b$ and $c/d$, from which the desired
result follows. When $a=\infty$ or $c=\infty$, both the middle
expression and the max expression equal $\infty$, and the result
trivially holds. The same thing happens when $b=d=0$. Finally,
when $a,c\in[0,\infty)$ and exactly one of $b$ and $d$ is $0$,
the max expression is again infinity, and the inequality
on the left and side can be easily verified.
\end{proof}

The simple lemma above is enough to imply that a
convex function divided by a concave function is quasiconvex,
and that a concave function divided by a convex function
is quasiconcave.

\begin{lemma}\label{lem:ratio_properties}
Let $V$ be a real topological vector space, and let $X\subseteq V$
be convex. Let $\phi\colon X\to[0,\infty]$ and
$\psi\colon X\to[0,\infty)$ be functions, and define
$\rho\colon X\to[0,\infty]$ by $\rho(x)\coloneqq\phi(x)/\psi(x)$,
with $r/0$ interpreted as $\infty$ for $r\in[0,\infty]$.
Then
\begin{enumerate}
    \item If $\phi$ is convex and $\psi$ is concave, $\rho$ is
    quasiconvex.
    \item If $\phi$ is concave and $\psi$ is convex, $\rho$ is
    quasiconcave.
    \item If $\phi$ is upper semicontinuous and $\psi$ is
    lower semicontinuous, $\rho$ is upper semicontinuous.
    \item If $\phi$ is lower semicontinuous and $\psi$ is upper
    semicontinuous, and if $\phi$ is
    strictly positive on $X$, then $\rho$ is lower semicontinuous.
\end{enumerate}
\end{lemma}

\begin{proof}
We start with (1). Fix $x,y\in X$ and $\lambda\in(0,1)$.
Then
\begin{align*}
\rho(\lambda x+(1-\lambda)y)
&=\frac{\phi(\lambda x+(1-\lambda)y)}{\psi(\lambda x+(1-\lambda)y)}\\
&\le \frac{\lambda\phi(x)+(1-\lambda)\phi(y)}
    {\lambda\psi(x)+(1-\lambda)\psi(y)} \\
&\le \max\left\{\frac{\phi(x)}{\psi(x)},
    \frac{\phi(y)}{\psi(x)}\right\}\\
&=\max\{\rho(x),\rho(y)\},
\end{align*}
so $\rho$ is quasiconvex, as desired.
Here we used the convexity of $\phi$ and concavity of $\psi$
in the first inequality, and \lem{ratio} in the second inequality.
(2) works similarly:
\begin{align*}
\rho(\lambda x+(1-\lambda)y)
&=\frac{\phi(\lambda x+(1-\lambda)y)}{\psi(\lambda x+(1-\lambda)y)}\\
&\ge \frac{\lambda\phi(x)+(1-\lambda)\phi(y)}
    {\lambda\psi(x)+(1-\lambda)\psi(y)} \\
&\ge \min\left\{\frac{\phi(x)}{\psi(x)},
    \frac{\phi(y)}{\psi(x)}\right\}\\
&=\min\{\rho(x),\rho(y)\}.
\end{align*}

Next, we prove (3). Fix $x\in X$; our goal is to show
$\rho$ is upper semicontinuous at $x$. If $\rho(x)=\infty$,
then any function $\rho$ is upper semicontinuous at $x$
by definition, so assume $\rho(x)<\infty$. In particular,
this means that $\phi(x)<\infty$ and that $\psi(x)>0$.
Now, fix $y>\rho(x)=\phi(x)/\psi(x)$. By the upper semicontinuity
of $\phi$, find a neighborhood $U_1$ of $x$ on which
$\phi(\cdot)$ is at most $\phi(x)+\epsilon$ (with $\epsilon>0$
to be chosen later). By the lower semicontinuity of $\psi$,
find a neighborhood $U_2$ of $x$ on which $\psi(\cdot)$
is at least $\psi(x)-\epsilon$. Setting $U\coloneqq U_1\cap U_2$,
we see that on $U$ we have
$\rho(\cdot)\le (\phi(x)+\epsilon)/(\psi(x)-\epsilon)$,
assuming we pick $\epsilon<\psi(x)$. We now simply
pick $\epsilon$ small enough that this expression is less than
$y$, giving us a neighborhood $U$ of $x$ on which
$\rho(\cdot)$ is less than $y$, as desired.

Finally, we prove (4). As before, we fix $x\in X$.
Our goal is to show $\rho(x)$ is lower semicontinuous in at $x$.
Let $y<\rho(x)$. We seek a neighborhood $U$ of $x$ on which
$\rho(\cdot)>y$. To start with, the upper semicontinuity of $\psi$
ensures there is a neighborhood $U_1$ of $x$ on which
$\psi(\cdot)<\psi(x)+\epsilon$, with $\epsilon>0$ arbitrarily small.
Now, if $\phi(x)=\infty$, then $\rho(x)=\infty$.
In this case, the lower semicontinuity
of $\phi$ ensures there is a neighborhood $U_2$ on which
$\phi(\cdot)$ is at least $z$, with $z\in\bR$ is arbitrarily large.
Then in $U_1\cap U_2$, the value of $\rho(\cdot)$ is also
arbitrarily large, and can be made to exceed $y\in\bR$ given
appropriate choices of $z$ and $\epsilon$.
Alternatively, if $\phi(x)<\infty$, then there is a neighborhood
$U_2$ on which $\phi(\cdot)>\phi(x)-\epsilon$. In this case,
on $U_1\cap U_2$ we have
$\rho(\cdot)>(\phi(x)-\epsilon)/(\psi(x)+\epsilon)$.
By picking $\epsilon$ sufficiently small, we can again
get a neighborhood $U_1\cap U_2$ of $x$ on which
$\rho(\cdot)>y$, meaning that $\rho$ is lower semicontinuous.
\end{proof}

We now state the minimax theorem for the ratio of two
positive saddle functions. In the statement below, it may help
to think of $\mathcal{R}$ as a set of randomized algorithms,
and to think
of $\Delta$ as the set of all probability distributions over a finite
input set. Further, think of $\cost(R,\mu)$ as measuring the cost
of the algorithm $R$ when run on $\mu$ (for some models, this will
depend only on $R$ and not on $\mu$), and think of $\score(R,\mu)$
as quantifying the success or bias that the algorithm $R$ achieves
against input distribution $\mu$.

\begin{theorem}
[Minimax theorem for the positive ratio of saddle functions]
\label{thm:minimax_ratio}
Let $V_1$ and $V_2$ be real topological vector spaces.
Let $\mathcal{R}\subseteq V_1$ be convex,
and let $\Delta\subseteq V_2$ be nonempty, convex, and compact.
Let the function $\cost\colon\mathcal{R}\times\Delta\to(0,\infty]$
be semicontinuous and saddle,
and let the function
$\score\colon\mathcal{R}\times\Delta\to[0,\infty)$
be such that its negation, $-\score$,
is semicontinuous and saddle.
Then using $x/0=\infty$ for $x\in(0,\infty]$, we have
\[\adjustlimits\inf_{R\in \mathcal{R}}\max_{\mu\in\Delta}
\frac{\cost(R,\mu)}{\score(R,\mu)}
=\adjustlimits\max_{\mu\in\Delta}\inf_{R\in \mathcal{R}}
\frac{\cost(R,\mu)}{\score(R,\mu)},\]
and the maximums are attained.
\end{theorem}

\begin{proof}
Let $\alpha\colon\mathcal{R}\times\Delta\to(0,\infty]$
be defined by $\alpha(R,\mu)\coloneqq\cost(R,\mu)/\score(R,\mu)$,
with $x/0$ interpreted as $\infty$ for $x\in(0,\infty]$.
For any fixed $\mu\in\Delta$, the function
$\alpha(\cdot,\mu)$ is quasiconvex and lower semicontinuous
by \lem{ratio_properties}.
Similarly, for any fixed $R\in\mathcal{R}$,
the function $\alpha(R,\cdot)$ is concave and upper semicontinuous
by \lem{ratio_properties}.
Hence $\alpha$ is semicontinuous and quasisaddle,
and the desired minimax theorem follows from \thm{sion_extended}.
Furthermore, since $\Delta$ is nonempty and compact, the supremums
are attained as maximums by \lem{semicontinuous_inf} and
\lem{attain}.
\end{proof}

Finally, we will need two extensions of this theorem.
First, we will want to allow the denominator to be
a function of the form $\score(R,\mu)^{+}$, where
the $+$ superscript denotes the maximum of $\score(R,\mu)$
with $0$, and where we only know about saddle properties
of $\score(R,\mu)$, not of $\score(R,\mu)^{+}$.
To do this, we need to show such a maximum with $0$ preserves
the properties we care about. We have the following lemma,
which we prove in \app{minimax}.

\begin{restatable}{lemma}{maxpreserves}\label{lem:positive_preserves}
Let $V$ be a real topological vector space, and let
$X\subseteq V$ be convex. For a function
$\psi\colon X\to\bar{\bR}$, let $\psi^{+}$ denote
the function $\psi^{+}(x)=\max\{\psi(x),0\}$.
Then this operation on $\psi$ preserves convexity,
quasiconvexity, quasiconcavity, upper semicontinuity,
and lower semicontinuity, but not concavity.
\end{restatable}

This lemma is useful, but doesn't quite give us everything we need,
because the operation $\psi^{+}$ does not preserve concavity.
We will need the following additional lemma, which says
that \lem{ratio_properties} also works when dividing
by $\psi^{+}$, despite its lack of concavity.

\begin{lemma}\label{lem:positive_ratio}
Let $V$ be a real topological vector space, and let $X\subseteq V$
be convex. Let $\phi\colon X\to[0,\infty]$ and
$\psi\colon X\to[-\infty,\infty)$ be functions, and define
$\rho\colon X\to[0,\infty]$ by $\rho(x)\coloneqq\phi(x)/\psi(x)^{+}$,
with $r/0$ interpreted as $\infty$ for $r\in[0,\infty]$.
Then if $\phi$ is convex and $\psi$ is concave,
$\rho$ is quasiconvex.
\end{lemma}

\begin{proof}
Fix $x,y\in X$ and $\lambda\in(0,1)$. If $\psi(x)> 0$
and $\psi(y)> 0$, we have
$\rho(\lambda x+(1-\lambda)y)\le\max\{\rho(x),\rho(y)\}$
using the same argument as in \lem{ratio_properties}.
On the other hand, if $\psi(x)\le 0$ or $\psi(y)\le 0$,
then we have $\max\{\rho(x),\rho(y)\}=\infty$, and the inequality
$\rho(\lambda x+(1-\lambda)y)\le\max\{\rho(x),\rho(y)\}$
trivially holds.
\end{proof}

The second extension we will need in our final minimax theorem
is to the case where the
numerator is allowed to be $0$. Unfortunately, as we can see
from the statement of \lem{ratio_properties},
the ratio does not preserve lower semicontinuity in this setting.
We will need to impose some additional conditions on the $\cost$
and $\score$ functions, particularly with regard to their
behavior around $0$.

\begin{definition}
We say that $\cost\colon\mathcal{R}\times\Delta\to[0,\infty]$
and $\score\colon\mathcal{R}\times\Delta\to[-\infty,\infty)$
are \emph{well-behaved} if the following
conditions hold:
\begin{enumerate}
    \item (Finite cost and score can be achieved.)
    For each $\mu\in\Delta$, there is some
    $R\in\mathcal{R}$ such that $\cost(R,\mu)>0$,
    $\cost(R,\mu)<\infty$, and $\score(R,\mu)>0$.
    \item (A zero-cost algorithm has zero cost regardless
    of the input.) For each $R\in\mathcal{R}$,
    either $\cost(R,\mu)=0$ for all $\mu\in\Delta$, or else
    $\cost(R,\mu)>0$ for all $\mu\in\Delta$.
    \item (Mixing a zero-cost algorithm with a nonzero-cost
    algorithm gives a nonzero-cost algorithm.) For each
    $\mu\in\Delta$, if $R,R'\in\mathcal{R}$ are such that
    $\cost(R,\mu)=0$ and $\cost(R',\mu)>0$, then $\cost(\lambda R+(1-\lambda)R',\mu)>0$
    for all $\lambda\in(0,1)$.
\end{enumerate}
\end{definition}

Finally, we are ready for our main workhorse minimax theorem.

\begin{theorem}\label{thm:main_minimax}
Let $V$ be a real topological vector space, and let
$\mathcal{R}\subseteq V$ be convex. Let $S$
be a nonempty finite set, and let $\Delta$ be the
set of all probability distributions over $S$,
viewed as a subset of $\bR^{|S|}$. Let
$\cost\colon\mathcal{R}\times\Delta\to[0,\infty]$
be semicontinuous and saddle,
and let $\score\colon\mathcal{R}\times\Delta\to[-\infty,\infty)$
be such that its negation, $-\score$,
is semicontinuous and saddle.
Suppose $\cost$ and $\score$ are well-behaved. Then
using the convention $r/0=\infty$ for all $r\in[0,\infty]$,
we have
\[\adjustlimits\inf_{R\in\mathcal{R}}\max_{\mu\in\Delta}
\frac{\cost(R,\mu)}{\score(R,\mu)^{+}}
=\adjustlimits\max_{\mu\in\Delta}\inf_{R\in\mathcal{R}}
\frac{\cost(R,\mu)}{\score(R,\mu)^{+}}.\]
Moreover, if $\cost(R,\cdot)$
and $\score(R,\cdot)$ are both linear in $\mu$
for each $R\in\mathcal{R}$, then
\[\adjustlimits\inf_{R\in\mathcal{R}}\max_{x\in S}
\frac{\cost(R,x)}{\score(R,x)^{+}}
=\adjustlimits\max_{\mu\in\Delta}\inf_{R\in\mathcal{R}}
\frac{\cost(R,\mu)}{\score(R,\mu)^{+}}.\]
Further, all of the above maximums are attained.
\end{theorem}

\begin{proof}
First, note that if $S=\{x_1,x_2,\dots,x_{|S|}\}$,
then we can view $\Delta$ as the convex hull of
the set $\{e_1,e_2,\dots,e_{|S|}\}\subseteq\bR^{|S|}$,
where the $e_i$ are the unit vectors
$e_i=(0,0,\dots,0,1,0,0,\dots,0)$ with the $1$ at position $i$.
Hence $\Delta$ is convex. It is also closed and bounded,
making it compact. We identify $e_i$ with $x_i$,
so that $\Delta=\Conv(S)$.

Note that since each $R\in\mathcal{R}$ has either cost
$0$ for all $\mu$ or cost greater than $0$ for all $\mu$,
we can define the set $\mathcal{R}'\subseteq\mathcal{R}$
of $R$ with nonzero cost. Now, on $\mathcal{R}'$,
the function $\alpha(R,\mu)=\cost(R,\mu)/\score(R,\mu)^{+}$
is semicontinuous and quasisaddle by \lem{ratio_properties}
together with \lem{positive_ratio} and \lem{positive_preserves}.
Additionally, $\Delta$ is nonempty, convex, and compact. Thus
by \thm{sion_extended}, we know that
\[\adjustlimits\inf_{R\in\mathcal{R}'}\max_{\mu\in\Delta}
\frac{\cost(R,\mu)}{\score(R,\mu)^{+}}
=\adjustlimits\max_{\mu\in\Delta}\inf_{R\in\mathcal{R}'}
\frac{\cost(R,\mu)}{\score(R,\mu)^{+}},\]
with the maximums attained.

What we want to show is this statement with the infimums over
$\mathcal{R}$ instead of $\mathcal{R}'$. The inf-sup
is always at least the sup-inf for every function, so we need
only show that the sup-inf is at least the inf-sup.
Moreover, since expanding the domain can only decrease the infimum,
we know that
\[\adjustlimits\max_{\mu\in\Delta}\inf_{R\in\mathcal{R}'}
\frac{\cost(R,\mu)}{\score(R,\mu)^{+}}
=\adjustlimits\inf_{R\in\mathcal{R}'}\max_{\mu\in\Delta}
\frac{\cost(R,\mu)}{\score(R,\mu)^{+}}
\ge \adjustlimits\inf_{R\in\mathcal{R}}\max_{\mu\in\Delta}
\frac{\cost(R,\mu)}{\score(R,\mu)^{+}},\]
where the rightmost maximum is attained by virtue
of the fact that we know it is attained when $R\in\mathcal{R}'$,
and if $R\in\mathcal{R}\setminus\mathcal{R}'$,
then $\cost(R,\mu)/\score(R,\mu)^{+}$ is either $0$ or
$\infty$ for all $\mu$. Thus
we only need to show that the max-inf over $\mathcal{R}$
is at least the max-inf over $\mathcal{R}'$,
and that the former maximum is attained.

To see this, let $\mu\in\Delta$ be the maximizing $\mu$ for
the expression
\[\max_{\mu\in\Delta}\inf_{R\in\mathcal{R}'}
\frac{\cost(R,\mu)}{\score(R,\mu)^{+}}.\]
Suppose by contradiction that there was some
$\hat{R}\in\mathcal{R}\setminus\mathcal{R}'$ such that
\[\frac{\cost(\hat{R},\mu)}{\score(\hat{R},\mu)^{+}}
<\inf_{R\in\mathcal{R}'}\frac{\cost(R,\mu)}{\score(R,\mu)^{+}}.\]
Since $\hat{R}\in\mathcal{R}\setminus\mathcal{R}'$,
we must have $\cost(\hat{R},\mu)=0$. Since
$0/\score(\hat{R},\mu)^{+}$ is less than something, and since
we're interpreting $0/0=\infty$, we must have
$\score(\hat{R},\mu)>0$, so that $0/\score(\hat{R},\mu)^{+}=0$.
We wish to show that
$\inf_{R\in\mathcal{R}'}\cost(R,\mu)/\score(R,\mu)^{+}=0$.
To this end, pick $\epsilon>0$. We will find $R\in\mathcal{R}'$
such that $\cost(R,\mu)/\score(R,\mu)^{+}<\epsilon$.
The idea is to pick some $R'\in\mathcal{R}'$ such
that $\cost(R',\mu)<\infty$ and $\score(R',\mu)>0$,
as guaranteed by the well-behaved condition on $\cost$ and $\score$.
Then set $R\coloneqq \lambda R'+(1-\lambda)\hat{R}$,
with $\lambda>0$ extremely small. Now, the well-behaved property
of $\cost$ says that $\cost(R,\mu)>0$, so $R\in\mathcal{R}'$.
By convexity, we also have
$\cost(R,\mu)=\cost(\lambda R'+(1-\lambda)\hat{R},\mu)
\le \lambda\cost(R',\mu)+(1-\lambda)\cost(\hat{R},\mu)
=\lambda\cost(R',\mu)$, and by the concavity of $\score(\cdot,\mu)$,
we have
$\score(R,\mu)=\score(\lambda R'+(1-\lambda)\hat{R},\mu)
\ge \lambda\score(R',\mu)+(1-\lambda)\score(\hat{R},\mu)
\ge (1/2)\score(\hat{R},\mu)$, assuming $\lambda\le 1/2$.

This means that $\score(R,\mu)$ and $\score(\hat{R},\mu)$
are both positive, and
$\cost(R,\mu)/\score(R,\mu)
\le 2\lambda\cost(\hat{R},\mu)/\score(\hat{R},\mu)$.
Since $\cost(\hat{R},\mu)<\infty$,
setting $\lambda>0$ to be small causes
the ratio $\cost(R,\mu)/\score(R,\mu)^{+}$
to be arbitrarily close to $0$, as desired. It follows that
there exists $\mu\in\Delta$ such that
\[\inf_{R\in\mathcal{R}}
\frac{\cost(R,\mu)}{\score(R,\mu)}\ge
\adjustlimits\inf_{R\in\mathcal{R}}\max_{\mu'\in\Delta}
\frac{\cost(R,\mu')}{\score(R,\mu')},\]
and since the inf-max is always at least the max-inf,
there does not exist a $\mu$ for which
the left-hand infimum
is any larger; thus we get the desired result and the
maximum is attained.

Finally, suppose that $\cost(R,\cdot)$ and $\score(R,\cdot)$
are linear for each $R\in\mathcal{R}$. In that case,
$\cost(R,\cdot)$ is convex and $\score(R,\cdot)$ is concave,
which means that $\cost(R,\cdot)/\score(R,\cdot)^{+}$
is quasiconvex on $\Delta$ by \lem{positive_ratio}.
Then \lem{quasiconvex_hull}
implies that the maximum over $\mu\in\Conv(S)$
is attained at a point in $S$.
Moreover, if $R\in\mathcal{R}\setminus\mathcal{R}'$,
then the maximum over $\mu\in\Delta$ evaluates to either
$0$ or $\infty$. If it is $0$, then it is clearly also
attained in $S$. If it is $\infty$, it means some $\mu\in\Delta$
has $\score(R,\mu)\le 0$; the concavity of $\score(R,\cdot)$
then gives us some $x\in S$ such that
$\score(R,x)\le\score(R,\mu)$, meaning there is a point $x\in S$
on which $\score(R,x)^{+}=0$ and $\cost(R,x)/\score(R,x)^{+}=\infty$,
as desired.
\end{proof}

\thm{main_minimax} is the main tool we will use to prove
minimax theorems for algorithmic models. We will usually
apply it in a setting where $\mathcal{R}$ is a set of
algorithms, $S$ is a finite input set, $\Delta$ is a set
of distributions over the inputs, $\cost(R,\mu)$ is a cost
measure for the performance of an algorithm against
a distribution, and $\score(R,\mu)$ is some kind of
success measure. We will sometimes choose
$\score(R,\mu)=\bias_f(R,\mu)$, where $\bias_f(R,\mu)$
is the bias $R$ achieves against distribution $\mu$ in computing $f$.

We will generally combine \thm{main_minimax} with an
amplification theorem; such a theorem will turn the left
hand side $\inf_R\max_x \cost(R,x)/\score(R,x)$ into
something more familiar, such as $\inf_R\max_x\cost(R,x)$
where the infimum is restricted to algorithms $R$ which
achieve at least constant bias (i.e.\ bounded error)
on each input. With such an amplification theorem,
the minimax result will guarantee the existence of a
hard distribution $\mu$ against which
$\cost(R,\mu)/\score(R,\mu)$ is large for all $R$;
this means $\mu$ is hard to solve even to small bias.

While the above strategy works
for models that can be amplified linearly in the bias
(going from bias $\gamma$ to constant bias
using $O(1/\gamma)$ overhead),
such as quantum query complexity, for randomized
algorithms the situation is more complicated.
For randomized algorithms,
we may instinctively want to use something
like $\score(R,\mu)=\bias_f(R,\mu)^2$, but this does not work
as it does not satisfy the right saddle properties.
Instead, we introduce a new way of evaluating the success
of randomized algorithms, called \emph{scoring rules}.
Evaluation via scoring rules ends up being the ``correct''
way to measure the success of a randomized algorithm,
and has more elegant properties than simply the bias.
It is also highly intuitive: to evaluate the success of
an algorithm, we require it to give a confidence
prediction for whether the output is $0$ or $1$,
and then we score the prediction using a scoring
rule which incentivizes honesty (that is, a scoring
rule that causes a Bayesian agent who wishes to maximize
her expected score to output her true subjective probability).

\section{Forecasting algorithms}
\label{sec:forecasting}

In this section we introduce the notion of
\emph{forecasting algorithms}, which output not just
a $\B$ guess at the function value but also a confidence
parameter $q\in[0,1]$ for that prediction. These algorithms
will be scored using a \emph{scoring rule}, which rewards them
$1$ point for a correct prediction made with perfect confidence,
and $0$ points for a confidence of $1/2$.
As we will see, normal algorithms can be converted into
forecasting algorithms and vice versa, and
the expected score of the forecasting version can often
be related to the bias of the algorithm in its regular
(discrete outputs) form.

\subsection{Scoring rules}

\begin{definition}[Scoring rule]
A \emph{scoring rule} is a function
$s:[0,1]\to[-\infty,1]$ such that $s(1)=1$, $s(1/2)=0$,
and $s(\cdot)$ is increasing over $[0,1]$. We say a scoring rule
is \emph{proper} if for each $p\in(0,1)$,
the expression $ps(q)+(1-p)s(1-q)$ is uniquely maximized at $q=p$.
\end{definition}

Generally, if a forecasting algorithm outputs $q\in[0,1]$,
we will interpret it as assigning confidence $q$ to the output $1$
and confidence $1-q$ to the output $0$; we give it score
$s(q)$ if the right answer was $1$, and score $s(1-q)$ if the
right answer was $0$. A proper scoring rule is therefore
a scoring rule that incentivizes the algorithm to output $q=p$
in the case where the right answer is sampled from
$\Bernoulli(p)$. In other words, a proper scoring rule
is one that incentivizes a Bayesian agent to output her true
subjective probability for the outcome being $1$.

\begin{definition}\label{def:rules}
We define the following scoring rules.
\begin{enumerate}
    \item $\hs(q)\coloneqq 1-\sqrt{\frac{1-q}{q}}$
    \item $\Brier(q)\coloneqq 1-4(1-q)^2$
    \item $\bias(q)\coloneqq 1-2(1-q)$
    \item $\ls(q)\coloneqq 1-\log(1/q)$.
\end{enumerate}
\end{definition}

We note that $\Brier(\cdot)$ and $\ls(\cdot)$ are known
as the Brier scoring rule and logarithmic scoring rule, respectively,
and are well-known in the literature. The Brier scoring rule
is useful because it is a proper scoring rule which is bounded
(that is, $s(q)\in[-3,1]$ for all $q\in[0,1]$,
instead of $s(\cdot)$ diverging to $-\infty$ at $0$).
The logarithmic scoring rule has an information-theoretic
interpretation, with the algorithm essentially starting
at score $1$ and losing an amount of score depending on its
``surprisal'' at the correct outcome.

The scoring rule $\bias(\cdot)$ is not proper, but as we will see,
it is closely related to the bias an algorithm will make.
Finally, the scoring rule $\hs(\cdot)$ will be the most
useful of the bunch for our purposes. Despite not having
any intuitive interpretation and not being bounded, it is
an incredibly convenient scoring rule due to the fact
that it can be \emph{amplified}, as we will see.
$\hs(\cdot)$ has been previously studied
(for example in \cite{BSS05}, where it is called
the ``boosting loss'' due to its relationship with
boosting), but we believe its amplification property
has not been previously known (we prove this amplification
property later on in \lem{score_amplification};
this ends up being a key ingredient of our
minimax theorems).

\begin{restatable}{lemma}{rules}
$\hs$, $\Brier$, and $\ls$ are proper scoring rules. $\bias$ is a
scoring rule which is not proper.
\end{restatable}

This lemma can be proven using elementary calculus,
and we do so in \app{calculus}.

\subsection{Distance measures}

Fascinatingly, the above scoring rules
all correspond to well-known distance
measures between probability distributions.
To describe the correspondence, we first start
by defining the following distance measures.

\begin{definition}\label{def:dist_measures}
For probability distributions $\nu_0$ and $\nu_1$
over a finite domain $P$, define
\begin{align*}
\Delta(\nu_0,\nu_1)&\coloneqq\frac{1}{2}\sum_{x\in P}
|\nu_0[x]-\nu_1[x]|&\mbox{(Total variation)}\\
\h^2(\nu_0,\nu_1)&\coloneqq\frac{1}{2}\sum_{x\in P}
(\sqrt{\nu_0[x]}-\sqrt{\nu_1[x]})^2&\mbox{(Hellinger)}\\
\Ess^2(\nu_0,\nu_1)&\coloneqq\frac{1}{2}\sum_{x\in P}
\frac{(\nu_0[x]-\nu_1[x])^2}{\nu_0[x]+\nu_1[x]}&\mbox{(Symmetrized }\chi^2\mbox{)}\\
\JS(\nu_0,\nu_1)&\coloneqq\frac{1}{2}\sum_{x\in P}
\nu_0[x]\log\frac{2\nu_0[x]}{\nu_0[x]+\nu_1[x]}
+\nu_1[x]\log\frac{2\nu_1[x]}{\nu_0[x]+\nu_1[x]}&\mbox{(Jensen-Shannon)}.
\end{align*}
\end{definition}

The above measures give the distance between two probability
distributions. We will sometimes want to have an asymmetric distance
that is weighted towards one of the two distributions;
while these asymmetric distances look strange at first, they
show up naturally in the study of scoring rules. We extend
the above distance measures as follows.

\begin{definition}\label{def:extended_dist_measures}
Given probability distributions $\nu_0$ and $\nu_1$
over a finite domain $P$, as well as a weight $w\in[0,1]$,
set $\nu=(1-w)\nu_0+w\nu_1$.
Let $R$ be the random variable over $x\in P$ defined by
$R(x)\coloneqq |(1-w)\nu_0[x]-w\nu_1[x]|/\nu[x]$ for all $x\in P$.
Then define
\begin{align*}
\Delta(\nu_0,\nu_1,w)&\coloneqq \mathop{\bE}_{x\leftarrow\nu}[R]\\
\h^2(\nu_0,\nu_1,w)&\coloneqq 
\mathop{\bE}_{x\leftarrow\nu}[1-\sqrt{1-R^2}]\\
\Ess^2(\nu_0,\nu_1,w)&\coloneqq \mathop{\bE}_{x\leftarrow\nu}[R^2]\\
\JS(\nu_0,\nu_1,w)&\coloneqq
\mathop{\bE}_{x\leftarrow\nu}\left[1-H\left(\frac{1+R}{2}\right)\right],
\end{align*}
where $H(\alpha)\coloneqq \alpha\log1/\alpha+(1-\alpha)\log 1/(1-\alpha)$
is the binary entropy function.
\end{definition}

It's not hard to see that when $w=1/2$, the expressions
in \defn{extended_dist_measures} equal the ones in \defn{dist_measures}.
Perhaps surprisingly, the distance measures
$\h^2$, $\Ess^2$, and $\JS$ are all related to each other by a constant
factor.

\begin{restatable}[Relations between distance measures]{lemma}{relations}
\label{lem:distance_measure_relations}
When applied to fixed $\nu_0$, $\nu_1$, and $w$,
the distance measures satisfy
\[\frac{\Ess^2}{2}\le 1-\sqrt{1-\Ess^2}
\le \h^2\le \JS\le \Ess^2\]
as well as
\[\Delta^2\le\Ess^2\le\Delta. \]
We also have $\JS\le \h^2/\ln 2$ and $\Ess^2\le (\ln 4)\JS$.
\end{restatable}

While these relationships are certainly known in the
literature, it is hard to chase down good citations
(though see \cite{Top00,MCAL17} for parts of this result);
in any case, we prove \lem{distance_measure_relations}
in \app{calculus}.

\subsection{The highest achievable expected score is a distance
measure}

Consider the following problem: suppose distributions
$\nu_0$ and $\nu_1$ are known (for example, perhaps they
are the distributions of the transcript of a fixed randomized algorithm
when run on a known $0$-distribution and a known $1$-distribution,
respectively). Further,
suppose a $\Bernoulli(w)$ process generates a bit $b\in\B$,
and then a sample $x\leftarrow\mu_b$ is provided.
We assume the parameter $w$ is known. What is the best algorithm
for predicting $b$ given $x$, assuming you wish to maximize
the expected score according to one of the scoring rules
$\hs(\cdot),\Brier(\cdot),\ls(\cdot),\bias(\cdot)$?
It turns out that best attainable expected score
is exactly the distance
between $\nu_0$ and $\nu_1$ according to the distance measures
$\h^2,\Ess^2,\JS,\Delta$, respectively.
To prove this, we introduce the following definitions.

\begin{definition}\label{def:score_subscript}
For a scoring rule $s:[0,1]\to[-\infty,1]$, we define
$s_1(p)\coloneqq s(p)$ and $s_0(p)\coloneqq s(1-p)$.
This way, if a forecasting algorithm outputs $p$
and the real outcome is $b$, the score
of this prediction will be $s_b(p)$.
\end{definition}

\begin{definition}[Expected score notation]
Let $S$ be a finite set, and let $\phi:S\to[0,1]$ be a
function representing predictions.
Let $\nu$ be a distribution over $S$, let
$P(x)$ be a Boolean-valued random variable for each $x\in S$
representing the correct outcome,
and let $s:[0,1]\to[-\infty,1]$ be a scoring rule.
The \emph{expected score} of $\phi$, denoted
$\score_{s}(\phi,\nu,P)$, is defined as
\[\score_{s}(\phi,\nu,P)\coloneqq
\bE_{x\leftarrow\nu}\bE_{b\leftarrow P(x)}[s_{b}(\phi(x))].\]
In these expectations, if a value of $\infty$ or $-\infty$
occurs with probability $0$, we set $0\cdot\infty\coloneqq 0$.
\end{definition}

We can also extend the $\score$ notation to the case where $\phi(x)$
outputs a probability distribution over $[0,1]$ instead of
always outputting a deterministic prediction given the observation
$x$. We won't worry about this case for now.

Equipped with these definitions, we are now ready to prove
the correspondence between scoring rules and distance measures.
This correspondence appears to be known
in the literature (indeed, variants of it seem to have been
rediscovered many times); see \cite{RW11} for an overview.
However, the form we need here
is somewhat different from the usual
form in the literature, which usually discusses
divergences instead of distances.
We therefore include the proof for completeness.

\begin{lemma}\label{lem:scores_distances}
Let $\nu_0$ and $\nu_1$ be probability distributions over a finite set
$S$, and let $w\in[0,1]$.
Let $M_s(\nu_0,\nu_1,w)$ be the maximum possible score of
for predicting $b\leftarrow \Bernoulli(w)$ given $x\leftarrow\nu_b$,
where $\nu_0$, $\nu_1$, and $w$ known. That is, $M_s(\nu_0,\nu_1,w)$
is the maximum over choice of $\phi:S\to[0,1]$ of the expression
$\score_s(\phi,\nu,P)$, where $\nu=(1-w)\nu_0+w\nu_1$ and
$P(x)$ is the posterior probability distribution of $b$
given prior $\Bernoulli(w)$ and observation $x\leftarrow\nu_b$. Then
\begin{align*}
    M_{\bias}(\nu_0,\nu_1,w)&=\Delta(\nu_0,\nu_1,w)\\
    M_{\hs}(\nu_0,\nu_1,w)&=\h^2(\nu_0,\nu_1,w)\\
    M_{\Brier}(\nu_0,\nu_1,w)&=\Ess^2(\nu_0,\nu_1,w)\\
    M_{\ls}(\nu_0,\nu_1,w)&=\JS(\nu_0,\nu_1,w).
\end{align*}
\end{lemma}

\begin{proof}
Consider a fixed $x\in D$.
The contribution of $x$ to the expected score of $\phi$
(with respect to scoring rule $s$) is simply
$(1-w)\nu_0[x]s_0(\phi(x))+w\nu_1[x]s_1(\phi(x))
=(1-w)\nu_0[x]s(1-\phi(x))+w\nu_1[x]s(\phi(x))$.
The total expected score of $\phi$ is therefore the sum
over $x\in D$ of the above expression. The function
$\phi$ which maximizes the expected score is simply the one
where $\phi(x)=q$, where $q$ maximizes the expression
$(1-w)\nu_0[x]s(1-q)+w\nu_1[x]s(q)$. Now, the expression
we wish to maximize has the form $\nu[x]\cdot((1-p)s(1-q)+ps(q))$,
where $p=w\nu_1[x]/\nu[x]$. Hence, if $s$ is proper, the unique
maximum occurs at $q=p=w\nu_1[x]/\nu[x]$. This means that for the
maximizing $\phi$, the contribution of each $x$ to the expected
score is
$(1-w)\nu_0[x]s((1-w)\nu_0[x]/\nu[x])+w\nu_1[x]s(w\nu_1[x]/\nu[0])$,
assuming $s$ is proper.

For $s\in\{\hs,\ls,\Brier\}$, the scoring rule $s$ is indeed proper,
meaning that we have a closed expression for the maximum possible
expected score. Setting $R[x]\coloneqq|w\nu_1[x]-(1-w)\nu_0[x]|/\nu[x]$,
it's not hard to check that for $\hs$, the contribution of each $x$
is $\nu[x](1-\sqrt{1-R[x]^2})$, for $\ls$, the contribution of each $x$
is $\nu[x](1-H((1+R[x])/2))$, and for $\Brier$, the contribution of each
$x$ is $\nu[x]R[x]^2$, as desired.

It remains to deal with $s=\bias$. The contribution of each $x$
is the maximum possible value of
$(1-w)\nu_0[x]\bias(1-q)+w\nu_1[x]\bias(q)$
for $q\in[0,1]$. Since $\bias(q)=2q-1$, it's not hard to see that the
maximizing value of $q$ is $q=0$ when $(1-w)\nu_0[x]>w\nu_1[x]$,
$q=1$ when $w\nu_1[x]>(1-w)\nu_0[x]$, and when $(1-w)\nu_0[x]=w\nu_1[x]$,
the contribution of $x$ to the score is $0$ regardless of the value
of $q$. The contribution of $x$ to the maximum score is therefore
$\nu[x]R[x]$, as desired.
\end{proof}

We note that in the statement of \lem{scores_distances},
we are implicitly assuming that the predictive algorithms
are deterministic: that given $x$, one is only allowed
to output a deterministic prediction $\phi(x)\in[0,1]$
instead of a random choice of prediction. However, it is not
hard to see that randomized algorithms won't help in this setting,
since we are maximizing the expected score, which is a linear
function of the probabilities inside the randomized choice.
That is to say, if the randomized algorithm chooses (on input $x$)
to output $a$ with probability $p$ and $b$ with probability $1-p$,
then the final score of this algorithm will be a linear function
of $p$, and hence the optimal choice of $p$ will be either $0$ or $1$.
Hence \lem{scores_distances} also characterizes the best
possible score of a randomized prediction algorithm with respect
to those four scoring rules.

\subsection{Linear amplification of hs score}

From here on out, we consider only the $\hs(\cdot)$ scoring rule
(and occasionally $\bias(\cdot)$,
which will correspond to the bias
of a randomized algorithm). We will sometimes omit
the subscript in the expression $\score_{s}(\phi,\nu,P)$
when $s=\hs$.

We now proceed to show a few nice properties of the $\hs$
scoring rule. First among them is the amplification property.
We believe this property (which is crucial for our purposes)
has not previously appeared in the literature.

\begin{lemma}[Amplification of $\hs$]\label{lem:score_amplification}
Let $S$ be a finite set, and let
$\phi:S\to[0,1]$ represent a prediction function.
Then for each $k\in\bN$, there is a function
$\phi^{(k)}:S^k\to[0,1]$ such that for any distribution $\nu$ over
$S$, we have
\[\score_{\hs}(\phi^{(k)},\nu^{\otimes k},0)
    \ge 1-(1-\score_{\hs}(\phi,\nu,0))^k\]
\[\score_{\hs}(\phi^{(k)},\nu^{\otimes k},1)
    \ge 1-(1-\score_{\hs}(\phi,\nu,1))^k.\]
Furthermore, equality holds except when
$\score_{\hs}(\phi,\nu,0)=\score_{\hs}(\phi,\nu,1)=-\infty$.
Here $0$ and $1$ are interpreted as the constant functions
$0(x)=0$ and $1(x)=1$.
\end{lemma}

Informally, this lemma is saying the following. Consider a
randomized forecasting algorithm $R$,
which takes input $x$ and outputs a confidence $q\in[0,1]$
representing its belief that $f(x)=1$. Evaluate
this algorithm according to its \emph{worst-case
expected score} with respect to the $\hs(\cdot)$ scoring rule.
That is to say, for each input $x\in f^{-1}(1)$,
consider the expectation $\bE[\hs(R(x))]$ of the
expected score $R$ gets when run on $x$, and for
each $x\in f^{-1}(0)$, consider the analogous expectation
$\bE[\hs(1-R(x))]$. Then take the minimum $\eta$
of all these expected scores, minimizing over any $x\in\Dom(f)$.
This is the worst-case expected score of $R$.
The lemma then says that we can run $R$ on $x$
several times, say $k$ times independently, and
combine the confidence outputs $q_1,q_2,\dots,q_k$
in such a way that the new algorithm has worst-case
expected score equal to $1-(1-\eta)^k$.

\begin{proof}
We define $\phi^{(k)}(x_1\dots x_k)$ as follows. First,
if it holds that some pair $(x_i,x_j)$ in the input satisfies
$\phi(x_i)=0$ and $\phi(x_j)=1$, we define
$\phi^{(k)}(x_1\dots x_k)\coloneqq 1/2$. Otherwise,
we set $\phi^{(k)}(x_1\dots x_k)\coloneqq
\left(1+\prod_{i=1}^k\frac{1-\phi(x_i)}{\phi(x_i)}\right)^{-1}$,
where we interpret $1/0=\infty$ if it occurs (we need not
interpret $\infty\cdot 0$ since that will only occur
if $\phi(x_i)=0$ and $\phi(x_j)=1$ for some $i$ and $j$).
Note that if $\phi(x)=0$ and $\phi(x')=1$ for $x,x'\in S$
that have nonzero weight in $\nu$, then we have
$\score_{\hs}(\phi,\nu,0)=\score_{\hs}(\phi,\nu,1)=-\infty$,
so the desired inequalities trivially hold. Otherwise,
for $b\in\B$ we write
\begin{align*}
\score_{\hs}(\phi^{(k)},\nu^{\otimes k},b)
&=\bE_{x_1\dots x_k\leftarrow\nu^{\otimes k}}\left[
1-\sqrt{\left(\frac{\phi^{(k)}(x_1\dots x_k)}
{1-\phi^{(k)}(x_1\dots x_k)}\right)^{(-1)^b}}\right]\\
&=1-\bE_{x_1\dots x_k\leftarrow\nu^{\otimes k}}\left[
\sqrt{\prod_{i=1}^k\left(\frac{\phi(x_i)}{1-\phi(x_i)}\right)^{(-1)^b}}
\right]\\
&=1-\prod_{i=1}^k\bE_{x_i\leftarrow\nu}\left[
\sqrt{\left(\frac{\phi(x_i)}{1-\phi(x_i)}\right)^{(-1)^b}}\right]\\
&=1-\left(\bE_{x\leftarrow \nu}[1-\hs_b(\phi(x))]\right)^k\\
&=1-(1-\score_{\hs}(\phi,\nu,b))^k.
\end{align*}
Note that equality holds except in the case where
$\score_{\hs}(\phi,\nu,0)=\score_{\hs}(\phi,\nu,1)=-\infty$.
\end{proof}

The following lemma will be convenient when using
this amplification theorem. We prove it in~\app{calculus}.

\begin{restatable}{lemma}{ampinequality}\label{lem:amp_inequality}
If $x\in[0,1]$ and $k\in[1,\infty)$, we have
\[\frac{1}{2}\min\{kx,1\}\le 1-(1-x)^k\le\min\{kx,1\}.\]
\end{restatable}

\subsection{Bias and \texorpdfstring{$\hs$}{hs} score}

Another nice property of $\hs$ is that it is at most $\bias$.

\begin{lemma}\label{lem:hs_bias}
For all $q\in[0,1]$, we have $\hs(q)\le\bias(q)$.
\end{lemma}

\begin{proof}
Recall that $\hs(q)=1-\sqrt{(1-q)/q}$ and $\bias(q)=1-2(1-q)$.
The desired inequality clearly holds at $q=0$ and $q=1$.
For $q\in(0,1)$, it suffices to show that $4(1-q)^2\le (1-q)/q$,
or equivalently $4q(1-q)\le 1\Leftrightarrow 1-4q+4q^2\ge 0
\Leftrightarrow (1-2q)^2\ge 0$, which also clearly holds.
\end{proof}

Finally, the last main property of $\hs$ that we exploit is that $\hs$ scores and biases 
are quadratically related. To explain
what we mean, start with the following definition of a general
algorithm, where we take care not to put any restriction
on the structure of the algorithm but want it to take inputs
and return outputs while incurring some cost.

\begin{definition}
Let $S$ be a finite set, and let $\Delta$ be the set of
probability distributions over $S$.
A \emph{general algorithm}, which
we denote by $R$, is a pair of functions. The first function
is from $\Delta$ to $[0,\infty]$, and we denote it by
$\cost(R,\cdot)\colon\Delta\to\infty$, so that $\cost(R,\mu)$
returns a value in $[0,\infty]$ for $\mu\in\Delta$.
The second function takes inputs from $S$ and returns
a random variable supported on $\B$, and we denote it by
$\outpt(R,\cdot)$, so that $\outpt(R,x)$ is a random variable
on $\B$ for each $x\in S$.

The \emph{bias} of a general algorithm $R$ on input $x\in S$
with respect to function $f\colon S\to\B$ is
$\bias_f(R,x)\coloneqq 1-2\Pr[\outpt(R,x)\ne f(x)]$.
\end{definition}

We note that if $\outpt(R,x)$ has distribution $\Bernoulli(q)$,
then $\bias_f(R,x)=\bias_{f(x)}(q)$, where the function
$\bias_{f(x)}(q)$ is defined according to
\defn{rules} and \defn{score_subscript}.

Just like we defined general algorithms, we also define
forecasting algorithms, which output confidences
in $[0,1]$ instead of values in $\B$.

\begin{definition}
Let $S$ be a finite set and let $\Delta$ be the set of all
probability distributions over $S$.
A \emph{forecasting algorithm}, which we also denote by $R$,
is a pair of functions. The first function
is $\cost(R,\cdot)\colon\Delta\to\infty$, just like a general
algorithm. The second function takes inputs from $S$
and returns a random variable supported on $[0,1]$,
and we denote it by $\pred(R,\cdot)$, so that
$\pred(R,x)$ is a random variable on $[0,1]$ for each $x\in S$.

The \emph{score} of a forecasting algorithm $R$ on input $x\in S$
with respect to function $f\colon S\to\B$ and scoring rule $s$
is $\score_{s,f}(R,x)\coloneqq \bE[s_{f(x)}(\pred(R,x))]$.
When the function $f$ is clear by the context, for notational simplicity 
we often omit it and write $\score_{s}(R,x)$.
Additionally, when $s=\hs$, we sometimes omit it and write simply $\score(R,x)$.
\end{definition}

The following lemma is key. It says that we can convert
any algorithm which achieves bias $\gamma$ into a forecasting
algorithm which achieves expected score at least $\gamma^2/2$
under the $\hs$ scoring rule; further, this conversion
only manipulates the output of the algorithm, meaning it can
be applied without changing the cost. That is, to turn $R$
into a forecasting algorithm, we only need to run $R$, get an
output $0$ or $1$, and then erase the output and write
$(1-\gamma)/2$ or $(1+\gamma)/2$, respectively.

Moreover, it is possible to convert backward as well!
To turn a forecasting algorithm $R$ into a normal randomized
algorithm, run $R$, take the output $q\in[0,1]$, erase it and
write down a sample from $\Bernoulli(q)$ instead.
If the original forecasting algorithm achieved expected score
$\eta$, the new algorithm will achieve bias at least $\eta$.
In particular, this lemma tells us that the best expected score
and the best bias that an algorithm can make (under any cost
restriction) are always quadratically related.

\begin{lemma}[Conversion between regular and forecasting algorithms]
\label{lem:conversion}
A general algorithm $R$ achieving worst-case bias $\gamma>0$
for a function $f$
can be converted into a forecasting algorithm $R'$ with
worst-case score at least $1-\sqrt{1-\gamma^2}\ge\gamma^2/2$
for $f$.
This conversion is pointwise: it depends
only on changing a sample from the random variable $\outpt(R,x)$
after receiving it, as well as on the value of the worst-case bias
$\gamma$.

Conversely, a forecasting algorithm $R$ with worst-case score
$\eta$ can be converted into a general algorithm $R'$ with
worst-case bias at least $\eta$. This conversion is pointwise: it
depends only on changing a sample from $\pred(R,x)$ after receiving
it (and not even on the value of $\eta$).
\end{lemma}

\begin{proof}
Start with a general algorithm $R$ with worst-case bias $\gamma>0$.
On input $x$, run $R$ to receive a sample $b\in\B$ from $\outpt(R,x)$.
Then output $\pred(R',x)=(1-\gamma)/2$ if $b=0$ and output
$\pred(R',x)=(1+\gamma)/2$
if $b=1$. It is clear that this $R'$ was constructed in a pointwise
fashion out of $R$, depending only on a sample from
$\outpt(R,x)$. Now, fix $x\in S$, and let $p\in[0,1]$ be the
probability that $\outpt(R,x)$ gives the right answer. Since
$R$ has worst-case bias $\gamma$, it has bias at least $\gamma$
on $x$, so $p\ge(1+\gamma)/2$. The expected score of $R'$ on $x$
is then
\begin{align*}
\score(R',x)&=p\hs((1+\gamma)/2)+(1-p)\hs((1-\gamma)/2)\\
&=p-p\sqrt{\frac{1-\gamma}{1+\gamma}}
    +(1-p)-(1-p)\sqrt{\frac{1+\gamma}{1-\gamma}}\\
&=1-\sqrt{\frac{1+\gamma}{1-\gamma}}
+p\left(\sqrt{\frac{1+\gamma}{1-\gamma}}
-\sqrt{\frac{1-\gamma}{1+\gamma}}\right)\\
&\ge 1-\sqrt{\frac{1+\gamma}{1-\gamma}}+\frac{1+\gamma}{2}
\left(\sqrt{\frac{1+\gamma}{1-\gamma}}
-\sqrt{\frac{1-\gamma}{1+\gamma}}\right)\\
&=1-\left(1-\frac{1+\gamma}{2}\right)\sqrt{\frac{1+\gamma}{1-\gamma}}
-\frac{1}{2}\sqrt{1-\gamma^2}\\
&=1-\sqrt{1-\gamma^2}.
\end{align*}

For the other direction, let $R$ be a forecasting algorithm
with worst-case score $\eta>0$. On input $x$, run $R$ to receive
a sample $q\in[0,1]$ from $\pred(R,x)$. Then output $1$ with
probability $q$ and $0$ with probability $1-q$,
i.e.\ $\outpt(R',x)\sim\Bernoulli(q)$. It is clear that
this $R'$ is constructed in a pointwise fashion out of $R$
(without even a dependence on $\eta$). Now, fix $x\in S$.
We know that $\eta\le\score(R,x)=\bE[\hs_{f(x)}(\pred(R,x))]$.
Now, we note that $\hs_{f(x)}(p)\le\bias_{f(x)}(p)$
by \lem{hs_bias}. Thus we get
$\eta\le\bE[\bias_{f(x)}(\pred(R,x))]=\bias(R',x)$, as desired.
\end{proof}

To demonstrate the power of these lemmas, observe that they
imply a well-known amplification theorem for randomized algorithms,
as we show in the lemma below.
Note that this lemma does not refer to scoring rules
or forecasting algorithms at all; those only appear as proof
techniques.

\begin{lemma}[informal]
A randomized algorithm with bias $\gamma$ can be amplified
to bias $1/2$ by repeating it $2/\gamma^2$ times.
\end{lemma}

\begin{proof}
Start with an algorithm making bias $\gamma$.
Using \lem{conversion}, get a forecasting
algorithm with expected score at least $1-\sqrt{1-\gamma^2}$.
Using \lem{score_amplification}, repeating the algorithm
$k$ times increases the expected score on each input $x$ to
at least $1-(1-\gamma^2)^{k/2}$. Using \lem{conversion},
we get an algorithm with worst-case bias at least
$1-(1-\gamma^2)^{k/2}$. Using \lem{amp_inequality},
this is at least $\min\{k\gamma^2/4,1/2\}$. Picking
$k\ge 2/\gamma^2$, we get an algorithm
with worst-case bias at least $1/2$ using only $k$
repetitions of the original algorithm, as desired.
\end{proof}

\section{Randomized query and communication complexity}
\label{sec:query}

To prove a strong minimax theorem for randomized query complexity,
we start by formally defining forecasting algorithms
in the query complexity setting. We will need these forecasting
algorithms as a tool, despite our final statement not referring to them.

\begin{definition}
A \emph{deterministic forecasting decision tree}
(on $n\in\bN$ bits, with finite alphabet $\Sigma$)
is a rooted tree on $n$ bits whose internal vertices are labeled by $[n]$, where each internal vertex has $|\Sigma|$ children
labeled by $\Sigma$, and where the leaves are labeled by $[0,1]$.

A \emph{randomized forecasting decision tree}
(on $n\in\bN$ bits, with finite alphabet $\Sigma$)
is a probability distribution over finitely many deterministic
forecasting decision trees.
\end{definition}

We interpret a randomized forecasting decision tree
as a forecasting algorithm in the intuitive way, where
$\cost(R,x)$ is the expected height of $R$ on $x$
(the expected height of the leaf of $x$ in a deterministic
forecasting tree sampled from the distribution $R$),
and where $\pred(R,x)$ is the random variable which samples
from the leaf label when a random deterministic tree from $R$
is run on $x$. Note that since we restrict to distributions
with finite support, we do not need to invoke measure theory
or integrals in interpreting these probabilities and expectations,
even though there are uncountably many deterministic forecasting
decision trees.

We extend $\cost(R,\cdot)$ to the set $\Delta$
of probability distributions over $S$ by writing
$\cost(R,\mu)=\bE_{x\leftarrow \mu}[\cost(R,x)]$,
and similarly for $\score(R,\mu)=\bE_{x\leftarrow\mu}[\score(R,x)]$.
We now show a minimax theorem for the ratio of $\cost$ to
$\score^{+}$ for forecasting randomized algorithms.
This minimax theorem will form the base of our final result:
we will convert the left-hand side to $\R(f)$, and convert
the right hand side to some desirable properties of a
hard distribution $\mu$.

\begin{theorem}\label{thm:randomized_query_minimax}
Let $n\in\bN$, let $\Sigma$ be a finite alphabet, let
$S\subseteq\Sigma^n$, and let $f\colon S\to\B$.
Let $\mathcal{R}$ be the set of all
randomized forecasting decision trees on $n$ bits with alphabet $\Sigma$.
Let $\Delta$ be the set of probability distributions over $S$. Then
\[\adjustlimits\inf_{R\in\mathcal{R}}\max_{x\in S}
\frac{\cost(R,x)}{\score(R,x)^{+}}
=\adjustlimits\max_{\mu\in\Delta}\inf_{R\in\mathcal{R}}
\frac{\cost(R,\mu)}{\score(R,\mu)^{+}},\]
and the maximums are attained.
\end{theorem}

\begin{proof}
We use \thm{main_minimax}. All we need to do is verify
that the conditions of the theorem hold.
Our first task will be to deal with the strange set $\mathcal{R}$;
we wish to turn it into a convex subset of a real topological
vector space. To do so, we define the vector
$v_R\in\bR^{2|S|}$ for each $R\in\mathcal{R}$ by
$v_R[x,1]=\cost(R,x)$ and $v_R[x,2]=\score(R,x)$, and consider
the set $V=\{\,v_R:R\in\mathcal{R}\,\}$. For a vector
$v\in V$, we define $\cost(v,x)=v[x,1]$ and $\score(v,x)=v[x,2]$,
and we extend these definitions to $\cost(v,\mu)$ and
$\score(v,\mu)$ by taking expectations over $\mu$.
Then it is clear that optimizing some function of
$\cost(R,\mu)$ and $\score(R,\mu)$ over $\mathcal{R}$
is the same as optimizing the corresponding function of
$\cost(v,\mu)$ and $\score(v,\mu)$ over $V$. Hence it suffices
to show that
\[\adjustlimits\inf_{v\in V}\max_{x\in S}
\frac{\cost(v,x)}{\score(v,x)^{+}}
=\adjustlimits\max_{\mu\in\Delta}\inf_{v\in V}
\frac{\cost(v,\mu)}{\score(v,\mu)^{+}},\]
with the maximums attained.

To do so, we first note that $V\subseteq\R^{2|S|}$ is convex.
This is because if $v_1,v_2\in V$ and $\lambda\in(0,1)$,
we know there are algorithms $R_1,R_2\in\mathcal{R}$
such that $v_1=v_{R_1}$ and $v_2=v_{R_2}$, and then
the algorithm $\lambda R_1+(1-\lambda)R_2$ (which
mixes the distributions $R_1$ and $R_2$ over deterministic
forecasting decision trees) is a valid member of $\mathcal{R}$.
Then we have
$v_{\lambda R_1+(1-\lambda)R_2}[x,1]
=\cost(\lambda R_1+(1-\lambda)R_2,x)
=\lambda\cost(R_1,x)+(1-\lambda)\cost(R_2,x)
=\lambda v_{R_1}[x,1]+(1-\lambda)v_{R_2}[x,2]$,
and similarly $v_{\lambda R_1+(1-\lambda)R_2}[x,2]
=\lambda v_{R_1}[x,2]+(1-\lambda) v_{R_2}[x,2]$,
so $v_{\lambda R_1+(1-\lambda)R_2}=\lambda v_{R_1}+(1-\lambda)v_{R_2}$.

Next, we note that $\cost(v,\cdot)$ and $\score(v,\cdot)$
are linear functions of $\mu$; this is because they are defined
as expectations over $\mu$. Further, observe
that $\cost(\cdot,\mu)$ and $\score(\cdot,\mu)$ are linear
in $v$. It is also
clear that $\cost(v,\mu)$ and $\score(v,\mu)$
are continuous in both $v$ and $\mu$.

It remains to check that $\cost$ and $\score$ are well-behaved.
First, note that there is always an algorithm which queries
all the bits and outputs the right answer $f(x)$ with perfect
confidence. Such an algorithm $R$ has $\cost(v_R,\mu)=n$
and $\score(v_R,\mu)=1$ for all $\mu$, so finite costs and scores
are attainable. Next, note that if $R$ is such that $\cost(v_R,\mu)=0$
for any $\mu$, then $R$ must make no queries when run on $\mu$.
This means $R$ makes no queries when run on any input, so
$\cost(v_R,\mu')=0$ for all $\mu'\in\Delta$. Finally,
note that $\cost(\cdot,\mu)$ is linear for each $\mu$, so if
$\cost(v,\mu)=0$ and $\cost(v',\mu)>0$, we necessarily have
$\cost(\lambda v+(1-\lambda) v',\mu)>0$ for $\lambda\in(0,1)$.
Hence all the conditions of \thm{main_minimax}
are satisfied, and the desired result follows.
\end{proof}

Our next task is to relate the left-hand side of the
equation in the last theorem to $\R(f)$. 



\begin{theorem}\label{thm:randomized_query_LHS}
Using the notation of \thm{randomized_query_minimax}, we have
\[\adjustlimits\inf_{R\in\mathcal{R}}
\max_{x\in S}\frac{\cost(R,x)}{\score(R,x)^{+}}\ge\frac{\R(f)}{240}.\]
\end{theorem}

To prove this theorem, the idea is to take $R$ from the
left-hand side, amplify
the score of $R$ up to a constant (using the fact that
score amplifies linearly), and then convert the
constant score to constant bias (and hence constant error), getting
an upper bound on $\R(f)$. This is slightly tricky,
because the amount we need to amplify by may depend on the input $x$;
for some $x$, both $\cost(R,x)$ and $\score(R,x)$ may be small,
while for other $x$ they are both large. Unfortunately,
we do not have access to $\score(R,x)$ when we receive input $x$.
Instead, in order to amplify by approximately the correct amount,
we estimate $\cost(R,x)$ (by repeatedly running $R$ on $x$
and observing the number of queries), and we use this
cost estimate to decide the amount of amplification needed.

\begin{proof}
Let $Y^{*}$ be the optimal value
of the left-hand side, and let $R$ be an algorithm such that
$\max_{x\in S}\cost(R,x)/\score(R,x)^{+}=Y$,
where $Y$ is arbitrarily close to $Y^{*}$ (and $Y\ge Y^{*}$).
Then in particular, $\score(R,x)>0$ for all $x\in S$,
and for each $x\in S$ we have $\cost(R,x)/\score(R,x)\le Y$.
Let $R'$ be a modification of $R$ where we cut off each
decision tree in the support of $R$ after $2Y$ queries,
and return $1/2$ in case of a cutoff (ensuring we get a score
of $0$ for that branch).
Note that by Markov's inequality, the probability of encountering
a cutoff branch on input $x$ to $R'$ is at most
$\cost(R,x)/2Y\le Y\score(R,x)/2Y=\score(R,x)/2$.
Since each non-cut-off leaf can contribute at most $1$
to the score (as the maximum of $\hs(\cdot)$ is $1$),
and since the score at a cutoff is $0$, the decrease in score
when going from $R$ to $R'$ is at most the probability of encountering
a cutoff. It follows that
$\score(R',x)\ge \score(R,x)-\score(R,x)/2=\score(R,x)/2$
for all $x\in S$.

Next, we describe a randomized forecasting algorithm $R''$.
The algorithm $R''$ runs $R'$ on $x$
until the number of queries made reaches $10Y$.
Let $L$ be the number of runs of $R'$
on $x$ it takes to reach $10Y$ queries.
Then $R''$ runs $R'$ on $x$ an additional $L$ times,
and uses those new runs to amplify the score, achieving
score $1-(1-\score(R',x))^L$. We wish to prove this score
is at least a constant and that the total number of queries is
only $O(Y)$.

First, we bound the expectation of $L$, the random variable
for the number of runs of $R'$ on $x$ it takes
to reach $10Y$ queries. Let $X_i$ be i.i.d.\ random variables
each representing the number of queries in a single run of
$R'$ on $x$ (so each $X_i$ is supported on $\{0,1,\dots,2Y\}$).
Consider the total number of queries made until the cutoff
is reached; this is $\sum_{i=1}^L X_i$. Let $I_i$
be the Boolean random variable which is $0$ if $L<i$ and $1$ if
$L\ge i$. Then $\sum_{i=1}^L X_i=\sum_{i=1}^\infty X_iI_i$.
Note that the value of $\sum_{i=1}^L X_i$ is always at most
$10Y+2Y$, because after the threshold $10Y$ is reached,
less than one full run of $R'$ on $x$ will happen (using
at most $2Y$ queries). Hence\footnote{The equality
$\bE\left[\sum_{i=1}^L X_i\right]=\bE[X_1]\bE[L]$,
which we rederive here, is known as Wald's equation.}
\begin{align*}
12Y&>\bE\left[\sum_{i=1}^L X_i\right]
=\bE\left[\sum_{i=1}^\infty X_iI_i\right]
=\sum_{i=1}^\infty\bE\left[X_iI_i\right]\\
&=\sum_{i=1}^\infty\Pr[I_i=0]\bE[X_iI_i|I_i=0]
    +\Pr[I_i=1]\bE[X_iI_i|I_i=1]\\
&=\sum_{i=1}^\infty\Pr[L\ge i]\bE[X_i]\\
&=\cost(R',x)\bE[L].
\end{align*}
It follows that $\bE[L]<12Y/\cost(R',x)$.
This means the total expected number of queries $R''$
makes is at most $12Y$ for getting the estimate $L$,
plus $\cost(R',x)\cdot\bE[L]<12Y$ for amplifying the score,
for a total of fewer than $24Y$ expected queries.

To bound the expected score, we start by ensuring $L$
is not too small except with small probability. Note that for a
constant $T$, we have
$\Pr[L\le T]=\Pr[\sum_{i=1}^T X_i\ge bY]$. The sum
$\sum_{i=1}^T X_i$ has expected value $T\cost(R',x)$
and has variance $T$ times the variance of one $X_i$. 
Since $X_i$ is non-negative and bounded above by $2Y$, its variance
is bounded above by $\Var[X_i] \le \bE[X_i^{\,2}] \le 2Y \bE[ X_i ] = 2Y\cost(R',x)$.
Hence, the variance of the sum is at most $2TY\cost(R',x)$. 
We use Chebyshev's inequality, writing
\begin{align*}
\Pr[L\le T]&=\Pr\left[\sum_{i=1}^T X_i\ge 10\R(f)\right]\\
&=\Pr\left[\sum_{i=1}^T X_i-T\cost(R',x)\ge 10Y-T\cost(R',x)\right]\\
&\le \frac{2TY\cost(R',x)}{(10Y-T\cost(R',x))^2},
\end{align*}
which holds assuming $T\le 10Y/\cost(R',x)$. In particular,
if $T=2Y/\cost(R',x)$, then $\Pr[L\le T]\le 1/16$.

Now, note that conditioned on $L=\ell$, the expected
score in the second round of $R''$ is at least
$1-(1-\score(R',x))^\ell$. This is increasing in $\ell$;
hence, conditioned on $L>T$, the expected score of $R''$ on $x$
is greater than $1-(1-\score(R',x))^T$.
Conditioned on $L\le T$, we still have the expected score
be at least $0$, since it is at least $0$ for every fixed $\ell$.
Hence the final expected score of $R''$ on $x$ is greater than
$(1-(1-\score(R',x))^T)(1-\Pr[L\le T])
\ge 1-(1-\score(R',x))^T-\Pr[L\le T]$.
Picking $T= 2Y/\cost(R',x)$, we get
\begin{align*}
\score(R'',x)&>1-(1-\score(R',x))^{2Y/\cost(R',x)}-1/16\\
&\ge \frac{1}{2}\min\left\{1,
    2Y\frac{\score(R',x)}{\cost(R',x)}\right\}
    -1/16\\
&\ge \frac{1}{2}\min\left\{1,
    \frac{\score(R,x)Y}{\cost(R,x)}\right\}
    -1/16\\
&\ge\frac{1}{2}-\frac{1}{16}=\frac{7}{16}.
\end{align*}
This algorithm $R''$ makes fewer than $24Y$ expected queries.
We cut if off after $240Y$ queries, outputting prediction
$1/2$ (getting score $0$) in case of a cutoff; this gives
an algorithm $R'''$ whose worst-case number of queries
is $240Y$, and whose expected score on each $x\in D$ is at least
$7/16-1/10\ge 1/3$. Using \lem{conversion},
we can view $R'''$ as a randomized algorithm computing $f(x)$
with worst-case bias at least $1/3$, and hence worst-case error
at most $1/3$. This means that $\R(f)\le 240Y$. Since we can pick
$Y$ arbitrarily close to $Y^{*}$, we also get that $\R(f)$
is at most the infimum of $240Y$ over feasible choices of $Y$,
which is $240Y^{*}$, and the desired result follows.
\end{proof}

Our next task is to show that the max-inf side of
\thm{randomized_query_minimax} gives us a distribution $\mu$
against which it is hard to tell apart $0$-inputs from $1$-inputs,
in terms of the achievable squared-Hellinger distance between
the distributions of the transcript on the $0$- and $1$-inputs.
The following lemma will come in useful. We prove it in
\app{calculus}.

\begin{restatable}[Hellinger distance of disjoint mixtures]
{lemma}{hmixture}\label{lem:hellinger_mixture}
Let $\mu$ be a distribution over a finite support $A$,
and for each $a\in A$, let $\nu_0^a$ and $\nu_1^a$
be two distributions over a finite support $S_a$.
Let $\nu_0^\mu$ and $\nu_1^\mu$ denote the mixture distributions
where $a\leftarrow \mu$ is sampled, and then a sample is produced
from $\nu_0^a$ or $\nu_1^a$ respectively. Assume the sets $S_a$
are disjoint for all $a\in A$. Then
\[\h^2(\nu_0^\mu,\nu_1^\mu)
=\bE_{a\leftarrow\mu}[\h^2(\nu_0^a,\nu_1^a)].\]
\end{restatable}

\begin{theorem}\label{thm:cost_hellinger}
Let $n\in\bN$, let $\Sigma$ be a finite alphabet,
let $S\subseteq\Sigma^n$, and let $f\colon S\to\B$
be a non-constant function.
Then there exist distributions $\mu_0$ on $f^{-1}(0)$ and
$\mu_1$ on $f^{-1}(1)$ such that for all randomized query
algorithms $R$,
\[\frac{\cost(R,\mu)}{\h^2(\tran(R,\mu_0),\tran(R,\mu_1))}
\ge\frac{\R(f)}{240}.\]
Here $\mu=(\mu_0+\mu_1)/2$, and we interpret $r/0=\infty$
for $r\in[0,\infty)$.
\end{theorem}

\begin{proof}
Using \thm{randomized_query_minimax} and \thm{randomized_query_LHS},
we get a distribution $\mu$ on $S$ such that for all randomized
forecasting algorithms $R$, we have
$\cost(R,\mu)/\score(R,\mu)^{+}\ge \R(f)/240$.
Note that it must be the case that an algorithm $R$ which
makes no queries must have $\score(R,\mu)\le 0$;
this is because we have $\R(f)\ge 1$ (since $f$ is non-constant),
and if there was an algorithm with cost $0$ achieving
positive score, we'd have $\cost(R,\mu)/\score(R,\mu)^{+}=0$,
giving a contradiction. Therefore, it must be the case that $\mu$
places equal weight on $0$ and $1$ inputs, because otherwise
a $0$-cost algorithm could indeed predict $f(x)$ with positive
bias (and hence positive score by \lem{conversion}) against $\mu$.
We set $\mu_0$ to be the conditional distribution of $\mu$
on the $0$-inputs of $f$, and set $\mu_1$ to be the conditional
distribution of $\mu$ on the $1$-inputs of $f$.

Next, we simplify the expression
\[\inf_R \frac{\cost(R,\mu)}{\h^2(\tran(R,\mu_0),\tran(R,\mu_1))}.\]
Note that both the numenator and the denominator
do not depend on the leaf labels, only on the queries of
the randomized decision trees. We can therefore view
the set of all randomized query algorithms $R$ as the convex
hull of the set of all deterministic decision trees with no
leaf labels. Now, note that $\cost(R,\mu)$ and
$\h^2(\tran(R,\mu_0),\tran(R,\mu_1))$ are both linear functions
of (the probability vector of) $R$; for the latter,
this is due to \lem{hellinger_mixture}. 
Then by \lem{ratio}, the ratio is quasiconcave in $R$,
and by \lem{quasiconvex_hull}, the
infimum of this ratio over randomized query algorithms $R$
is equal to the minimum over deterministic query algorithms $A$.
Therefore, it suffices to show that for each deterministic
query algorithm $A$ making a non-zero number of queries,
we have $\cost(A,\mu)/\h^2(\tran(A,\mu_0),\tran(A,\mu_1))\ge\R(f)/240.$

Fix such $A$. We assume its leaves are not labeled.
By \lem{scores_distances}, we can label the leaves of $A$ such that
$\score(A,\mu)=\h^2(\tran(A,\mu_0),\tran(A,\mu_1))$. This
labeling does not affect the cost. Then
\[\frac{\cost(A,\mu)}{\h^2(\tran(A,\mu_0),\tran(A,\mu_1))}
=\frac{\cost(A,\mu)}{\score(A,\mu)^{+}}\ge\frac{\R(f)}{240},\]
as desired.
\end{proof}

Finally, we strengthen this to a lower bound
for the \emph{minimum} of $\cost(R,\mu_0)$ and
$\cost(R,\mu_1)$, instead of for their average $\cost(R,\mu)$.

\begin{theorem}
\label{thm:query_shaltiel_free}
Let $n\in\bN$, let $\Sigma$ be a finite alphabet,
let $S\subseteq\Sigma^n$, and let $f\colon S\to\B$ be a non-constant
function. Then there exist distributions $\mu_0$ on $f^{-1}(0)$
and $\mu_1$ on $f^{-1}(1)$ such that for all randomized query
algorithms $R$,
\[\frac{\min\{\cost(R,\mu_0),\cost(R,\mu_1)\}}
{\h^2(\tran(R,\mu_0),\tran(R,\mu_1))}\ge\frac{\R(f)}{3000},\]
where we interpret $r/0=\infty$ for $r\in[0,\infty)$.
\end{theorem}

\begin{proof}
We use $\mu_0$ and $\mu_1$ from \thm{cost_hellinger}.
Note that
\[\inf_R\frac{\min\{\cost(R,\mu_0),\cost(R,\mu_1)\}}
{\h^2(\tran(R,\mu_0),\tran(R,\mu_1))}
=\inf_{R,b\in\B}\frac{\cost(R,\mu_b)}
{\h^2(\tran(R,\mu_0),\tran(R,\mu_1))}\]
\[=\adjustlimits \min_{b\in\B}\inf_R\frac{\cost(R,\mu_b)}
{\h^2(\tran(R,\mu_0),\tran(R,\mu_1))}.\]
By the same argument as in the proof of \thm{cost_hellinger},
this last infimum over $R$ is equal to the infimum over
deterministic unlabeled decision trees $D$ with height at least $1$.

Let $D$ be such an algorithm. By \thm{cost_hellinger},
it suffices to show that
\[\frac{\min\{\cost(D,\mu_0),\cost(D,\mu_1)\}}
{\h^2(\tran(D,\mu_0),\tran(D,\mu_1))}
\ge(1/c)\min_{D'}\frac{\cost(D',\mu)}
{\h^2(\tran(D',\mu_0),\tran(D',\mu_1))},\]
where $\mu=(\mu_0+\mu_1)/2$. By \lem{scores_distances},
we can label the leaves of $D$ so that we have
the property
$\h^2(\tran(D,\mu_0),\tran(D,\mu_1))=\score(D,\mu)$,
and similarly for $D'$.
The desired inequality is trivial when $\score(D,\mu)=0$
(since the ratio is then $\infty$), so suppose $\score(D,\mu)>0$.
We wish to show
\[\frac{\min\{\cost(D,\mu_0),\cost(D,\mu_1)\}}
{\score(D,\mu)}\ge (1/c)\min_{D'}\frac{\cost(D',\mu)}{\score(D',\mu)}.\]
In other words, we wish to show that there exists a deterministic
forecasting algorithm $D'$ such that
$\cost(D',\mu)/\score(D',\mu)\le c\cost(D,\mu_b)/\score(D,\mu)$,
regardless of whether $b=0$ or $b=1$.

We construct $D'$ such that
$\cost(D',\mu)/\score(D',\mu)\le c\cost(D,\mu_b)/\score(D,\mu)$.
The idea is to start with $D$, and then cut off
the branches that are much more likely under $\mu_{1-b}$
than under $\mu_b$. That is, for a vertex $v$ of $D$,
let $\mu_0[v]$ denote the probability that $v$ is reached
when $D$ is run on an input from $\mu_b$, and define $\mu_{1-b}[v]$
similarly. Recall that the leaves of $D$ are labeled according
to the strategy that achieves
$\score(D,\mu)=\h^2(\tran(D,\mu_0),\tran(D,\mu_1))$,
which, by \lem{scores_distances}, is such that at a leaf $v$,
the algorithm $D$ outputs $\mu_1[v]/2\mu[v]$.

Pick a constant $a\in(1/2,1)$, and let $D'$ be the algorithm
which cuts off $D$ the first time it enters a vertex for which
$\mu_{1-b}[v]/2\mu[v]\ge a$, and outputs $a$ (if $b=0$) or $1-a$
(if $b=1$) instead of continuing to run $D$.
Let $V$ be the set of all vertices which cause such a cutoff;
note that no vertex in $V$ is a descendant of another vertex in $V$.
For $v\in V$, let $\mu^v$ be the distribution $\mu$ conditioned
on reaching $v$, and similarly define $\mu_0^v$ and $\mu_1^v$.
Let $\mu^{*}$ be the distribution $\mu$ conditioned on reaching
none of the vertices in $V$, and similarly define $\mu_0^{*}$
and $\mu_1^{*}$.
Since we are dealing with a deterministic decision tree,
all the distributions $\mu_0^v$ and $\mu_1^v$ have disjoint
supports for all the different $v\in V$, and they're also
disjoint from $\mu_0^{*}$ and $\mu_1^{*}$; indeed, $\mu$ is a
disjoint mixture of all different distributions.
It follows that $\score(D,\mu)$ is a mixture of terms
$\score(D,\mu^v)$ and of $\score(D,\mu^{*})$.
The score $\score(D',\mu)$ of the algorithm $D'$ is also such a mixture.

Now, note that $\score(D,\mu^v)\le 1$, and that
$\score(D',\mu^v)=\bE_{x\leftarrow \mu^v}[\hs_{f(x)}(a)]$
if $b=0$ and $\score(D',\mu^v)=\bE_{x\leftarrow \mu^v}[\hs_{f(x)}(1-a)]$
if $b=1$. This means
\[\score(D',\mu^v)=\frac{\mu_b[v]}{2\mu[v]}\hs(1-a)
+\frac{\mu_{1-b}[v]}{2\mu[v]}\hs(a)=(1-p)\hs(1-a)+p\hs(a)\]
\[=1-(1-p)\sqrt{a/(1-a)}-p\sqrt{(1-a)/a},\]
where $p=\mu_{1-b}[v]/2\mu[v]\ge a$. Since $a> 1/2$, this is
increasing in $p$, so we have $\score(D',\mu^v)\ge 1-2\sqrt{a(1-a)}$,
and hence $\score(D',\mu^v)\ge (1-2\sqrt{a(1-a)})\score(D,\mu^v)$.
It also holds that $\score(D',\mu^{*})=\score(D,\mu^{*})
\ge(1-2\sqrt{a(1-a)})\score(D,\mu^{*})$.
Since $\score(D,\mu)$ and $\score(D',\mu)$ are matching mixtures of
$\score(D,\mu^v)$ and $\score(D',\mu^v)$ respectively, it follows that
$\score(D',\mu)\ge(1-2\sqrt{a(1-a)}\score(D,\mu)$.

We now analyze the cost of $D'$. Note that
$\cost(D',\mu)=(1/2)\cost(D',\mu_b)+(1/2)\cost(D',\mu_{1-b})$;
we clearly have $\cost(D',\mu_b)\le \cost(D,\mu_b)$, so it suffices
to upper bound $\cost(D',\mu_{1-b})$. This
is the expected height of a leaf $D'$ reaches when run on $\mu_{1-b}$,
which is a mixture of $\cost(D',\mu^{*}_{1-b})$ and
$\cost(D',\mu^v_{1-b})$. Now, note that a leaf $u$ reached by
$\cost(D',\mu^{*}_{1-b})$ must have $\mu_{1-b}[u]/2\mu[u]< a$,
or $\mu_b[u]<(1-a)/a\cdot\mu_{1-b}[u]$. It follows that
\[\cost(D',\mu^{*}_{1-b})\le (1-a)/a\cdot \cost(D',\mu^{*}_b)
=(1-a)/a\cdot \cost(D,\mu^{*}_b).\]
Similarly, for each $v\in V$, the parent $u$ of $v$
satisfies $\mu_{1-b}[u]/2\mu[u]<a$, meaning that
$\mu_b[u]>(1-a)/a\cdot \mu_{1-b}[u]$; note that since this parent
$u$ of $v$ is not a leaf, conditioned on reaching $u$ the height
of the path will always be at least the height of $v$ (one more
than the height of $u$); since $\cost(B,\mu^v_{1-b})$ is exactly
the height of $v$, we necessarily have
\[\cost(D,\mu^v_b)\ge\cost(D',\mu^v_b)
\ge (1-a)/a\cdot\cost(D',\mu^v_{1-b}).\]
We conclude that $\cost(D',\mu_{1-b})\le \frac{a}{1-a}\cost(D,\mu_b)$,
and hence
\[\cost(D',\mu)
\le\left(\frac{1}{2}+\frac{a}{2(1-a)}\right)\cost(D,\mu_b)
=\frac{\cost(D,\mu_b)}{2(1-a)}.\]
We therefore have
\[\frac{\cost(D',\mu)}{\score(D',\mu)}
\le\frac{1}{2(1-a)(1-2\sqrt{a(1-a)})}\frac{\cost(D,\mu_b)}{\score(D,\mu)}.\]
Finally, optimizing $a$, we pick $a=(2+\sqrt{2})/4$
to get
\[\frac{\cost(D',\mu)}{\score(D',\mu)}
\le(6+4\sqrt{2})\frac{\cost(D,\mu_b)}{\score(D,\mu)},\]
from which the desired result follows.
\end{proof}

\begin{corollary}
\label{cor:query-complexity}
Let $n\in\bN$, let $\Sigma$ be a finite alphabet,
let $S\subseteq\Sigma^n$, and let $f\colon S\to\B$ be a function.
Then there exists a distribution $\mu$ on $S$ such that
for all $\gamma\in[0,1]$,
\[\bar{\R}^\mu_{\dot{\gamma}}(f)\ge\frac{\gamma^2\R(f)}{500}.\]
Here $\dot{\gamma}=(1-\gamma)/2$ and $\bar{\R}^\mu_\epsilon(f)$
denotes the average cost (against $\mu$) of a randomized algorithm
achieving error at most $\epsilon$ (against $\mu$) for solving $f$.
\end{corollary}

\begin{proof}
If $f$ is constant, then $\R(f)=0$ and the desired bound trivially
follows. Therefore, assume $f$ is not constant.
We use the distribution $\mu$ from \thm{cost_hellinger}.
Let $R$ be a randomized algorithm which achieves bias
$\gamma$ against $\mu$. Then using \lem{conversion},
we can convert $R$ into a forecasting algorithm $R'$
which achieves expected score $1-\sqrt{1-\gamma^2}\ge\gamma^2/2$
against $\mu$,
and has the same distribution over query trees (that is,
only the leaves changed).
Now, by the property of $\mu$, we know that
\[\frac{\cost(R',\mu)}{\score(R',\mu)}\ge\frac{\R(f)}{240},\]
where we used \lem{scores_distances} to get a result for score
instead of Hellinger distance in the denominator, and
where we used the fact that $R$ achieves
non-zero bias against $\mu$ (despite $\mu$ being balanced
between $0$- and $1$-inputs) to conclude that $R$ does not make
$0$ queries. Using $\score(R',\mu)\ge\gamma^2/2$ and
$\cost(R,\mu)=\cost(R',\mu)$, we get
$2\cost(R,\mu)/\gamma^2\ge \R(f)/240$,
or $\cost(R,\mu)\ge \gamma^2\R(f)/480$, as desired.
\end{proof}

\subsection{Communication complexity}

\newtheorem*{cchthm}{\thm{cc-h}}
\begin{cchthm}[Restated]
For any non-constant partial function $F\colon \mathcal{X} \times \mathcal{Y} \to \B$ over finite sets $\mathcal{X}$ and $\mathcal{Y}$,
there is a pair of distributions $\mu_0$ on $F^{-1}(0)$ and $\mu_1$ on $F^{-1}(1)$
such that for any public-randomness communication protocol $\Pi$, the squared Hellinger distance between 
the distribution of its transcripts on $\mu_0$ and $\mu_1$ is bounded above by
\[
\h^2\big(\tran(\Pi,\mu_0),\tran(\Pi,\mu_1)\big) = O\left( \frac{\min\{\cost(\Pi,\mu_0),\cost(\Pi,\mu_1)\}}{\RCC(F)} \right).
\]
\end{cchthm}

\begin{proof}
This theorem follows directly from \thm{query_shaltiel_free}
once we realize that a communication function can be interpreted
as a query function. That is, we take $F$ and convert it into
a query function $f$ as follows. The input to $f$
will contain one bit for each possible function of $\mathcal{X}$
(that Alice might send to Bob), and one bit for each possible
function of $\mathcal{Y}$ (that Bob might send to Alice), for a total
input length of $n=2^{|\mathcal{X}|}+2^{|\mathcal{Y}|}$.
The inputs to $f$ will be the strings in $\B^n$ which are
generated by a pair $(x,y)\in S$, that is, the strings $z\in\B^n$
for which there exists a pair $(x,y)\in S$ such that $z_k$
is the result of applying the $k$-th possible function to $x$
(if $k\le 2^{|\mathcal{X}|}$) or the $(k-2^{|\mathcal{X}|})$-th
possible function to $y$ (if $k> 2^{|\mathcal{X}|}$).
Then $f$ is a Boolean function of domain of size $|S|$,
with each string in its domain corresponding to a string in $S$.

We note that $\RDT(f)=\RCC(F)$. This is clear from the definition
of $\RCC(F)$: the public-coin randomness essentially means that
Alice and Bob agree on a randomized decision tree in advance, including
on who speaks when (as a function of the transcript), which is equivalent
to agreeing in a decision tree for $f$ in advance. The transcript
of $f$ on an input is precisely the transcript of $F$ on the corresponding
input, with the catch that in query complexity we defined the transcript
to include the deterministic decision tree by the protocol; hence,
the query version of a transcript of $f$ actually corresponds to $(R,\Pi)$
for $F$, where $R$ is the public randomness and $\Pi$ is the usual
communication complexity transcript. The desired result then
follows immediately from applying \thm{query_shaltiel_free} to $f$.
\end{proof}

\begin{corollary}
\label{cor:communication}
Let $\mathcal{X}$ and $\mathcal{Y}$ be finite sets, let
$S\subseteq\mathcal{X}\times\mathcal{Y}$, and let $F\colon S\to\B$
be a function. Then there exists a distribution
$\mu$ on $S$ such that for all $\gamma\in[0,1]$,
\[
\bar{\RCC}_{\dot{\gamma}}^{\mu}(F)=\Omega(\gamma^2\RCC(F)).
\]
\end{corollary}

\section{Quantum query and communication complexity}
\label{sec:quantum}

In contrast to the classical case,
it is well-known that quantum algorithms can be amplified
\emph{linearly} in $1/\gamma$, where $\gamma$ is the bias.
Formally, we have the following theorem.

\begin{theorem}[Amplitude estimation]\label{thm:amplitude_estimation}
Suppose we have access to a unitary $U$ (representing a quantum
algorithm) which maps $|0\rangle$ to $|\psi\rangle$,
as well as access to a projective measurement $\Pi$,
and we wish to estimate $p\coloneqq\|\Pi|\psi\rangle\|_2^2$
(representing the probability the quantum algorithm accepts).
Fix $\epsilon,\delta\in(0,1/2)$.
Then using at most $(100/\epsilon)\cdot\ln(1/\delta)$
controlled applications of $U$ or $U^\dagger$
and at most that many applications of $I-2\Pi$, we
can output $\tilde{p}\in[0,1]$ such that
$|\tilde{p}-p|\le\epsilon$ with probability
at least $1-\delta$.
\end{theorem}

This theorem follows from \cite{BHMT02}, as well as from the
arguably simpler techniques in \cite{AR20}. (In fact,
these authors show something slightly stronger: amplitude
estimation can be done with overhead
$O(\sqrt{\epsilon+p}\cdot(1/\epsilon)\cdot\log1/\delta)$.
We refer the interested reader to \app{amplitude} to see
how this follows from \cite{BHMT02}.)

Given that quantum algorithms can be amplified linearly
in the bias, it would seem that the desired minimax
theorem follows easily from \thm{main_minimax}:
simply apply a minimax to $\cost(Q,\mu)/\bias_f(Q,\mu)^{+}$,
where $Q$ is a quantum algorithm and $\mu$ is a distribution
over the inputs. Then use the linear amplification
result to argue that
$\min_Q\max_\mu \cost(Q,\mu)/\bias_f(Q,\mu)^{+}$
is $\Theta(\Q(f))$. Sounds simple! (This works
better than for randomized algorithms, because
$\bias_f(\cdot,\cdot)$ is saddle while $\bias_f(\cdot,\cdot)^2$
is not.)

Unfortunately, there is an annoying hole in this argument:
the function $\cost(Q,\mu)$ is not convex in $Q$.
While it is not immediately clear what a
convex combination of two quantum algorithms $Q_1$ and $Q_2$
should be, most intuitive definitions
will have the convex combination use a number of unitaries
that is equal to the \emph{maximum} of the number used in $Q_1$
and $Q_2$, rather than the average.

To get around this, we switch the computational model from
quantum algorithms to probability distributions over quantum
algorithms. These probabilistic quantum algorithms have
outputs and biases defined in the intuitive way, but their
cost is defined as the \emph{expected} cost of the underlying
quantum algorithms, rather than the maximum cost.
This ensures the function $\cost(\cdot,\cdot)$ will be saddle,
and \thm{main_minimax} can be applied. The trick
then becomes showing that
these probabilistic quantum algorithms can still be amplified
linearly. This turns out to be true, up to logarithmic factors.
Once amplified, constant-error probabilistic quantum algorithms
can be converted into ordinary quantum algorithms,
giving us a minimax theorem that can be applied to ordinary
quantum algorithms as well.

\subsection{Quantum query complexity}

Our goal in this section will be to prove the following theorem.

\begin{theorem}\label{thm:minimax_quantum}
For any Boolean-valued function $f$, there exists a distribution
$\mu$ over $\Dom(f)$ such that for any $\gamma\in[0,1]$,
we have $\Q^\mu_{\dot{\gamma}}(f)\ge \gamma\cdot\tOmega(\Q(f))$.
Here $\Q^\mu_{\dot{\gamma}}(f)$ denotes the minimum number of queries required by a quantum algorithm which achieves bias $\gamma$ against $\mu$ for computing $f$. The constants in the $\tOmega$ notation are universal.
\end{theorem}

In fact, we will prove a stronger (and tighter) version
in terms of \emph{probabilistic} quantum algorithms.
These are simply probability distributions over quantum
algorithms of possibly different query costs; we define
the cost of a probabilistic quantum algorithm as the expected
cost of a quantum algorithm sampled from the probability
distribution.

\begin{definition}
A \emph{probabilistic quantum algorithm} is a probability
distribution $P$ over quantum algorithms.
For an input string $x$,
we let $P(x)$ be the random variable that outputs a sample
from $Q(x)$ where $Q$ is a quantum algorithm sampled from $P$.
The cost of $P$, denoted $|P|$, is the expected cost
of a quantum algorithm sampled from $P$. The error
of $P$ on input $x$ to a Boolean function $f$ is defined as
$\Pr_{Q\sim P}[Q(x)\ne f(x)]$.
\end{definition}

\begin{definition}
Let $f$ be a Boolean-valued function with
$\Dom(f)\subseteq\Sigma^n$. We define
$\PQ_{\dot{\gamma}}(f)$ to be the minimum cost $|P|$
of a probabilistic quantum algorithm $P$ which computes
$f$ to worst-case bias $\gamma$.
\end{definition}

\begin{theorem}\label{thm:PQ_Q}
For any Boolean function $f$ and any $\gamma\in(0,1/3)$,
we have $\PQ_{\dot{\gamma}}(f)=\tTheta(\gamma\Q(f))$.
More explicitly,
\begin{align*}
\PQ_{\dot{\gamma}}(f)&=O(\gamma\Q(f))\\
\PQ_{\dot{\gamma}}(f)&=\Omega\left(
    \frac{\gamma\Q(f)}{\log(1/\gamma)\log\log(1/\gamma)}\right).
\end{align*}
\end{theorem}

\begin{proof}
For the upper bound,
let $Q$ be a quantum algorithm computing $f$ to
error $1/3$ using $\Q(f)$ queries.
Let $Q'$ be the probabilistic quantum
algorithm which runs $\Q(f)$ with probability $3\gamma$
and otherwise uses no queries and guesses the output at random
(with probability $1/2$ for outputting both $0$ and $1$).
The probability of error of $Q'$ is at most
$(1/2)(1-3\gamma)+(1/3)(3\gamma)=(1/2)(1-\gamma)$,
which means its bias is at least $\gamma$
on every input. The expected
number of queries $Q'$ uses is $3\gamma\Q(f)$. Hence we
have $\PQ_{\dot{\gamma}}(f)\le 3\gamma\Q(f)$.

For the lower bound, we start with a probabilistic
quantum algorithm $P$ which achieves worst-case bias
$\gamma$ and has cost $|p|=\PQ_{\dot{\gamma}}(f)$,
and make several modifications to it.
First, we remove from the support of $P$ all quantum
algorithms which use more than $2|P|/\gamma$
queries, and we replace them with a $0$-query
quantum algorithm that guesses the answer at random
(with $1/2$ probability on outputs $0$ and $1$).
This gives us a probabilistic quantum algorithm $P_1$
which uses at most $2|P|/\gamma$ queries even
in the worst case, and has $|P_1|\le|P|$ and the worst-case
bias of $P_1$ is at least $\gamma/2$
(since by Markov's inequality, the probability
mass over the removed quantum algorithms was at most
$\gamma/2$, and they could have had bias at most $1$
which turned into bias $0$, decreasing the overall bias
by at most $\gamma/2$).

Next, we modify $P_1$ to get a probabilistic algorithm
$P_2$ which always uses a number of queries which
is a power of $2$. This can be done simply by
increasing the number of queries each algorithm in the support
of $P_1$ makes (and ignoring the extra queries).
This way, we have $|P_2|\le 2|P_1|\le 2|P|$, the
largest number of queries $P_2$ can make is at most
$4|P|/\gamma$, and the bias of $P_2$ is at least $\gamma/2$
on every input.

Further, we modify $P_2$ to get a probabilistic quantum
algorithm $P_3$ which always uses at least $8|P|$
queries (but still only uses a number of queries
which is a power of $2$). This can be done
by again increasing the number of queries a quantum
algorithm in the support of $P_2$ makes, when necessary.
This adds at most an additive $16|P|$ queries
(since the smallest power of $2$ which is at least $8|P|$
is smaller than $16|P|$). Hence $|P_3|<|P_2|+16|P|\le 18|P|$.
Note that $P_3$ achieves bias at least $\gamma/2$
on every input, and that $P_3$ always uses
a number of queries which is a power of $2$ in the range
$[8|P|,4|P|/\gamma)$.

Finally, we modify $P_3$ to get $P_4$ which collapses
together all quantum algorithms in the support of $P_3$
that use the same number of queries. That is,
instead of placing support on two different quantum
algorithms which both use (say) $32$ queries, $P_4$
will place support on a single quantum algorithm which
implements the mixture of both. This does not affect
the number of queries or the bias of the algorithm.
Hence we have $|P_4|<18|P|$, and $P_4$ achieves bias
at least $\gamma/2$ on each input. Further,
$P_4$ has support on fewer than $\log(1/\gamma)$
quantum algorithms.

Next we introduce some notation for talking about $P_4$.
Let $L=\lfloor\log(1/\gamma)\rfloor$
and let $2^k$ be the smallest power of $2$ which is at least
$4|P|$. Let the quantum algorithms in the support of $P_4$
be $Q_1,Q_2,\dots,Q_L$, with $Q_i$ using $2^{k+i}$
queries for each $i$. Let $p_i$ be the probability
$P_4$ assigns to algorithm $Q_i$. Then $p_i\ge 0$
for all $i$, and $\sum_{i=1}^L p_i=1$. We also have
$\sum_{i=1}^L p_i2^{k+i}=|P_4|<18|P|$,
which means $\sum_{i=1}^L p_i 2^i<5$.
On input $x$, let $\alpha_i(x)$ be the probability
that $Q_i$ outputs $1$ when run on $x$, and let
$\beta_i(x)\coloneqq 1-2\alpha_i(x)$. This way,
$(-1)^{f(x)}\beta_i(x)$ is the bias of $Q_i$ when run on $x$.
Then $\sum_{i=1}^L p_i\beta_i(x)$ is $(-1)^{f(x)}$ times the bias
of $P_4$ on $x$, which means that it is negative if $f(x)=1$,
positive if $f(x)=0$, and satisfies
$\left|\sum_{i=1}^L p_i\beta_i(x)\right|\ge\gamma/2$.

We now wish to amplify $P_4$ from bias $\gamma/2$
to constant bias. To do so, it suffices to estimate
$\sum_{i=1}^L p_i\beta_i(x)$ to additive error less than $\gamma/2$,
and output the sign of this estimate. Our query budget
for this task will be roughly $|P|/\gamma$.
We know the values $p_i$, and seek to generate estimates
$\tilde{\beta}_i(x)$ for $\beta_i(x)$. We will say
an estimate $\tilde{\beta}_i(x)$ is \emph{good} if
$|\tilde{\beta}_i(x)-\beta_i(x)|\le 2^i\gamma/10$.
This way, if all $\tilde{\beta}_i(x)$ are good, our final estimate
for the sum will satisfy
\[\left|\sum_{i=1}^L p_i\tilde{\beta}_i(x)
    -\sum_{i=1}^L p_i\beta_i(x)\right|
=\left|\sum_{i=1}^L p_i(\tilde{\beta}_i(x)-\beta_i(x))\right|
\le \sum_{i=1}^L p_i|\tilde{\beta}_i(x)-\beta_i(x)|
\le \sum_{i=1}^L p_i 2^i\gamma/10
<\gamma/2,\]
where we used $\sum_i p_i 2^i<5$.

To generate $\tilde{\beta}_i(x)$, we use \thm{amplitude_estimation}
on algorithm $Q_i$ with $\epsilon=2^i\gamma/20$ and
$\delta=1/3L$. Since the query cost of $Q_i$
is $2^{k+i}$, this uses at most
$2000\cdot (2^k/\gamma)\cdot \ln(3L)$ queries.
Since $2^k<8|P|$ and $L\le\log (1/\gamma)$, this
costs $O(|P|/\gamma\cdot\log\log(1/\gamma))$.
The query cost of generating all $L$ estimates this way is
therefore $O(|P|/\gamma\cdot\log(1/\gamma)\log\log(1/\gamma))$.
The probability that any one estimate is not good
is at most $1/3L$ by our choice of $\delta$, so by the union bound,
all are good except with probability $1/3$; hence
we've given a quantum algorithm which achieves
worst-case bounded error for computing $f$, and
whose query cost is
$O(\PQ_{\dot{\gamma}}(f)/\gamma\cdot\log(1/\gamma)\log\log(1/\gamma))$,
as desired.
\end{proof}

Using this theorem, we now proceed to prove a strong minimax
theorem for $\PQ_{\dot{\gamma}}(f)$, showing that a single
hard distribution $\mu$ works to lower bound this measure
for all values of $\gamma$ at once.

\begin{theorem}
Fix a finite alphabet $\Sigma$ as well as $n\in\bN$.
Let $f$ be a Boolean-valued function with $\Dom(f)\subseteq\Sigma^n$.
Then there exists a distribution $\mu$ over $\Dom(f)$
such that for any $\gamma\in[0,1]$, we have
\[\PQ_{\dot{\gamma}}^\mu(f)\ge\gamma\cdot\tOmega(\Q(f)),\]
where the constants in the $\tOmega$ notation
are universal.
\end{theorem}

As usual, the notation $\PQ_{\dot{\gamma}}^\mu(f)$
denotes the expected cost
of a probabilistic quantum algorithm which is required to achieve
bias at least $\gamma$ against $\mu$ (rather than in the worst case);
that is, the algorithm and the bias level $\gamma$ are both
allowed to depend on the distribution $\mu$.
Note that since $\PQ(f)$ is always smaller than $\Q(f)$
for any given bias level, this implies \thm{minimax_quantum}.

\begin{proof}
Fix $\Sigma$, $n$, and $f$. Let $\mathcal{R}$ be the set of all
probabilistic quantum algorithms for computing $f$.
For each $P\in\mathcal{R}$ and each distribution $\mu$
over $\Dom(f)$, define $\cost(P,\mu)\coloneqq |P|$
and define $\score(P,\mu)$ to be the bias $P$ makes against
distribution $\mu$ for computing $f$ (this will be in the range
$[-1,1]$). We will use \thm{main_minimax}. It is clear that
$\mathcal{R}$ is convex, and that $\Dom(f)$ is a nonempty finite set.
Let $\Delta$ denote the set of all probability distributions
over $\Dom(f)$. Then $\cost$ and $\score$ are
continuous functions $\mathcal{R}\times\Delta\to\bR$, with
$\cost(\cdot,\cdot)$ always non-negative, and both functions
are linear in both variables. These functions
are well-behaved, since finite cost and score can be achieved
(some quantum algorithm computes $f$ with positive bias),
the cost is independent of the input,
and mixing a zero-cost algorithm with a nonzero-cost algorithm
gives a nonzero-cost algorithm. Hence \thm{main_minimax}
gives us
\[\inf_{P\in\mathcal{R}}\max_{x\in\Dom(f)}
    \frac{|P|}{\score(P,x)^{+}}
=\max_{\mu\in\Delta}\inf_{P\in\mathcal{R}}
    \frac{|P|}{\score(P,\mu)^{+}},\]
where we use the convention $r/0=\infty$ for all $r\in\bar{\bR}$.

We simplify the left-hand side. For a probabilistic quantum
algorithm $P$, use $\bias_f(P)$ to denote its worst-case bias,
that is, $\bias_f(P)\coloneqq \min_{x\in\Dom(f)}\score(P,x)$.
Then the left-hand side is the infimum over $P$ of
$|P|/\bias_f(P)^{+}$. Since a probabilistic algorithm $P$
with $\bias_f(P)\le 0$ will never be selected in this infimum,
the left-hand side is equal to
\[\inf_{\gamma\in(0,1]}\inf_{P\in\mathcal{R}_\gamma}
    \frac{|P|}{\gamma},\]
where $\mathcal{R}_\gamma$ denotes the set of all probabilistic
quantum algorithms which achieve worst-case bias at least $\gamma$. The inner infimum is the definition of
$(1/\gamma)\cdot \PQ_{\dot{\gamma}}(f)$,
so the left-hand side equals $\inf_{\gamma\in(0,1]}\PQ_{\dot{\gamma}}(f)/\gamma$.

Note that this is at most $3\Q(f)$
by picking $\gamma=1/3$ and using $\PQ(f)\le\Q(f)$.
We claim there is no reason to use any $\gamma\in(0,1/6\Q(f))$
in the infimum.
The reason is that if $P$ is a probabilistic quantum
algorithm achieving worst-case bias at least $\gamma$
such that $|P|/\gamma<3\Q(f)$, and if $\gamma<1/6\Q(f)$,
it means that $P$ has nonzero support on zero-cost
quantum algorithms. Without loss of generality,
we can assume $P=aP_0+bP_1+(1-a-b)P'$, where $P_0$
is a zero-cost algorithm that always outputs $0$,
$P_1$ is a zero-cost algorithm that always outputs $1$,
and $P'$ is a probabilistic algorithm with no support
on zero-cost algorithms. Let $c=\min\{a,b\}$,
and write $P=2cZ+(1-2c)P''$, where $Z$ is the $0$-cost
algorithm which is an even mixture of $P_0$ and $P_1$.
Then it is not hard to see that
$|P|=(1-2c)|P''|$ and $\score(P,\mu)=(1-2c)\score(P'',\mu)$
for all $\mu$. This means that $P''$ has the same
cost-to-score ratio as $P$ for all distributions $\mu$.
Hence we can always use $P''$ in place of $P$ for the
infimum. Further, supposing without loss of generality
that $b\ge a$, we have $P''=(b-a)P_1+(1-b+a)P'$.
Since $f$ is not constant, let $x$ be an input on
which $f(x)$ disagrees with $P_1(x)$ (that is, a $0$-input).
Then note that if $b-a\ge 1/2$, the algorithm $P''$
cannot output $0$ on $x$ with probability above $1/2$,
so $\score(P'',x)\le 0$ and $P''$ will not be used
in the infimum. On the other hand, if $b-a<1/2$,
we have $|P''|=(1-b+a)|P'|> (1/2)\cdot 1=1/2$, as
$P'$ does not place weight on algorithms which make $0$ queries.
Now, unless $P''$ achieves worst-case bias at least $1/(6\Q(f))$,
its ratio of cost to score would be greater than $3\Q(f)$, which
we already know is achievable.

This means we only need
to use $\gamma>1/(6\Q(f))$ in the infimum.
Thus the left-hand side equals
\[\inf_{\gamma\in[\frac{1}{6\Q(f)},1]}
    \frac{\PQ_{\dot{\gamma}}(f)}{\gamma}.\]
Using \thm{PQ_Q}, this is at least
\[\inf_{\gamma\in[\frac{1}{6\Q(f)},1]}
    \frac{\Q(f)}{C\log(1/\gamma)\log\log(1/\gamma)}\]
for some universal constant $C$. The above is clearly
optimized at $\gamma=1/(6\Q(f))$, which means the left-hand
side is at least
$\Omega\left(\frac{\Q(f)}{\log\Q(f)\log\log\Q(f)}\right)$.

Looking at the right hand side,
we see that there exists a distribution
$\mu$ such that every probabilistic quantum algorithm $P$
satisfies $|P|/\score(P,\mu)^{+}\ge \tOmega(\Q(f))$,
from which the desired statement follows.
\end{proof}

\subsection{Abstraction of the query complexity argument}\label{sec:abstraction}

We note that the argument we used to prove the existence
of the hard distribution for quantum query complexity
only used a few properties of quantum algorithms.
Since we will want to apply the same argument to
quantum communication, polynomial degree, and logrank,
it makes sense to step back and provide an abstraction
of this argument to more general models.

In general, we will consider Boolean-valued functions $f$
with a finite input set $\Dom(f)$.
We will have a set $\mathcal{A}$ of algorithms
that may attempt to compute $f$. Formally, we will need
$\mathcal{A}$ to be a subset of a real vector space.
Each $A\in\mathcal{A}$ will have an associated \emph{cost},
denoted $|A|$, with $|\cdot|\colon \mathcal{A}\to[0,\infty)$.
We write $\mathcal{A}_T$ to denote the set
$\{\,A\in\mathcal{A}:|A|\le T\,\}$.

For an algorithm $A\in\mathcal{A}$
and an input $x\in\Dom(f)$, we let $\bias_f(A,x)$
denote the bias of algorithm $A$ on input $x$.
For now, the only property we need of the bias
is that it is a function
$\bias_f\colon\mathcal{A} \times \Dom(f)\to[-1,1]$.
The worst-case bias of an algorithm $A$ will be
denoted $\bias_f(A)\coloneqq \min_{x\in\Dom(f)}\bias_f(A,x)$.
If $\mu$ is a distribution over $\Dom(f)$,
we will further write
$\bias_f(A,\mu)\coloneqq\bE_{x\sim\mu}[\bias_f(A,x)]$.
Similarly, if $P$ is a probability distribution
over $\mathcal{A}$ with finite
support, we denote $\bias_f(P,\mu)
\coloneqq\bE_{A\sim P}\bE_{x\sim\mu}[\bias_f(A,x)]$
and $\bias_f(P)\coloneqq\min_{x\in\Dom(f)}\bias_f(P,x)$.
We also set $|P|\coloneqq\bE_{A\sim P} |A|$.
Finally, we define
$M(f)\coloneqq \inf_{A\in\mathcal{A}:\bias_f(A)\ge 1/3}|A|$.

So far, this setting is extremely general, capturing many
computational models. For the quantum-style strong minimax
to work, we will need the following properties to also hold
for a given function $f$.

\begin{enumerate}
    \item $\mathcal{A}_T$ is convex for each $T\in[0,\infty)$,
    and $\bias_f(\cdot,x)$ is linear over
    $A\in\mathcal{A}_T$ for each $x\in\Dom(f)$.
    \label{itm:convex}
    \item There exists some 
    $A\in\mathcal{A}$ such that $\bias_f(A)\ge 1/3$.
    (Equivalently, $M(f)<\infty$.)
    \label{itm:finite}
    \item All $A\in\mathcal{A}$ with $|A|<1$ have
    $|A|=0$, and $\mathcal{A}_0$ is the convex hull
    of exactly two algorithms, $Z_0$ and $Z_1$.
    For each $x\in\Dom(f)$, we also have
    $\bias_f(Z_0,x)=-\bias_f(Z_1,x)=\pm 1$,
    and if $f$ is not constant, $\bias_f(Z_0,x)$ attains both values
    $1$ and $-1$ for $x\in\Dom(f)$.
    \label{itm:zero_cost}
    \item Suppose $P$ is a probability distribution
    over $\mathcal{A}$ that has support $\{A_1,A_2,\dots,A_k\}$,
    with probability $p_i$ for $A_i$, such that (a)
    $|A_i|\le 2^i T$ for some $T\in[1/10,\infty)$, (b)
    $\sum_i 2^i p_i\le 5$, and (c) $\bias_f(P)\ge 2^{-k-1}$.
    Then there is some $A\in\mathcal{A}$ with $\bias_f(A)\ge 1/3$
    and $|A|\le 2^kT\cdot \poly(k)$
    (with the constants in the $\poly$ being universal).
    \label{itm:amplify}
\end{enumerate}

We note that \itm{convex} essentially requires the computational
model to be randomized (or, in communication complexity, to
have public randomness). \itm{finite} only says
that each function can be computed by some finite-cost algorithm.
\itm{zero_cost} says that algorithms with cost less than $1$
cannot look at the input, and therefore have cost $0$
and must either always output $0$ or always output $1$
(or some convex combination of the two).

The main important point is \itm{amplify}.
This point amplifies a certain restricted type of
low-bias probability distribution over algorithms into
a full-blown constant-bias algorithm, and the cost of amplification
is nearly linear in one over the bias.

We now prove that these points together suffice to guarantee
the existence of a strongly-hard distribution.
To start, we establish the following lemma, which says
that if \itm{amplify} holds -- meaning we can amplify the restricted
type of probabilistic algorithms -- then we can amplify all
probabilistic algorithms.

\begin{lemma}\label{lem:amplify_many_from_few}
Suppose $f$ and $\mathcal{A}$ satisfy the above
conditions. Let $P$ be any finite-support probability
distribution over $\mathcal{A}$ with $\bias_f(P)>0$. Then
\[M(f)\le \frac{|P|}{\bias_f(P)}\cdot\polylog(1/\bias_f(P)).\]
\end{lemma}

\begin{proof}
The proof of this will be directly analogous to the quantum
query case. We convert $P$ into the restricted form of
\itm{amplify}, being careful to lose only a constant factor
in the bias and in the cost.
Let $\gamma\coloneqq\bias_f(P)>0$. We first use Markov's
inequality to argue that the total probability mass $P$
places on algorithms $A$ of cost $|A|\ge 2|P|/\gamma$
is at most $\gamma/2$, and hence discarding all such algorithms
from the support of $P$ decreases its bias by at most
$\gamma/2$ (while not increasing its cost).
Next, we group the remaining algorithms in the support
of $P$ into $\log(1/\gamma)$ bins: one bin
for algorithms of cost $0$ to $2T$
(with $T$ equal to something like $4|P|$),
and one additional bin for algorithms of cost
$2^iT$ to $2^{i+1}T$ for $i$ between $1$ and $\log(1/\gamma)$.
Within each bin, we use the convexity of $\mathcal{A}_{2^{i}T}$
to replace the entire bin with a single algorithm
(whose cost is up to the upper boundary of that bin).
For the first bin, this increases the cost $|P|$
by up to an additive $O(T)$, while for the other bins,
this increases the cost by up to a factor of $2$. Altogether,
we have only $\log(1/\gamma)$ algorithms
remaining in the support, and setting $k=\log(1/\gamma)$
it is not hard to check that the conditions in \itm{amplify}
are satisfied.
\end{proof}

\begin{theorem}\label{thm:quantum_abstraction}
Suppose $f$ and $\mathcal{A}$ satisfy the above conditions.
Then there exists a distribution $\mu$ over $\Dom(f)$
such that for any finite-support probability distribution $P$
over $\mathcal{A}$, we have
\[\bias_f(P,\mu)\le O(M(f)/|P|\cdot\polylog M(f)).\]
In particular, if $M_{\dot{\gamma}}^\mu(f)$ denotes the infimum
cost $|A|$ over algorithms $A\in\mathcal{A}$ with
$\bias_f(A,\mu)\ge\gamma$, then for all $\gamma\in(0,1/3)$ we have
\[M_{\dot{\gamma}}^\mu(f)\ge\gamma\cdot\tOmega(M(f)).\]
\end{theorem}

\begin{proof}
The proof will be exactly the same as in the quantum query setting.
In the special case where $f$ is constant, the result trivially
follows as $M(f)=0$, so assume $f$ is not constant.

First, we let $\mathcal{R}$ the set of all finite-support probability
distributions over $\mathcal{A}$, and let $\Delta$
be the set of probability distributions over $\Dom(f)$.
Then we define $\cost\colon\mathcal{R}\times\Delta\to[0,\infty)$
by $\cost(P,\mu)\coloneqq |P|$, and
$\score\colon\mathcal{R}\times\Delta\to[-1,1]$ by
$\score(P,\mu)\coloneqq\bias_f(P,\mu)$.
Note that $\cost$ and $\score$ are both continuous and
linear in each variable. They are also well-behaved,
because $M(f)<\infty$ ensures finite cost and score
can be achieved, $\cost$ does not depend on $\mu$, and
cost is linear in $P$. Hence \thm{main_minimax} gives
\[\inf_{P\in\mathcal{R}}\max_{x\in\Dom(f)}
\frac{|P|}{\bias_f(P,x)^+}
=\max_{\mu\in\Delta}\inf_{P\in\mathcal{R}}
\frac{|P|}{\bias_f(P,\mu)^+}.\]

We examine the left-hand side. It equals
$\inf_{P\in\mathcal{R}}\frac{|P|}{\bias_f(P)^+}$.
We note that this infimum is at most $3M(f)$ by the definition
of $M(f)$. We now claim that there is no need to use
any $P$ in the infimum if $\bias_f(P)<1/(6M(f))$. To show this,
it suffices to show that there is no need to use
any $P$ in the infimum if $|P|<1/2$, because we know that
$3M(f)$ is attainable using only algorithms in $A$ with
cost at least $1$.

Now, suppose that $|P|<1/2$ and $\bias_f(P)>0$.
We can write $P=aZ_0+bZ_1+(1-a-b)P'$
where $P'$ has support only on $A\in\mathcal{A}$ with $|A|\ge 1$.
Define $P''\coloneqq (a-c)Z_0+(b-c)Z_1+(1-a-b+2c)P'$, where
$c=\min\{a,b\}$. Then as we showed in the quantum query case,
we have $|P''|/\bias_f(P'')\ge |P|/\bias_f(P)$.
Moreover, since $f$ is not constant, there is some input
$x\in\Dom(f)$ such that $\bias_f(Z_0,x)=-1$, and some input
$y\in\Dom(f)$ such that $\bias_f(Z_1,y)=-1$.
Since $\bias_f(P'')>\bias_f(P)>0$, and since $\bias_f(P')\le 1$,
we must have $(1-a-b+2c)>1/2$, meaning that $|P''|>1/2$,
as desired.

Hence the left-hand side equals
$\inf_{P\in\mathcal{R}'}\frac{|P|}{\bias_f(P)}$,
where $\mathcal{R}'$ is the set of all $P\in\mathcal{R}$
with $\bias_f(P)\ge 1/(6M(f))$. Using \lem{amplify_many_from_few},
we know that for each $P\in\mathcal{R}'$,
we have
$M(f)\le |P|/\bias_f(P)\cdot\polylog(1/\bias_f(P))
\le |P|/\bias_f(P)\cdot\polylog M(f)$. Hence the left-hand
side is at least $\frac{M(f)}{\polylog M(f)}$.
Finally, examining the right hand side,
we see that there is a distribution $\mu$ over
$\Dom(f)$ such that for all $P\in\mathcal{R}$,
we have $|P|\ge \bias_f(P)\cdot M(f)/\polylog M(f)$,
and the desired result follows.
\end{proof}

\subsection{Quantum communication complexity}
\label{sec:qcc}

To prove an analogous minimax for quantum communication complexity,
all we need is to show that quantum communication complexity
satisfies the four conditions from \sec{abstraction}.
It's easy to see that as long as there is public randomness
(whether or not there is also shared entanglement), the first
three conditions are satisfied. It remains to
deal with the fourth condition. Let $P$ be
a probability distribution over protocols
$\Pi_1,\Pi_2,\dots,\Pi_k$, which assigns probability $p_i$
to $\Pi_i$ and satisfies $|\Pi_i|\le 2^i T$,
$\sum_i 2^ip_i\le 10$, and $P$ achieves bias at least
$2^{-k-1}$ for computing communication function $F$
on any input $(x,y)\in\Dom(F)$. Our goal is to construct
a communication protocol which uses $T\cdot\tO(2^k)$
communication to compute $F$ to bounded error.

As in the quantum query case, all we need to do is
create a protocol $\Pi$ in which Alice and Bob estimate
the biases $\Pi_i(x,y)$ of the protocols $\Pi_i$ when run on
their inputs. Each estimate for protocol $i$ needs
to be within $2^{-(k-i)}/20$ of the correct bias,
and it must satisfy this property with probability at least
$1-1/3k$ (see the query complexity section
for a formal analysis). To achieve this, it suffices for
Alice and Bob to use amplitude estimation
from \thm{amplitude_estimation} to generate
an estimate of the probability $\Pi_i(x,y)$ outputs $1$.
Hence the only remaining difficulty is running
amplitude estimation of a communication protocol
in the communication complexity setting.

This turns out to be possible in both the shared-entanglement
and the non-shared-entanglement settings (though note that
we've already assumed shared randomness, so we cannot
handle the non-shared-randomness non-shared-entanglement
quantum communication complexity model).
The idea is to have one of the players, say Alice, take charge.
We will assume that Alice is the one who outputs the final
answer in $\Pi_i$. Then from Alice's point of view,
$\Pi_i(x,y)$ can be viewed as a unitary $U$ and a measurement
$M$ such that Alice needs Bob's help to apply $U$,
and after applying $U$ to a shared state $\ket{0}_A\ket{0}_B$,
Alice can apply the measurement $M$ on her side alone to get
the output $\Pi_i(x,y)$. Now, to apply amplitude estimation,
Alice only needs the ability to apply controlled
$U$, $U^\dagger$, and $(I-2M)$ operations. She can do the latter
alone. For controlled $U$ and $U^\dagger$ applications, she
needs Bob's help, but that's fine: she will just send him
a qubit each time alerting him to whether they are about to apply
$U$ or $U^\dagger$ to their shared state (Bob will
return that qubit afterwards to ensure coherence of Alice's
controlled applications of $U$ and $U^\dagger$).

We conclude the following theorem.

\begin{theorem}
\label{thm:quantum-communication}
Let $F\colon\mathcal{X}\times\mathcal{Y}\to\B$ be a
(possibly partial) communication function. Then there exists
a probability distribution $\mu$ over $\Dom(F)$ such that
for all $\gamma\in(0,1/3)$, we have
\[\QCC^\mu_{\dot{\gamma}}(F)\ge\gamma\cdot\tOmega(\QCC(F)).\]
Here $\QCC^\mu_{\dot{\gamma}}(F)$ denotes the minimum
amount of communication required by a quantum communication
protocol which achieves bias at least $\gamma$ against $\mu$.
This theorem works in both the shared entanglement setting
and in the shared-randomness, non-shared entanglement setting.
\end{theorem}

\section{Approximate polynomial degree and logrank}
\label{sec:polynomials}

As in the quantum case, polynomials can be amplified
linearly in the bias. However, also as in the quantum case,
the degree of polynomials is not convex: the degree of
the convex combination of $p_1$ and $p_2$ is the maximum
degree of $p_1$ and $p_2$, not the average degree.

The same ideas that worked for quantum query and communication
complexities will allow us to get a strong hard distribution
for approximate polynomial degree and approximate logrank.
The main difference will be how we do the estimation
of success probabilities:
instead of amplitude estimation, we will need a polynomial
variant of this, which turns out to be a little tricky.

\subsection{Approximate degree}

The approximate degree of a (possibly partial)
Boolean function $f\colon\B^n\to\B$
is the minimum degree of an $n$-variate polynomial $p$
which satisfies $|p(x)-f(x)|\le \epsilon$
for all $x\in\Dom(f)$, where $\epsilon$ is a parameter
representing the allowed error.
When $f$ is a partial function,
there are actually two different notions of polynomial
degree: one where $p$ is required to be bounded
on the entire Boolean hypercube
(that is, $p(x)\in[0,1]$ for all $x\in\B^n$,
even when $x\notin\Dom(f)$), and one where $p$
is not restricted outside the domain of $f$.
Our results will apply to both versions of polynomial
degree, but for conciseness, we restrict our attention
to the bounded version.

With polynomials, it is often convenient to switch
from talking about functions $f\colon\B^n\to\B$ to talking
about functions $f\colon\{+1,-1\}^n\to\{+1,-1\}$.
Note that by doing a simple variable substitution,
we can convert between $\B$ variables to $\{+1,-1\}$
variables without changing the degree of the polynomial.
That is, we can substitute $1-2x_i$ in place of the variable
$x_i$ inside $p$ to make it take $\B$ inputs instead
of $\{+1,-1\}$ inputs, and we can substitute
$(1-x_i)/2$ to go the other way. We can similarly change
the output of $p$ from being in the range $[0,1]$
to the range $[-1,1]$ and vice versa (the error
changes by a factor of $2$ when switching between
these bases). Another well-known observation
is that to approximate a Boolean function $f$,
we only need \emph{multilinear} polynomials,
and their degree only needs to be at most $n$.

To get our hard distribution, we will use \thm{quantum_abstraction}.
We need to check the four conditions, but using
polynomials as our ``algorithms''. More explicitly,
the set $\mathcal{A}$ will be the set of all real
$n$-variate multilinear bounded polynomials,
viewed in the $\{+1,-1\}$ basis
(bounded means that $p(x)\in[-1,1]$ for all
$x\in\{+1,-1\}^n$).
For a polynomial $p\in\mathcal{A}$, we define $\bias_f(p,x)$
to be $f(x)p(x)$. Then \itm{convex} holds,
as the set of polynomials of a given degree is convex
and $\bias_f(\cdot,x)$ is linear over that set.
\itm{finite} holds because every Boolean function
can be computed exactly by a polynomial of degree $n$.
Next, \itm{zero_cost} holds because polynomials
of degree less than $1$ have degree $0$, and since
we're dealing with bounded polynomials, these are a
convex combination of the two constant polynomials
$-1$ and $1$.

It remains to show \itm{amplify}. To this end,
let $P$ be a probability distribution over $k$
polynomials $q_1,q_2,\dots,q_k$, with $\deg(q_i)\le 2^iT$.
Let $p_i$ be the probability $P$ assigns to $q_i$,
and suppose $\sum_{i=1}^k 2^ip_i\le 5$.
Finally, suppose that $\bias_f(P)\ge 2^{-k-1}$.
Our goal is to find a polynomial $q$ of degree
at most $2^kT\cdot\poly(k)$ that computes $f$ to
constant error. To do so, we'll need a polynomial
version of the amplitude estimation algorithm
we did in the quantum case. That is, we'd like
to estimate the output that polynomial $q_i(x)$
returns, and do arithmetic computations on it.
Crucially, one of the arithmetic computations we'd
like to do is comparison, for example, to see if
$q_i(x)>0$. Such a comparison is not a polynomial operation,
so we cannot use the polynomial $q_i(x)$ itself.
What we'll do instead is to create polynomials
that compute the \emph{bits} of the binary expansion
of $q_i(x)$, to a certain precision. We will then
do arithmetic operations
using those bits, and we'll be able to implement
those operations using polynomials.

To do so, we'll need some approximation theory.
The following theorem, known as Jackson's theorem,
will be useful. It traces back to Jackson (1911)
\cite{Jac11}, but see also \cite{MMR94} (page 750,
Theorem 3.1.1) for some discussion and a more thorough
list of references.

\begin{theorem}[Jackson's theorem]\label{thm:jackson}
Let $\alpha\colon[-1,1]\to\bR$ be a continuous
function, and let $n\in\bN$. Then there is
a real polynomial $p$ of degree $n$ such that
for all $x\in[-1,1]$, we have
\[|p(x)-\alpha(x)|\le 6\cdot\sup_{|y-z|
\le 1/n}|\alpha(y)-\alpha(z)|.\]
\end{theorem}

In particular, if $\alpha$
has Lipschitz constant $K$, then for each $n\in\bN$
there is a polynomial $p_n$ of degree
at most $n$ which approximates
$\alpha$ to within an additive $6K/n$ at each point
in $[-1,1]$. Jackson's theorem
can be used to prove the well-known
result that polynomials can be amplified with
a linear dependence in the bias.
For completeness, we reprove this here
(see also e.g.\ \cite{GKKT17}).

\begin{corollary}
[Polynomial amplification (small bias to constant bias)]
For each $\gamma\in(0,1)$,
there is a real polynomial $p$ of degree at most $13/\gamma$
such that $p$ maps $[-1,1]$ to $[-1,1]$,
$p$ maps $[-1,-\gamma]$ to $[-1,-1/3]$, and
$p$ maps $[\gamma,1]$ to $[1/3,1]$.
\end{corollary}

\begin{proof}
Let $\alpha\colon[-1,1]\to\bR$ be the function with
$\alpha(x)=-2/3$ for $x\in[-1,-\gamma]$,
$\alpha(x)=2/3$ for $x\in[\gamma,1]$, and
$\alpha(x)=2x/3\gamma$ for $x\in(-\gamma,\gamma)$.
Then $\alpha$ is continuous and has Lipschitz
constant $2/3\gamma$. By \thm{jackson},
for every $n\in\bN$, there exists a polynomial
$p_n$ of degree at most $n$ which approximates
$\alpha$ to additive error $4/\gamma n$.
Picking $n=\lceil 12/\gamma\rceil\le 13/\gamma$,
we get a polynomial which approximates $\alpha$
to error $1/3$, which means it has the desired
properties.
\end{proof}

We will also need an amplification polynomial
that goes from constant bias to small error.
We reprove the following well-known lemma here
for completeness (it also appears
in \cite{BNRdW07}, and another version appears
in \cite{She13}).

\begin{lemma}
[Polynomial amplification (constant error to small error)]
\label{lem:polynomial_amplification}
For each $\epsilon\in(0,2/3)$, there is a real polynomial
$p$ of degree at most $17\log(1/\epsilon)$ such that $p$ maps
$[-1,1]$ to $[-1,1]$, $p$ maps $[-1,-1/3]$ to $[-1,-(1-\epsilon)]$,
and $p$ maps $[1/3,1]$ to $[1-\epsilon,1]$.
\end{lemma}

\begin{proof}
We set
\[q(x)=\sum_{i=0}^k\binom{2k+1}{i}
\left(\frac{1+x}{2}\right)^i
\left(\frac{1-x}{2}\right)^{2k+1-i},\]
and set $p(x)=1-2q(x)$. Note that for $x\in[-1,1]$,
the value $q(x)$ is exactly the probability that,
when flipping a coin $2k+1$ times, less than half
of the coin flips will come out heads, assuming the probability
of heads is $(1+x)/2$. Because of this interpretation,
we know that $q$ maps $[-1,1]$ to $[0,1]$ and is decreasing
in $x$, so $p$ maps $[-1,1]$ to $[-1,1]$ and is increasing in $x$.
We also have $q(x)=1-q(-x)$, which means that
$p(-x)=-p(x)$, i.e.\ $p$ is odd. Given these properties,
the lemma will follow if we show that $p(1/3)\ge 1-\epsilon$,
or equivalently, that $q(1/3)\le \epsilon/2$.

We have
\[q(1/3)=\sum_{i=0}^k\binom{2k+1}{i}
\left(\frac23\right)^i
\left(\frac13\right)^{2k+1-i}
=3^{-(2k+1)}\sum_{i=0}^k\binom{2k+1}{i}2^i
\le 3^{-(2k+1)}2^k\sum_{i=0}^k\binom{2k+1}{i}\]
\[=3^{-(2k+1)}2^k2^{2k}
=(1/3)(8/9)^k.\]
To get this to be smaller than $\epsilon/2$, it suffices
to pick $k$ large enough so that $(8/9)^k\le\epsilon$,
or equivalently, $k\ge \frac{1}{\log(9/8)}\log(1/\epsilon)$.
Hence we can pick $k=\lceil \frac{1}{\log(9/8)}\log(1/\epsilon)\rceil
\le \frac{1}{\log(9/8)}\log(1/\epsilon)+1$.
The degree of $p$ will be
$2k+1\le \frac{2}{\log(9/8)}\log(1/\epsilon)+3$.
Note that $\epsilon\le 2/3$, so $\log(1/\epsilon)\ge \log(3/2)$,
and hence $\frac{2}{\log(9/8)}\log(1/\epsilon)+3\le
\left(\frac{2}{\log(9/8)}
+\frac{3}{\log(3/2)}\right)\log(1/\epsilon)
\le 17\log(1/\epsilon)$.
\end{proof}

Equipped with these approximation-theoretic tools,
we will now tackle \itm{amplify}, showing that probability
distributions over polynomials (which achieve a small
amount of worst-case bias $\gamma$ for computing $f$) can 
be amplified to polynomials which compute $f$ to constant
error, using only a nearly-linear dependence on $1/\gamma$.

\begin{lemma}
As in \itm{amplify}, let $P$ be a probability distribution
over $k$ bounded multilinear polynomials $q_1,q_2,\dots,q_k$,
which assigns them probabilities $p_1,p_2,\dots,p_k$
respectively. Suppose that $\sum_{i=1}^k 2^ip_i\le 5$,
that $\deg(q_i)\le 2^i T$ for some real number $T$,
and that $f(x)\sum_{i=1}^k p_iq_i(x)\ge 2^{-k-1}$
for all $x\in\Dom(f)$. Then there is a bounded
multilinear polynomial $q$ which approximates
$f$ with bias at least $1/3$ and which satisfies
$\deg(q)\le 2^kT\cdot\poly(k)$.
\end{lemma}

\begin{proof}
Recall that in the quantum case, we estimated
the bias of the $i$-th algorithm to within
$2^{-(k-i)}/20$, with success probability
at least $1-1/3k$. We will do a polynomial version of this.
What does estimating $q_i(x)$ mean, for polynomials?
It means we will construct polynomials
which approximately compute the bits in the binary expansion
of the number $q_i(x)$. We will have one polynomial
for the sign, and an additional $k-i+4$ polynomials
for the first $k-i+4$ digits in the binary expansion
of $q_i(x)$.

In order to do so, we compose univariate polynomials
with $q_i$. This way, the task reduces to creating
univariate polynomials which output the bits in
the binary expansion of their input (assuming
they all receive the same input). More explicitly,
the correctness condition is as follows.
We say the binary expansion of a real number
$\beta\in[-1,1]$ is $2^{-\ell}$-\emph{robust} to $t$ bits if
the first $t$ bits of the binary expansion of
$\beta+\epsilon$ is the same
as that of $\beta$ for all
$\epsilon\in[-2^{-\ell},2^{-\ell}]$.
Then we require univariate polynomials
$d_0^\ell,d_1^\ell,\dots, d_k^\ell$
such that if $\beta\in[-1,1]$ is $2^{-\ell}$-robust to
at least $t$ bits,
then $d_t^\ell(\beta)$ is within $O(1/k^{10})$
of the $t$-th bit in the binary expansion of $\beta$.
The polynomial $d_0^\ell$ needs to output
the sign of $\beta$ if $\beta$ is $2^{-\ell}$-robust to 
at least $0$ bits
(that is, if the sign of $\beta$ does not change upon adding
or subtracting $2^{-\ell}$).
We will also require all these polynomials to be bounded,
i.e.\ they must map $[-1,1]$ to $[-1,1]$.

To implement these polynomials, we use \thm{jackson}.
For simplicity, let's represent the bits in the binary
expansion using $+1$ and $-1$ instead of $0$ and $1$
(converting back is easy). Consider the function
$\alpha_i$ which outputs the $i$-th bit
of the binary expansion of its input (or the sign
if $i=0$). This $i$ is a step function: for $i=0$,
$\alpha_0(\beta)$ jumps from $-1$ to $1$ at $\beta=0$;
for $i=1$, $\alpha_1(\beta)$ similarly jumps from
$-1$ to $1$ and back at $\beta=-1/2,0,1/2$. More generally,
$\alpha_i$ has $2^{i+1}$ different plateaus of
$1$ or $-1$ on its domain $[-1,1]$.
Now, since we only care about getting the $i$-th bit correct
if the $i$-th bit is robust to $\beta$ changing
by $2^{-\ell}$, consider the continuous functions
$\alpha_i^\ell$ which make the jumps from $-1$ to $1$
continuous by starting from $2^{-\ell}$ before
the jump point, ending $2^{-\ell}$ after the jump point,
and drawing a continuous line in between
(the slope of the line will be $\pm 2^{-\ell}$).
This is well-defined as long as $\ell$ is sufficiently
larger than $i$, say $\ell\ge i+2$.

Note that $\alpha_i^\ell$ has Lipschitz constant
$2^{-\ell}$. This means we can use \thm{jackson}
to estimate $\alpha_i^\ell$ by a polynomial of degree
$O(2^\ell)$ which achieves constant additive error
(say, $1/10$). We can scale down these polynomials slightly
to ensure they remain bounded in $[-1,1]$. We then plug them
into a single variate bounded polynomial of degree
$O(\log k)$ that we get from \lem{polynomial_amplification},
in order to amplify the error down to $O(1/k^{10})$.
The result are polynomials $d_t^\ell$ (for $\ell\ge t+2$)
that have degree $O(2^\ell\log k)$ and, on input
$\beta$ which is $2^{-\ell}$-robust to bit at least $t$,
correctly output the $t$-th bit of $\beta$
except with additive error $O(1/k^{10})$.

Now, to get an estimate of $q_i(x)$ to $k-i+5$
bits, we set $\ell=k-i+O(\log k)$ and compose
$d_t^\ell(q_i(x))$ for $t=0,1,2,\dots,k-i+5$.
Actually, we scale down $q_i(x)$ and
add an extra variable $y_i$
representing a noise term for $q_i(x)$;
the final estimating polynomials will be
the $n+1$ variate polynomials
$r_{i,t}(x,y_i)\coloneqq d_t^\ell((9/10)q_i(x)+y_i)$.
Note that the degree of $r_{i,t}$ is
$O(2^{k-i+O(\log k)}\log k\cdot \deg(q_i))=O(2^k T\poly(k))$.

Next, consider the function which takes binary
representations (to $k+5$ bits each) of numbers
$\lambda_i\in[-1,1]$,
and outputs the sign of $\sum_{i=1}^k p_i\lambda_i$, where
$p_i$ are known non-negative constants which sum to $1$.
This is a Boolean function of $O(k^2)$ variables,
so it can be computed exactly by a multilinear polynomial
of degree $O(k^2)$. Call this polynomial $s$.
Next, plug in the polynomials $r_{i,t}$ into the inputs
of $s$, so that $s$ calculates the sign of the sum
$\sum_{i=1}^k p_i\tilde{\beta}_i$ where each
$\tilde{\beta}_i$ is the estimate of $(9/10)q_i(x)+y_i$
that is computed by the polynomials $d^{k-i+10}_t$.
Call this composed polynomial $u(x,y)$.

Observe that $u(x,y)$ is a polynomial in $n+k$ variables
($n$ variables from $x$ and $k$ variables $y_i$),
and has degree $O(2^kT\poly(k))$.
This polynomial attempts to compute the sign of
$(9/10)\sum_{i=1}^k p_iq_i(x)+\sum_{i=1}^k p_iy_i$.
Since we know that
$\sum_{i=1}^k p_iq_i(x)\cdot f(x)\ge 2^{-k-1}$,
this sign computed by $u(x,y)$ will equal $f(x)$
so long as $\left|\sum_{i=1}^k p_iy_i\right|\le 2^{-k-2}$.
Recall that $\sum_{i=1}^k 2^ip_i\le 5$.
Hence to guarantee that
$\left|\sum_{i=1}^k p_iy_i\right|\le 2^{-k-2}$,
it suffices to choose each $y_i$ such that $|y_i|\le 2^{-(k-i+5)}$.
Now, let's call $q_i(x)+y_i$ good if it
is $2^{-(k-i+O(\log k))}$-robust to $k-i+5$ bits.
If all $q_i(x)+y_i$ are good for all $i$, then $r_{i,t}$
correctly compute the bits to additive error $O(1/k^{10})$,
then a multilinear polynomial of degree $O(k^2)$ in $O(k^2)$
variables will still correctly compute its output to
small error, certainly $O(1/k)$. Hence if all $q_i(x)+y_i$
are good for all $i$ and if $y_i\le 2^{-k-2}/k$ for all $i$,
$u(x,y)$ outputs $f(x)$ to error $O(1/k)$.

To ensure that $q_i(x)+y_i$ are good, we pick $y_i$ at random.
That is, we have an allowed range $[-2^{-(k-i+5)},2^{-(k-i+5)}]$
for $y_i$; we fit $\poly(k)$ evenly spaced points into
this range, so that the gap between the points is
$2^{-(k-i+O(\log k))}$. Note that for all but a constant
number of choices of $y_i$ among these $\poly(k)$ options,
the resulting number $q_i(x)+y_i$ will be
$2^{-(k-i+O(\log k))}$-robust to $k-i+5$ bits. Hence
by randomly selecting $y_i$, the probability that $q_i(x)+y_i$
is not good is at most $O(1/\poly(k))$. By the union bound,
this choice means that all $q_i(x)+y_i$ are good except
with constant probability.
Hence $u(x,y)$ computes $f(x)$ to $O(1/k)$ error with high
probability when $y$ is chosen at random according to the above
procedure.

Finally, we let $q(x)$ be the average of the polynomials
$u(x,y)$ for all possible choices of $y$ in the above procedure.
Since $u(x,y)$ outputs a number very close to $f(x)$
when $y$ is good, and since it is always bounded in $[-1,1]$,
and since $y$ is good with high probability, we conclude that
$q(x)$ computes $f(x)$ to bounded error. It is also bounded
outside the promise of $f$. The degree of $q(x)$
was $O(2^kT\poly(k))$. We note that $q(x)$
as we constructed it here can actually be viewed
as a polynomial $\rho$ in $k$ variables composed with the polynomials
$q_1,q_2,\dots, q_k$.
\end{proof}

The above amplification theorem allows us to conclude the
following theorem.

\begin{theorem}\label{thm:minimax_polynomials}
let $f\colon\{+1,-1\}^n\to\{+1,-1\}$ be a (possibly partial)
Boolean function. Then there is a vector $\psi\in[-1,1]^{\Dom(f)}$
such that $\|\psi\|_1=1$, $\langle \psi,f\rangle =1$,
and for any polynomial $p$ which is bounded
(i.e.\ $|p(x)|\le 1$ for $x\in\{+1,-1\}^n$), we have
\[\langle\psi,p\rangle\le 
\frac{\deg(p)}{\tOmega\left(\adeg(f)\right)}.\]
Here $\adeg(f)$ denotes the minimum degree of a bounded polynomial
$p$ which computes $f$ to bounded error. The constants in the
$\tOmega$ notation are universal.
\end{theorem}

\begin{proof}
This follows immediately by taking $\psi$ to be defined by
$\psi(x)=f(x)\mu[x]$, where $\mu$ is the hard distribution
we get from \thm{quantum_abstraction}.
\end{proof}

\subsection{Approximate logrank and gamma 2 norm}

Instead of tackling approximate logrank directly,
we use approximate $\gamma_2$ norm. This measure deserves
some introduction. First, we note that the $\gamma_2$
norm is a well-known norm of a matrix. One way to define it
is to say that $\gamma_2(A)$ is the minimum,
over factorizations $A=BC$ of $A$ into a product
of matrices $B$ and $C$, of the maximum $2$-norm of a
row of $B$ times the maximum $2$-norm of a column of $C$.
The $\gamma_2$ norm has several useful properties known
in the literature \cite{She12,LSS08}:
\begin{enumerate}
    \item $\gamma_2$ is a norm, so $\gamma_2(A)\ge 0$
    (with equality if and only if $A$ is the all-zeros matrix)
    and
    $\gamma_2(A+\lambda B)\le \gamma_2(A)+|\lambda|\gamma_2(B)$.
    \item $\gamma_2(A\otimes B)=\gamma_2(A)\gamma_2(B)$,
    where $\otimes$ denotes the tensor (Kronecker) product
    \item $\gamma_2(A\circ B)\le \gamma_2(A)\gamma_2(B)$,
    where $\circ$ denotes the Hadamard (entry-wise) product
    \item $\gamma_2(J)=1$ where $J$ is the all-ones matrix
    \item $\|A\|_\infty\le\gamma_2(A)\le\|A\|_\infty\sqrt{\rank(A)}$.
\end{enumerate}
In the above, $A$ and $B$ are matrices of the same dimensions,
and $\lambda$ is a scalar. $\gamma_2(A)$ can be thought of
as a smoother version of rank.

Let $F\colon\mathcal{X}\times\mathcal{Y}\to\{+1,-1\}$
be a (possibly partial) communication function. We
identify $F$ with its \emph{communication matrix},
which is a matrix with rows indexed by $\mathcal{X}$
and columns indexed by $\mathcal{Y}$, with the $(x,y)$
entry being $F(x,y)\in\{+1,-1\}$ if $(x,y)\in\Dom(F)$
and being $*$ if $(x,y)\notin\Dom(F)$. This way,
$F$ is a $\{+1,-1,*\}$-valued matrix.

For such a matrix $F$, we say that a real-valued matrix $A$
\emph{approximates} $F$ (to bias $1/3$) if
$|A[x,y]|\le 1$ for all $(x,y)\in\mathcal{X}\times\mathcal{Y}$
and $F(x,y)A[x,y]\ge 1/3$ for all $(x,y)\in\Dom(F)$.
The \emph{approximate $\gamma_2$ norm} of $F$, denoted
$\tilde{\gamma}_2(F)$, is defined as the minimum
value of $\gamma_2(A)$ over all matrices which approximate
$F$ to bias $1/3$. It is not hard to see that this minimum
is attained, as the set of such matrices is compact.

We will actually care about the logarithm of the approximate
$\gamma_2$ norm, that is, about $\log\tilde{\gamma}_2(F)$.
We note that the constant $1/3$ in the definition
of this measure is arbitrary,
as approximations to $F$ can be amplified with only
a constant factor overhead in the log-approximate-$\gamma_2$-norm
(see, e.g., \cite{BBGK18}). An annoying detail, however,
is that such amplification can in general lose not just a
multiplicative constant but also an additive constant,
since $\tilde{\gamma}_2(F)$ may in general be less than $1$
(meaning the logarithm of it will be less than $0$).
To avoid such complications, we will define our
measure of interest as
$M(F)\coloneqq\max\{1,\log\tilde{\gamma}_2(F)\}$
if $F$ is not constant and $M(F)=0$ if $F$ is constant,
and we will write $M_{\dot{\gamma}}(F)$ for the
bias $\gamma$ version of $M(F)$ instead of the
default bias $1/3$ version.

In order to get a minimax theorem analogous to
\thm{minimax_polynomials}, we will again use
\thm{quantum_abstraction}. Our set of algorithms
$\mathcal{A}$ will be the set of bounded real matrices $A$
(that is, real matrices $A$ of the same dimensions as $F$
which satisfy $|A[x,y]|\le 1$ for all
$(x,y)\in\mathcal{X}\times\mathcal{Y}$).
The cost of a matrix $A$ will be
$\cost(A)\coloneqq \max\{1,\log\gamma_2(A)\}$
if $A$ is not a multiple of the all-ones
matrix $J$, and otherwise $\cost(A)=0$ if $A=\lambda J$.
We define $\bias_F(A,(x,y))=F(x,y)A[x,y]$ for $(x,y)\in\Dom(F)$.

We show that $\mathcal{A}_T$ is convex for each $T\in[0,\infty)$.
For $T<1$, the set $\mathcal{A}_T$ is the set of all
matrices of the form $\lambda J$ for $\lambda\in[-1,1]$,
which is clearly convex. For $T\ge 1$, suppose
$A,B\in\mathcal{A}_T$ and let $\lambda\in(0,1)$. Then
$\cost(\lambda A+(1-\lambda)B)$ is either $0$, $1$,
or $\log\gamma_2(\lambda A+(1-\lambda)B)$. In the former
two cases, we clearly have $\lambda A+(1-\lambda)B\in\mathcal{A}_T$,
so consider the latter case. We have
$\log\gamma_2(\lambda A+(1-\lambda)B)
\le \log(\lambda\gamma_2(A)+(1-\lambda)\gamma_2(B))
\le \log\max\{\gamma_2(A),\gamma_2(B)\}
=\max\{\log\gamma_2(A),\log\gamma_2(B)\}
\le \max\{\cost(A),\cost(B)\}
\le T$. Hence $\mathcal{A}_T$ is convex. It is also
clear that $\bias_F(\cdot,(x,y))$
is linear, so \itm{convex} is satisfied.

By taking $A$ to equal $F$ inside $\Dom(F)$ and to be $0$
elsewhere, we get $\bias_F(A)=1$, so \itm{finite}
is satisfied. By our definition of $\cost(A)$, we have
$\cost(A)\ge 1$ or $\cost(A)=0$, with the latter happening
only if $A$ is a convex combination of $J$ and $-J$,
so \itm{zero_cost} is satisfied.

As usual, it remains to handle \itm{amplify}.
We do so in the following lemma.

\begin{lemma}
Let $P$ a probability distribution over matrices
$A_1,A_2,\dots,A_k$ with probability $p_i$ for $A_i$.
Suppose that $\sum_{i=1}^k 2^ip_i\le 5$, and that
for all $i$, we have $\cost(A_i)\le 2^i T$
for some real number $T\ge 1/10$.
Suppose further that $\bias_F(P)\ge 2^{-k-1}$.
Then there is some bounded matrix $A$ which approximates $F$
to bias $1/3$ and satisfies $\cost(A)\le 2^kT\cdot\poly(k)$
(with the constants in the $\poly$ being universal).
\end{lemma}

\begin{proof}
Let $\rho$ be the polynomial from the proof of
\thm{minimax_polynomials} with respect to the
probabilities $p_1,p_2,\dots,p_k$. This is a polynomial
in $k$ variables with the property that if
values $\beta_1,\beta_2,\dots,\beta_k$ are plugged in
and $\left|\sum_i p_i\beta_i\right|\ge 2^{-k-1}$,
then $\rho(\beta_1,\beta_2,\dots,\beta_k)$ returns the
sign of $\sum_i p_i\beta_i$ to bounded error.
The polynomial $\rho$ further has the property that it is bounded
(i.e.\ it returns values in $[-1,1]$ when given inputs in
$[-1,1]^k$), and that if you plug in any
polynomials $q_i$ in place of $\beta_i$, with $\deg(q_i)\le 2^i$,
then the degree of the composed polynomial is at most $2^k\poly(k)$.

This latter property means that the \emph{weighted degree}
of $\rho$ with weights $(2^1,2^2,\dots,2^k)$ is at most $O(2^k\poly(k))$.
Here the term weighted degree means that we count the degree
of each monomial of $\rho$ differently depending on the variables
in that monomial: the $i$-th variable gets weight $2^i$,
so a monomial of the form
$\beta_1^{c_1}\beta_2^{c_2}\dots\beta_k^{c_k}$
will have weighted degree $2^ic_1+2^2c_2+\dots+2^kc_k$.
We know that the weighted degree of $\rho$, meaning the maximum
weighted degree of one of its monomials, is at most
$O(2^k\poly(k))$.

We will now use this polynomial $\rho$ to construct
a matrix $A$ which approximates $F$ and has $\gamma_2$ norm
that is not too large. The idea is to simply apply
$\rho$ to the matrices $A_1,A_2,\dots,A_k$, using
the Hadamard product for multiplication and the usual matrix
addition and scalar multiplication.
Since $\gamma_2$ is a norm, we know that
$\gamma_2(\rho(A_1,A_2,\dots,A_k))$
is the sum, over all monomials of $\rho$, of the absolute
value of the coefficient
of that monomial multipled by the $\gamma_2$-norm
of the Hadamard product defined by that monomial.
This is upper bounded by the sum of absolute coefficients
of $\rho$ (which we'll denote $C$) multiplied
by the $\gamma_2$ norm of the largest monomial.

The $\gamma_2$ norm of a single monomial
$\beta_1^{c_1}\dots\beta_k^{c_k}$ composed with matrices
$A_1,\dots, A_k$ is at most
$\gamma_2(A_1)^{c_1}\dots\gamma_2(A_k)^{c_k}$, since
the $\gamma_2$ norm is sub-multiplicative under the Hadamard
product. Hence $\log\gamma_2(\rho(A_1,\dots A_k))$
is at most $\log C$ plus the maximum value of
$c_1\log\gamma_2(A_1)+\dots+c_k\log\gamma_2(A_k)$
for some monomial $(c_1,c_2,\dots,c_k)$ of $\rho$.
Since $\log\gamma_2(A)\le\cost(A)$ for all bounded matrices $A$,
and since $\cost(A_i)\le 2^i T$, this maximum is at most
the maximum of $T\cdot(2^1c_1+\dots+2^kc_k)$ over monomials
of $\rho$, which is at most $O(2^k T\poly(k))$.

We now upper bound $C$, the sum of absolute coefficients of $\rho$.
Recall that $\rho$ was constructed as an average of different
polynomials with different values of the constants $y_i$.
Let $\rho'$ be the polynomial within that set we averaged
over which has the largest sum of absolute coefficients.
Then to upper bound $C$ it suffices to upper bound
the sum of absolute coefficients of $\rho'$.
To do so, we essentially want to replace all coefficients
of $\rho'$ with their absolute values, and then plug in all ones
for the variables. We note that $(9/10)+y_i$ will be at most
$1$ for the values of $y_i$ used in $\rho'$, which means
that if we replace the terms $(9/10)q_i+y_i$ with
simply $q_i$, we would only increase the sum of absolute
coefficients (here we treat $q_i$ as variables).

Let the resulting polynomial be $\rho''$.
Then $\rho''$ is simply the result of composing the polynomial
$s$ with the polynomials $r_{i,t}$. Since $s$
is a bounded multilinear polynomial of degree $O(k^2)$,
its sum of absolute coefficients is at most $2^{O(k^2)}$,
and it is not hard to see that the sum of absolute coefficients
of $\rho''$ will be at most $2^{O(k^2)}$ times
$D^{O(k^2)}$, where $D$ is the maximum sum of absolute
coefficients over the polynomials $d_t^\ell$ with $\ell=k-i+O(\log k)$.
In other words, $\log C\le O(k^2)+O(k^2 D)$,
where $D$ is the sum of absolute coefficients of some such polynomial
$d_t^\ell$.

The polynomial $d_t^\ell$ is a single variate bounded polynomial
of degree at most $O(2^\ell\log k)$, which, using
$\ell\le k+O(\log k)$,
is at most $2^{k}\poly(k)$. A bounded univariate
polynomial of this degree
must have sum of absolute coefficients at most $4^{2^k\poly(k)}$
by \cite{She13} (Lemma 4.1). Hence $\log D\le 2^k\poly(k)$, so $\log C\le 2^k\poly(k)$.

We conclude that if $A=\rho(A_1,A_2,\dots,A_k)$,
then $\log\gamma_2(A)\le 2^k(T+1)\poly(k)$,
and hence $\cost(A)\le 2^k(T+1)\poly(k)$. This is at most
$O(2^kT\poly(k))$ since we have $T\ge 1/10$.
Further, each entry $A[x,y]$ is equal to
$\rho(A_1[x,y],A_2[x,y],\dots,A_k[x,y])$,
which means that $A$ is bounded (since $\rho$ is bounded
and the matrices $A_i$ are bounded), and
for $(x,y)\in\Dom(F)$, we have $F(x,y)A[x,y]\ge 1/3$
by the guarantees on $A_i$ and on $\rho$.
\end{proof}

Using \thm{quantum_abstraction}, we can now conclude the following
theorem.

\begin{theorem}
Let $F\colon\mathcal{X}\times\mathcal{Y}\to\{+1,-1\}$ be a
(possibly partial) communication function. Then there is a distribution
$\mu$ over $\Dom(F)$ such that for any bounded real matrix $A$
(meaning $|A[x,y]|\le 1$ for all 
$(x,y)\in\mathcal{X}\times\mathcal{Y}$),
we have
\[\bE_{(x,y)\sim\mu}[F(x,y)A[x,y]]
\le\frac{\log\gamma_2(A)}{\tOmega(\log\tilde{\gamma}_2(F))}.\]
\end{theorem}

Note that for bounded matrices, $\log\gamma_2(A)\le\log\rank(A)$.
We also have, from \cite{LS09},
\[\log\widetilde{\rank}(F)\le
6\log\tilde{\gamma}_2(F)+O(\log\log|\mathcal{X}\times\mathcal{Y}|).\]
This means we can write a minimax theorem for logrank as well.

\begin{theorem}\label{thm:minimax_logrank}
Let $F\colon\mathcal{X}\times\mathcal{Y}\to\{+1,-1\}$ be a
(possibly partial) communication function,
and suppose that $\log\widetilde{\rank}(F)
\ge C\log\log|\mathcal{X}\times\mathcal{Y}|$ where $C$ is a universal
constant. Then there is a distribution
$\mu$ over $\Dom(F)$ such that for any bounded real matrix $A$
(meaning $|A[x,y]|\le 1$ for all 
$(x,y)\in\mathcal{X}\times\mathcal{Y}$),
we have
\[\bE_{(x,y)\sim\mu}[F(x,y)A[x,y]]
\le\frac{\log\rank(A)}{\tOmega(\log\widetilde{\rank}(F))}.\]
\end{theorem}

In other words, $\mu$ is such that if $A$ has low rank
compared to $F$, then $A$ cannot correlate well with $F$
under $\mu$, and hence $A$ does not approximate $F$ very well
against $\mu$.

\section{Circuit complexity}
\label{sec:circuits}

A Boolean circuit $C$ is a collection of gates connected to each other and to bits of its input $x$ by wires, with a single output wire representing the value of $C(x)$. The \emph{size} of a circuit is the number of gates in the circuit, and the \emph{depth} of a circuit is the length of the longest path between an input bit and an output wire. A \emph{randomized Boolean circuit} is a probability distribution over Boolean circuits, and the \emph{size} of a randomized Boolean circuit is defined to be the expected size of a Boolean circuit drawn from that distribution.

In \sec{general}, we examine the randomized circuit complexity of partial Boolean functions when it is computed by circuits of unbounded fan-in and unlimited depth. In \sec{nc1}, we show that the main result also holds in the NC$^1$ setting of logarithmic-depth circuits whose gates each have fan-in at most $2$. Finally, in \sec{hardcore} we establish the strengthening of the hardcore lemma.

\subsection{General circuits}
\label{sec:general}

In this section, let $\R(f)$ denote the minimum size of a randomized Boolean circuit of unbounded fan-in and unlimited depth that computes the partial Boolean function $f$ with error at most $\frac13$ on every input $x \in \Dom(f)$. Similarly, let $\R^\mu_{\dot{\gamma}}(f)$ denote the minimum size of randomized Boolean circuits that compute $f$ with error at most $\dot{\gamma} = \frac{1-\gamma}2$ when the input is drawn from $\mu$. We establish a relation between those two complexity measures via the study of forecasting circuits.

\begin{definition}
A \emph{forecasting circuit} is a randomized Boolean circuit with one modification: instead of having a single output wire, the forecasting circuit has $k+1$ output wires that represent the binary encoding of a value in the range $\{0,\frac1{2^k},\frac2{2^k},\ldots,\frac{2^k-1}{2^k}, 1\}$.
\end{definition}

The \emph{resolution} of a forecasting circuit is $k$ when it has $k+1$ output wires. (Or, equivalently, when it outputs values that are multiples of $2^{-k}$.)
The \emph{score} of a forecasting circuit is computed in the same way as we did for forecasting algorithms in previous sections.
The \emph{size} of a randomized forecasting circuit is, as in the the case of randomized Boolean circuits, the expected number of gates in a circuit drawn from the distribution.
Forecasting circuits can be defined for each model of randomized Boolean circuits; in this section, we consider forecasting circuits with unbounded fan-in and unlimited depth.

We begin by showing that if there is a Boolean circuit that computes a function with non-negligible advantage over random guessing, then there is also a forecasting algorithm with non-trivial score.

\begin{proposition}
\label{prop:bias-to-score-circuits}
For any partial function $f : \{0,1\}^n \to \{0,1\}$, if there is a size $s \ge 1$ and parameter $\gamma \ge \frac4{R(f)+1}$ for which there is a randomized Boolean circuit $R$ of average size $s$ that satisfies
$\Pr_{C \sim R}[ C(x) \neq f(x) ] \le \dot{\gamma}$ for every $x \in \Dom(f)$, then there is also a randomized forecasting circuit $R'$ with resolution $\left\lceil \log R(f)\right\rceil$, average size at most $s+1$, and $h$-score
\[
\score(R',x) = \E_{C' \sim R'}[ \score(C'(x),f(x))] \ge \gamma^2/8
\]
for each $x \in \Dom(f)$.
\end{proposition}

\begin{proof}
For each circuit $C$ in the support of $R$, define $C'$ to be the forecasting circuit of resolution $k = \left\lceil \log R(f) \right\rceil$ and size $\mathrm{size}(C) + 1$ which outputs the value 
\[
\frac{1 + (-1)^{C(x)}\gamma'}2
\]
on input $x \in S$ where $\gamma' = \frac{2m}{2^k}$ for the largest integer $m$ such that $\gamma' \le \gamma$. The definition of $\gamma'$ guarantees that $\gamma - \frac2{2^k} \le \gamma' \le \gamma$. The value of $k$ and the lower bound on $\gamma$ in the proposition statement imply that $\gamma - \frac2{2^k} \ge \gamma - \frac2{R(f)+1} \ge \frac{\gamma}2$, so
$\frac{\gamma}2 \le \gamma' \le \gamma$.

This circuit $C'$ can be constructing by adding a single extra $\neg$ gate to the output wire of $C$: the output of $C$ and its negations can then be combined with constant value wires to generate the two required output values of the forecasting circuit.
(Namely, if the two output values $(1 \pm \gamma')/2$ of $C'$ are denoted by $z^{(0)}$ and $z^{(1)}$, then the $i$th output bit of $C'$ is a hardcoded  constant value $0$ or $1$ when $z^{(0)} = z^{(1)}$ and otherwise it is either $C(x)$ or $\neg C(x)$ when $z^{(0)} \neq z^{(1)}$.)

The randomized forecasting circuit $R'$ is then defined to be the distribution on circuits obtained by drawing $C \sim R$ and outputing the modified circuit $C'$ as described above. Following the same argument as in \lem{conversion}, 
the score of this randomized forecasting circuit satisfies
\[
\score(R',x) \ge \gamma'^2/2 \ge \gamma/8. \qedhere
\]
\end{proof}

In the second step in the proof of \thm{main-circuit}, we show that the minimax theorem applies in this setting.

\begin{lemma}
\label{lem:minimax-circuits}
Fix any $k \ge 1$ and let $\mathcal{R}_k$ denote the set of all randomized forecasting circuits with resolution $k$. Then for partial function $f : \{0,1\}^n \to \{0,1\}$, if we let $\Delta$ denote the set of distributions over $\Dom(f)$, we have
\[
\inf_{R \in \mathcal{R}_k} \max_{x \in \Dom(f)} \frac{\mathrm{size}(R)}{\score(R,x)^+} = \max_{\mu \in \Delta} \inf_{R \in \mathcal{R}_k} \frac{\mathrm{size}(R)}{\score(R,\mu)^+}.
\]
\end{lemma}

\begin{proof}
The lemma follows from \thm{main_minimax}, and the argument showing that the conditions of that theorem are satisfied follows closely the analogous argument of \thm{randomized_query_minimax}.

First, we want to show that $\mathcal{R}_k$ can be viewed as a convex subset of a real topological space $V$. We can do so with the same construction as in \thm{randomized_query_minimax}, though here we can also use a slightly simpler constructions: fix $V = \mathbb{R}^{|\Dom(f)|+1}$, and for each randomized forecasting circuit $R \in \mathcal{R}_k$ define $v_R(x) = \score(R,x)$ for each $x \in \Dom(f)$ and define the $|\Dom(f)|+1$th coordinate of $v_R$ to be $\cost(R,x)$. That the resulting set is convex follows directly from the fact that a vector $v' = \lambda v_{R_1} + (1-\lambda)v_{R_2}$ for any $R_1,R_2 \in \mathcal{R}_k$ corresponds to the vector of the randomized forecasting circuit $R' = \lambda R_1 + (1-\lambda) R_2$.

The linearity of cost and score measures in both $R$ and $\mu$ follows from their definition.

Lastly, the notions of cost and score satisfy the well-behaved condition of \thm{main_minimax}. First, because the existence of a circuit of size at most $2^{|S|}$ that computes $f$ exactly implies the existence of a finite-cost and non-zero score randomized forecasting circuit for any distribution $\mu$ on $\Dom(f)$. Second, because the cost of circuits does not depend on the input, and third because the definition of cost immediately implies that the mixture of a zero-cost and a nonzero-cost randomized circuit gives a nonzero-cost randomized circuit.
\end{proof}

The next step is the main one in the proof of the theorem: we want to show that the score of forecasting circuits can be amplified efficiently.

\begin{lemma}
\label{lem:amplification-circuits}
For every partial Boolean function $f : \{0,1\}^n \to \{0,1\}$, when we set $k = \lceil \log R(f) \rceil$ then
\[
\inf_{R \in \mathcal{C}_k} \max_{x \in \Dom(f)} \frac{ \mathrm{size}(R) }{ \score(R,x)^+} = \widetilde{\Omega}\big( R(f) \big)
\]
where the $\widetilde{\Omega}$ hides terms that are polylogarithmic in $R(f)$.
\end{lemma}

The proof of the lemma uses the following bounds regarding the circuit complexity of basic arithmetic operations.

\begin{proposition}[\cite{BCH86,Alt88}]
\label{prop:elementary-functions}
For any two numbers $a, b$ represented to accuracy $2^{-k}$ in binary, then the values
\[
a b, \qquad \frac{a}{1-a}, \qquad \ln(a), \qquad e^a \mbox{, and} \qquad \frac{1}{1+a}
\]
can all be computed to additive accuracy $2^{-k}$ by circuits of size polynomial in $k$ and depth $O(\log k)$.
\end{proposition}

We also need another result regarding the circuit complexity of iterated multiplication up to fixed accuracy.

\begin{proposition}
\label{prop:iterated-addition}
When $a_1,\ldots,a_m$ and $b_1,\ldots,b_m$ are $k$-bit integers, then there is a circuit of size $O(m \log m + mk + k^c)$ for some constant $c \ge 1$ 
and depth $O(\log k + \log m)$ 
that computes the ratio 
\[
\frac{a_1 \cdots a_m}{b_1 \cdots b_m}
\]
up to multiplicative accuracy $1 \pm 2^{-k}$.
\end{proposition}

\begin{proof}
This result can be obtained by computing $\ln \frac{a_1 \cdots a_m}{b_1 \cdots b_m} = \sum_{i=1}^m \ln a_i - \ln b_i$ to additive accuracy $2^{-k}$.
The computation of each of the values $\ln a_i$ and $\ln b_i$ for $1 \le i \le m$ up to additive accuracy $\frac{2^{-k}}{2m}$ can be done with a circuit of size polynomial in $n := k + \log m + 1$ and depth $O(\log n)$. 
The sum of the $2m$ terms can be done with a circuit for iterated addition of size $O(m n) = O(m \log m + mk)$ and depth $O(\log m + \log n) = O(\log m + \log k)$ to compute the natural log of the ratio up to additive error $2^{-k}$.~\cite{Ofm62} (See also~\cite{Pip87,Weg87} and the references therein.)
Finally, a circuit of size polynomial in $k$ and depth logarithmic in $k$ can be used to compute the exponential of the final ratio.
\end{proof}

Using these propositions, we can complete the proof of the lemma.

\begin{proof}[Proof of Lemma~\ref{lem:amplification-circuits}]
Let $R$ be a randomized forecasting circuit which comes arbitrarily close to the infimum on the left-hand side.

Consider the randomized forecasting circuit $R'$ obtained by drawing $m$ forecasting circuits $C_1,\ldots,C_m$ independently at random from $R$ and combining their output values using the formula
\[
C'(x) = \left( 1 + \prod_{i \le m} \frac{1-C_i(x)}{C_i(x)} \right)^{-1}.
\]
Fixing $m = \max_x 1/\score(R,x)^+$, we obtain a randomized circuit $R'$ with score $\score(R',x)^+ = \Omega(1)$ for each $x \in S$ and average size 
\[
\mathrm{size}(R') = \mathrm{size}(R) \cdot m + O(m \log m + mk + k^c)
\]
for some universal constant $c \ge 1$.
Then the proof is completed by noting that $m = O(\frac{\R(f)}{\mathrm{size}(R)})$.
\end{proof}

Finally, we show that when there is a forecasting circuit with score $\gamma$, there is also a Boolean circuit with error at most $\dot{\gamma}$.

\begin{proposition}
\label{prop:score-to-bias-circuits}
For any partial function $f : \{0,1\}^n \to \{0,1\}$, if there is a size $s \ge 1$ and parameter $\gamma$ for which there is a randomized forecasting circuit $R$ with $k$ output wires, size $s$, depth $d$ and $\score(R,x) \ge \gamma$ for each $x \in \Dom(f)$, then there is also a randomized Boolean circuit $R'$ of size $s + O(k)$ and depth $d + O(1)$ that satisfies
$\Pr_{C \sim \mathcal{R}}[ C(x) = f(x) ] \ge \frac{1+\gamma}{2}$ for every $x \in \Dom(f)$.
\end{proposition}

\begin{proof}
Given a forecasting circuit $C$ that outputs the value $p$ on input $x$, we want to design a randomized Boolean circuit $R_C$ that outputs the value $1$ with probability $p$ and $0$ with probability $1-p$ on input $x$.

We can do this by adding $k$ random inputs $r_1,\ldots,r_k$ that are used to generate a uniformly random value $r \in \{\frac1{2^k},\frac{2}{2^k},\ldots,1\}$. Then if the value $p$ in the output of the circuit is $0$, we output zero; otherwise we use a comparator circuit to return $1$ if and only if $r \le p$. This value has the desired bias $p$, and using standard constructions (see, e.g.~\cite{Weg87,Vol99}) we can implement the comparator circuit with $O(k)$ gates in a circuit of constant depth (in the unbounded fan-in model; or $O(\log n)$ depth in the bounded fan-in model).

The final randomized Boolean circuit $R'$ is defined by drawing a forecasting circuit $C$ from $R$ and outputting $R_C$. The bound on the error of $R'$ is then obtained as in the argument of \lem{conversion}.
\end{proof}

Putting the above lemmas and propositions together completes the proof of the following theorem.

\begin{theorem}
\label{thm:main-circuit}
Fix $n \in \bN$. For every partial function $f : \{0,1\}^n \to \{0,1\}$, there is a distribution $\mu$ on $\Dom(f)$ such that for all $\gamma \in (0,1]$,
\[
\R_{\dot{\gamma}}^\mu(f) = \tOmega\big( \gamma^2 \R(f) \big).
\]
\end{theorem}



\subsection{Circuits with bounded depth}
\label{sec:nc1}

Define $\RNC(f)$ to be the minimum size of a randomized Boolean circuit of fan-in two and logarithmic depth that computes the partial Boolean function $f$ with error at most $\frac13$ on every input $x \in \Dom(f)$. Similarly, let $\RNC^\mu_{\dot{\gamma}}(f)$ denote the minimum size of a randomized Boolean circuit with the same fan-in and depth restrictions that computes $f$ with error at most $\dot{\gamma} = \frac{1-\gamma}2$ when the input is drawn from $\mu$. 

The constructions of \prop{bias-to-score-circuits}, \lem{amplification-circuits}, and \prop{score-to-bias-circuits} can all be achieved with circuits of fan-in 2 that add only logarithmic depth overhead to the base circuits, so the analogue of \thm{main-circuit} also holds for the class of circuits of fan-in two and logarithmic depth.

\begin{theorem}
\label{thm:nc1-circuit}
Fix $n \in \bN$. For every partial function $f : \{0,1\}^n \to \{0,1\}$, there is a distribution $\mu$ on $\Dom(f)$ such that for all $\gamma \in (0,1]$,
\[
\RNC_{\dot{\gamma}}^\mu(f) = \tOmega\big( \gamma^2 \RNC(f) \big).
\]
\end{theorem}

In fact, we can say even more about the efficiency of the transformations in each constructions: all three of them can be accomplished with constant-depth and polynomial-size overhead when the circuits have threshold gates. For \prop{bias-to-score-circuits}, this is because only a single additional gate is required. For \lem{amplification-circuits}, this is because the functions in \prop{elementary-functions} can all be computed the the required accuracy with threshold circuits of polynomial size and constant depth~\cite{RT92} and the iterated addition problem can be solved by a threshold circuit of constant depth and size $O(m \log m (k + \log m))$~\cite{CSV84}. And for \prop{score-to-bias-circuits}, this is because comparison can also be completed with polynomial-size and constant-depth circuits. Therefore, letting $\RTC_\epsilon(f)$ denote the minimum size of a randomized constant-depth threshold circuit with unbounded fan-in that computes $f$ with error probability at most $\frac13$ on every input and $\RNC_{\dot{\gamma}}^\mu(f)$ denote the minimum size of the same type of circuit that computes $f(x)$ correctly with probability $\frac{1+\gamma}2$ when $x$ is drawn from $\mu$, we obtain the following result.

\begin{theorem}
\label{thm:tc0-circuit}
Fix $n \in \bN$. For every partial function $f : \{0,1\}^n \to \{0,1\}$, there is a distribution $\mu$ on $\Dom(f)$ such that for all $\gamma \in (0,1]$,
\[
\RTC_{\dot{\gamma}}^\mu(f) = \tOmega\big( \gamma^2 \RTC(f) \big).
\]
\end{theorem}

\subsection{Hardcore lemma}
\label{sec:hardcore}

In order to complete the proof of the hardcore lemma as stated in \thm{univ-hardcore}, we need the following variant of the ratio minimax theorem that applies to the setting where we consider a compact convex set of distributions, not just the set of all distributions over the function's domain.

\begin{theorem}
Let $V$ be a real topological vector space, and let
$\mathcal{R}\subseteq V$ be convex. Let $S$
be a nonempty finite set, and let $\Delta$ be a compact and convex
set of probability distributions over $S$,
viewed as a subset of $\bR^{|S|}$. Let
$\cost\colon\mathcal{R}\times\Delta\to[0,\infty]$
be semicontinuous and saddle,
and let $\score\colon\mathcal{R}\times\Delta\to[-\infty,\infty)$
be such that its negation, $-\score$,
is semicontinuous and saddle.
Suppose $\cost$ and $\score$ are well-behaved. Then
using the convention $r/0=\infty$ for all $r\in[0,\infty]$,
we have
\[\adjustlimits\inf_{R\in\mathcal{R}}\max_{\mu\in\Delta}
\frac{\cost(R,\mu)}{\score(R,\mu)^{+}}
=\adjustlimits\max_{\mu\in\Delta}\inf_{R\in\mathcal{R}}
\frac{\cost(R,\mu)}{\score(R,\mu)^{+}}.
\]
\end{theorem}

\begin{proof}
The proof is identical to the one for (the first part of) \thm{main_minimax}, since that argument only uses the fact that the set of all distributions over $S$ is convex and compact.
\end{proof}

From this theorem we obtain the following variant of \lem{minimax-circuits} for distributions with min-entropy $\delta$.

\begin{lemma}
\label{lem:minimax-circuits-regular}
Fix any $k \ge 1$ and let $\mathcal{R}_k$ denote the set of all randomized forecasting circuits with resolution $k$. Then for every $\delta > 0$ and function $f : \{0,1\}^n \to \{0,1\}$, if we let $\Delta_{\delta}$ denote the set of distributions over $\{0,1\}^n$ with min-entropy $\delta$, we have
\[
\inf_{R \in \mathcal{R}_k} \max_{\mu \in \Delta_{\delta}} \frac{\mathrm{size}(R)}{\score(R,\mu)^+} = \max_{\mu \in \Delta_{\delta}} \inf_{R \in \mathcal{R}_k} \frac{\mathrm{size}(R)}{\score(R,\mu)^+}.
\]
\end{lemma}

We are now ready to complete the proof of \thm{univ-hardcore}.

\begin{proof}[Proof of Theorem~\ref{thm:univ-hardcore}]
Fix $s' = c \cdot s/\log \frac1\delta$ for some constant $c$ to be fixed later.
By \lem{minimax-circuits-regular}, the two cases to consider are the following.

\paragraph{Case 1:} $\max_{\mu \in \Delta_{\delta}} \inf_{R \in \mathcal{R}_k} \frac{\mathrm{size}(R)}{\score(R,\mu)^+} \ge s'$.

Fix a distribution $\mu$ with min-entropy $\delta$ for which the maximum is attained. Then every randomized forecasting circuit $R$ has score
\[
\score(R,\mu) \le \frac{\mathrm{size}(R)}{s'}. 
\]
By \prop{bias-to-score-circuits}, this implies that every randomized circuit $R'$ with $\mathrm{size}(R') \le \epsilon^2 s'/64$ has success probability 
\[
\Pr_{C \sim R, x \sim \mu}[ C(x) = f(x) ] \le
\frac{1 + 8\sqrt{\mathrm{size}(R')/s'}}{2} \le \frac{1+\epsilon}{2},
\]
and the theorem holds in this case.

\paragraph{Case 2:} $\inf_{R \in \mathcal{R}_k} \max_{\mu \in \Delta_{\delta}} \frac{\mathrm{size}(R)}{\score(C,\mu)^+} < s'$.

Fix a randomized forecasting circuit $R$ that satisfies
\[
\frac{\mathrm{size}(R)}{\score(C,\mu)^+} < s'
\]
for each distribution $\mu$ over $\{0,1\}^n$ with min-entropy $\delta$. 
Set $\alpha = \mathrm{size}(R)/s'$ and define $T \subseteq \{0,1\}^n$ to be the set of inputs $x$ for which
$\score(R,x) < \frac{\alpha}{2}$.
Then 
\[
|T| \le \delta(1-\tfrac \alpha2) 2^n
\]
since otherwise the score of $R$ on the distribution $\mu'$ that is uniform over any set $T' \supseteq T$ of size $|T'| = \delta 2^n$ (and thus has min-entropy $\delta$) would be bounded above by 
$\score(R,\mu') < (1-\tfrac\alpha2)\cdot \tfrac{\alpha}2 + \tfrac\alpha2 < \alpha$,
contradicting the definition of $R$.

By \lem{amplification-circuits}, there is a forecasting circuit $R'$ which satisfies 
$\mathrm{size}(R') = O(s')$ 
and
$\score(R,x) = \Omega(1)$ for each $x \in \{0,1\}^n \setminus T$. 
Then by \prop{score-to-bias-circuits} there is a randomized Boolean circuit of size $O(s')$ that errs with probability at most $\frac13$ on each $x \in \{0,1\}^n \setminus T$, and by standard success amplification it also means that there is a circuit $C$ of size $s'' = O(s' \log \frac1\delta)$ with error less than $\delta$. Choosing the value $c$ in the definition of $s'$ appropriately, we then get that this circuit has size at most $s$, contradicting the premise of the theorem and therefore showing that Case 2 cannot occur.
\end{proof}

\end{fulltext}

\section*{Acknowledgements}

We thank Justin Thaler for discussions and references related
to approximate polynomial degree and its amplification.
We also thank Andrew Drucker, Mika G{\"o}{\"o}s, and Li-Yang Tan for correspondence about their ongoing work~\cite{BDG+20}.
We thank anonymous reviewers for many helpful comments.

\iffocs
\else

\newpage
\appendix

\section{Proofs related to the minimax theorem}\label{app:minimax}

\attain*

\begin{proof}
The lower semicontinuous case follows from the upper semicontinuous
case simply by negating $\phi$, so we focus on the upper semicontinuous
case. Let $z=\sup_{x\in X}\phi(x)$, where $z\in\bar{\bR}$.
Let $x_0$ be any element of $X$. If $\phi(x_0)=z$, we are done,
so assume $\phi(x_0)<z$; in particular, $z>-\infty$.
We define a sequence $x_1,x_2,\dots$ as follows. If $z<\infty$,
define $x_{i}$ to be any element of $X$ such that
$\phi(x_i)>z-1/i$. If $z=\infty$, define $x_i$ to be any element of $X$
such that $\phi(x_i)>i$. Moreover, for each $i\in\bN$, let
$U_i=\{\,x\in X:\phi(x)<\phi(x_i)\}$. Note that any $x\in X$
for which $\phi(x)<z$ must be in $U_i$ for some $i\in\bN$;
hence if the supremum $z$ is not attained, the sets
$U_i$ form a cover for $X$ (meaning $\bigcup_{i\in \bN}U_i=X$).

The key claim is that the $U_i$ sets are all open if $\phi$
is upper semicontinuous. This is is because if $x\in U_i$,
then $\phi(x)<\phi(x_i)$, and by the definition of upper semicontinuity,
there is a neighborhood $U$ of $x$ on which $\phi(\cdot)$ is still
less than $\phi(x_i)$; thus there is a neighborhood $U$ of $x$
contained in $U_i$, so that $U_i$ is open. In this case,
if the supremum $z$ is not attained, the collection $\{U_i\}_{i\in\bN}$
is an open cover of $X$, and by the definition of compactness, it
has a finite subcover. Let $i$ be the largest index of some $U_i$
in this subcover. Then it follows that $\phi(x)<\phi(x_i)$ for all
$x\in X$, which is a contradiction. Hence the supremum $z$ must
be attained as a maximum, as desired.
\end{proof}

\semicontinuousinf*

\begin{proof}
Note that the case where $\phi_i$ are all lower semicontinuous
follows from the case where they are all upper semicontinuous simply
by negating the functions, since negation flips upper and lower
semicontinuity and flips infimums and supremums. We focus
on the case where $\phi_i$ are all upper semicontinuous.

Fix $x\in X$. If $\phi(x)=\infty$,
$\phi$ is upper semicontinuous at $x$ by definition.
If $\phi(x)<\infty$, fix any $y>\phi(x)$. By the definition
of $\phi(x)$ as an infimum, there is some
$i\in I$ such that $\phi_i(x)<y$.
By the upper semicontinuity of $\phi_i(\cdot)$, there is a
neighborhood $U$ of $x$ such that
for all $x'\in U$, we have $\phi_i(x')<y$. Then for all $x'\in U$,
we clearly have $\phi(x')=\inf_{i\in I}\phi_i(x')<y$. Thus
$\phi$ is upper semicontinuous at $x$, as desired.
\end{proof}

\begin{lemma}\label{lem:convex_characterization}
Let $V$ be a real vector space, and let $X\subseteq V$.
The convex hull of $X$ is the set of all $v\in V$
which can be written as a convex combinatotion
of vectors in $x$; that is, $v$ for which there exist
$k\in\bN$, $x_1,x_2,\dots,x_k\in X$, and
$\lambda_1,\lambda_2,\dots,\lambda_k\in[0,1]$ with
$\lambda_1+\lambda_2+\dots+\lambda_k=1$ such that
$v=\lambda_1 x_1+\lambda_2 x_2+\dots+\lambda_k x_k$.
\end{lemma}

\begin{proof}
This is a well-known characterization of the convex hull,
which can be shown as follows:
let $Y$ be the set of all finite convex combinations
of points in $X$; that is, $Y$ contains all points in $V$ of the form
$\lambda_1 x_1+\lambda_2 x_2+\dots+\lambda_k x_k$, where
$k\in\bN$, $x_1,x_2,\dots,x_k\in X$, and $\lambda_1,\lambda_2,\dots,
\lambda_k\in[0,1]$ with $\lambda_1+\lambda_2+\dots+\lambda_k=1$.
Then $Y$ is clearly convex, since for all $y_1,y_2\in Y$
and $\lambda\in(0,1)$, we know that $y_1$ and $y_2$
are finite convex combinations of points in $x$, meaning that
$\lambda y_1+(1-\lambda)y_2$ is also a finite convex combination
of points in $X$. Furthermore, if $Z$ is any other convex
set containing $X$, then it's easy to show by induction that
$Z$ contains all convex combinations of $k$ points in $X$
for each $k\in\bN$; hence $Z$ must be a superset of $Y$.
It follows that $\Conv(X)$, the intersection of all
convex sets containing $X$, must exactly equal $Y$.
\end{proof}

\quasiconvexhull*

\begin{proof}
The quasiconcave case follows from the quasiconvex case by negating
$\phi$; hence it suffices to prove the quasiconvex case.
It is clear that $\sup_{x\in\Conv(X)}\phi(x)$ is at least
$\sup_{x\in X}\phi(x)$, so we only need to show the latter
is at least the former. To this end, let
$y^*\coloneqq \sup_{x\in\Conv(X)}\phi(x)$,
and let $\hat{x}\in\Conv(X)$
be such that $\phi(\hat{x})$ is arbitrarily close to $y^*$.
We must show
that $\sup_{x\in X}\phi(x)\ge\phi(\hat{x})$, or equivalently,
that there is some $x\in X$ with $\phi(x)\ge\phi(\hat{x})$.

Using \lem{convex_characterization}, we can now write
$\hat{x}\in\Conv(X)$ as
$\hat{x}=\lambda_1 x_1+\lambda_2 x_2+\dots+\lambda_k x_k$,
with $k\in\bN$, $x_1,x_2,\dots, x_k\in X$, and
$\lambda_1,\lambda_2,\dots,\lambda_k\in[0,1]$ with
$\lambda_1+\lambda_2+\dots+\lambda_k=1$. Furthermore,
assume that $\lambda_i>0$ for each $i\in[k]$ (we can remove
$\lambda_i x_i=0$ from the linear combination otherwise).
Now, note that by quasiconvexity, we have
$\phi(\lambda x_1+(1-\lambda)x_2)\le\max\{\phi(x_1),\phi(x_2)\}$.
It is not hard to show by induction that
$\phi(\lambda_1 x_1+\lambda_2 x_2+\dots+\lambda_k x_k)
\le\max\{\phi(x_1),\phi(x_2),\dots,\phi(x_k)\}$.
Hence there is some $x\in X$ such that $\phi(x)\ge\phi(\hat{x})$,
as desired.
\end{proof}

\maxpreserves*

We actually prove a stronger statement, where the maximum
is taken with an arbitrary constant.

\begin{lemma}\label{lem:max_preservations}
Let $V$ be a real topological vector space, and let $X\subseteq V$
be convex. Let $\psi\colon X\to\bar{\bR}$ be a function,
let $c\in\bR$ be a constant, and let $\psi'\colon X\to\bar{\bR}$
be the function $\psi'(x)=\max\{\psi(x),c\}$.
Then if $\psi$ is convex, $\psi'$ is convex; if
$\psi$ is quasiconvex, $\psi'$ is quasiconvex;
if $\psi$ is quasiconcave, $\psi'$ is quasiconcave;
if $\psi$ is upper semicontinuous, $\psi'$ is upper semicontinuous;
and if $\psi$ is lower semicontinuous, $\psi'$ is lower semicontinuous.
\end{lemma}

\begin{proof}
Let $x,y\in X$, and let $\lambda\in(0,1)$. Then
\[\psi'(\lambda x+(1-\lambda)y)=\max\{\psi(\lambda x+(1-\lambda)y),c\}.\]
If this maximum equals $c$, it is certainly at most
$\lambda\max\{\psi(x),c\}+(1-\lambda)\max\{\psi(y),c\}$,
since these two latter maximums are each at least $c$.
Hence the inequalities for convexity and quasiconvexity always
hold when the original maximum equals $c$.
Alternatively, if $\max\{\psi(\lambda x+(1-\lambda)y),c\}
=\psi(\lambda x+(1-\lambda)y)$, then using
$\psi(x)\le\psi'(x)$ and $\psi(y)\le\psi'(y)$,
we see that convexity of $\psi$ gives the inequality for
convexity of $\psi'$, and quasiconvexity of $\psi$ gives the
inequality for quasiconvexity of $\psi'$.

Next, suppose $\psi$ is quasiconcave. Without loss of generality,
say that $\psi(x)\le\psi(y)$. Then
$\psi'(\lambda x+(1-\lambda)y)=\max\{\psi(\lambda x+(1-\lambda)y),c\}
\ge\max\{\psi(x),c\}=\psi'(x)\ge\min\{\psi'(x),\psi'(y)\}$,
and $\psi'$ is quasiconcave.

Preservation of lower semicontinuity follows from
\lem{semicontinuous_inf}, where we note that $c$ is continuous
as a function from $X$ to $\bar{\bR}$. It remains to show
upper semicontinuity is preserved.
Suppose $\psi$ is upper semicontinuous, and let $x\in X$.
If $\psi'(x)=\infty$, upper semicontinuity at $x$ vacuusly holds.
Fix $y>\psi'(x)$. Since $\psi'(x)\ge c$, we have $y> c$,
so $\psi(x)=\psi'(x)>y$,
and upper semicontinuity gives us a neighborhood $U$ of $x$
on which $\psi(\cdot)$ is less than $y$. Since $y> c$,
we have $\psi'(\cdot)=\max\{c,\psi(\cdot)\}<y$ on $U$.
Hence $\psi'$ is upper semicontinuous.
\end{proof}

\begin{theorem}[Sion's minimax \cite{Sio58}]\label{thm:sion_reals}
Let $V_1$ and $V_2$ be real topological vector spaces, and let
$X\subseteq V_1$ and $Y\subseteq V_2$ be convex. Let
$\alpha:X\times Y\to\bR$ be semicontinuous and quasisaddle.
If either $X$ or $Y$ is compact, then
\[\adjustlimits \inf_{x\in X}\sup_{y\in Y}\alpha(x,y)
=\adjustlimits\sup_{y\in Y}\inf_{x\in X}\alpha(x,y).\]
\end{theorem}

\sionextended*

\begin{proof}
First, note that the inf-sup is always at least the sup-inf.
This is because these expressions can be thought of as two players,
one choosing $x$ and trying to minimize $\alpha(x,y)$, and the
other choosing $y$ and trying to maximize $y$;
in the inf-sup, the sup player chooses $y$ after already knowing $x$,
and therefore has more information and is better positioned to maximize
$\alpha(x,y)$ than in the sup-inf, where the inf player goes second.

Now, let
\[a\coloneqq \adjustlimits\sup_{y\in Y}\inf_{x\in X}\alpha(x,y),\qquad
b\coloneqq \adjustlimits\inf_{x\in X}\sup_{y\in Y}\alpha(x,y).\]
We have $a,b\in\bar{\bR}$, and $a\le b$. We wish to show $a=b$.
Suppose by contradiction that $a<b$. Then we can pick $a',b'\in\bR$
such that $a<a'<b'<b$. We then define
$\alpha':X\times Y\to\bR$ by $\alpha'(x,y)\coloneqq a'$ if
$\alpha(x,y)\le a'$, $\alpha'(x,y)\coloneqq b'$ if $\alpha'(x,y)\ge b'$,
and $\alpha'(x,y)\coloneqq\alpha(x,y)$ if
$\alpha(x,y)\in[a',b']$.

Note that $\alpha'(x,y)=\max\{a',\min\{b',\alpha(x,y)\}\}$.
By \lem{max_preservations}, we know that taking a maximum
with a constant preserves quasiconvexity, quasiconcavity,
and upper and lower semicontinuities. By negating the function,
it also follows that taking a minimum with a constant preserves
these properties. From this it follows that $\alpha'$
is quasisaddle and semicontinuous, since $\alpha$ has these properties.

Now, since $a=\adjustlimits\sup_{y\in Y}\inf_{x\in X}\alpha(x,y)$
and since $a'>a$, we know that for all $y\in Y$, there exists
some $x\in X$ for which $\alpha(x,y)<a'$.
This means that for all $y\in Y$, there exists $x\in X$ for which
$\alpha'(x,y)=a'$. Hence 
$\adjustlimits\sup_{y\in Y}\inf_{x\in X}\alpha'(x,y)=a'$.
Similarly, since $b=\adjustlimits\inf_{x\in X}\sup_{y\in Y}\alpha(x,y)$
and since $b'<b$, we know that for all $x\in X$, there exists
some $y\in Y$ for which $\alpha(x,y)>b'$. This means that for all
$x\in X$, there exists $y\in Y$ for which $\alpha'(x,y)=b'$.
Hence $\adjustlimits\inf_{x\in X}\sup_{y\in Y}\alpha'(x,y)=b'$.
By \thm{sion_reals}, we then have
\[b'=\adjustlimits\inf_{x\in X}\sup_{y\in Y}\alpha'(x,y)
=\adjustlimits\sup_{y\in Y}\inf_{x\in X}\alpha'(x,y)=a'.\]
But this is a contradiction, since we picked $a'<b'$.
We conclude that we must have had $a=b$ to begin with,
as desired.
\end{proof}

\section{Distance measures}\label{app:calculus}

\rules*

\begin{proof}
It is clear that all of the functions from \defn{rules}
are smooth on $(0,1)$ and increasing on $[0,1]$, where
we interpret $\hs(0)=\ls(0)=-\infty$. It is also clear
that all these functions evaluate to $1$ at $1$ and to $0$
at $1/2$. It remains to show that $\Brier$, $\ls$, and $\hs$
are proper. To do so, we need to show that
$ps(q)+(1-p)s(1-q)$ is uniquely optimized at $q=p$ when
$s$ is one of these functions and $p\in(0,1)$.
Fix such $p\in(0,1)$, and observe that the critical
points of the expression we wish to maximize
are the points $q$ such that $ps'(q)=(1-p)s'(1-q)$.

For $\ls(q)=1-\log(1/q)=1+(\log e)\ln q$, the critical points
$q$ satisfy $(\log e)p/q=(\log e)(1-p)/(1-q)$, or $p/(1-p)=q/(1-q)$.
Noting that the function $x/(1-x)$ is increasing on $(0,1)$,
and hence injective on $(0,1)$, we conclude that the only
critical point is $q=p$. Moreover, at the boundaries $q=0$
and $q=1$, we clearly have $p\ls(q)+(1-p)\ls(1-q)=-\infty$,
whereas in the interior the expression is finite. Hence the
unique maximum must occur at $q=p$.

For $\hs(q)=1-\sqrt{(1-q)/q}=1-\sqrt{1/q-1}$, we have
$\hs'(q)=1/2\sqrt{q^3(1-q)}$, so the critical points $q$
satisfy $p/2\sqrt{q^3(1-q)}=(1-p)/2\sqrt{(1-q)^3q}$,
or $p/q=(1-p)/(1-q)$, which once again only occurs at $q=p$.
At the boundaries, we once again have
$p\hs(q)+(1-p)\hs(1-q)=-\infty$ for $q=0$ or $q=1$, so
the unique maximum occurs at $q=p$.

Finally, for $\Brier(q)=1-4(1-q)^2=-4q^2+8q-3$,
we have $\Brier'(q)=8(1-q)$, so the critical points $q$
satisfy $8p(1-q)=8(1-p)q$, which again implies $q=p$.
This time, the boundary points are finite,
but we can use the second order condition: the second derivative
of $p\Brier(q)+(1-p)\Brier(1-q)$ is
$p\Brier''(q)+(1-p)\Brier''(1-q)$. Noting that
$\Brier''(q)=-8$, this is $-8p-8(1-p)=-8<0$. Hence the critical
point is a maximum, and since it is unique (with the boundaries
$0$ and $1$ not being critical even if we extend the domain
of the function), we conclude it is the unique maximum.
\end{proof}

\begin{lemma}\label{lem:inequalities}
For any $x\in[0,1]$, we have
\[\frac{x^2}{2}\le 1-\sqrt{1-x^2}\le 1-H\left(\frac{1+x}{2}\right)
\le x^2\le x.\]
Additionally, $x^2$ and $1-\sqrt{1-x}$ are convex functions on $[0,1]$.
\end{lemma}

\begin{proof}
$x^2\le x$ is clearly true for $x\in[0,1]$. To see that
$x^2/2\le 1-\sqrt{1-x^2}$, note that this is equivalent to
$y/2\le 1-\sqrt{1-y}$ for $y\in[0,1]$ (by setting $y=x^2$);
the latter is clearly true at $y=0$, so it suffices to show the right
hand side grows faster. Taking derivatives, it suffices to show
$1/2\le 1/2\sqrt{1-y}$, which is clearly true for $y\in[0,1]$.

Next, write
\[1-H((1+x)/2)=1-((1+x)/2)\log2/(1+x)-((1-x)/2)\log2/(1-x)\]
\[=1-(1+x)/2-(1-x)/2+((1+x)/2)\log(1+x)+((1-x)/2)\log(1-x)\]
\[=\frac{1}{\ln 4}((1+x)\ln(1+x)+(1-x)\ln(1-x)).\]
Let $\alpha(x)=(1+x)\ln(1+x)+(1-x)\ln(1-x)$.
We show that $\alpha(x)/x^2$ is increasing and that
$\alpha(x)/(1-\sqrt{1-x^2})$ is decreasing; this suffices to
show the desired inequalities, since it means we only
need to check $x=1$, where the inequalities
$1-\sqrt{1-x^2}\le 1-H((1+x)/2)\le x^2$ hold with equality.

The derivative of $\alpha(x)$ is $\ln(1+x)-\ln(1-x)$.
The derivative of $\alpha(x)/x^2$ is therefore
$x^2\ln(1+x)-x^2\ln(1-x)-2x(1+x)\ln(1+x)-2x(1-x)\ln(1-x)$
divided by $x^4>0$ (for $x\in(0,1)$). This simplifies
to $-2x\ln(1-x^2)-x^2\ln((1+x)/(1-x))$. This is positive
if and only if $2\ln(1-x^2)+x\ln((1+x)/(1-x))$ is negative.
This expression equals $0$ at $x=0$, so it suffices to show
it is decreasing on $(0,1)$. The derivative is
$-2x/(1-x^2)+\ln((1+x)/(1-x))$, which is again $0$ at $x=0$,
so it again suffices to show the derivative is negative on $(0,1)$.
The derivative of this expression is $-4x^2/(1-x^2)^2$, which
is finally a quantity that is clearly negative, completing the argument;
hence $\alpha(x)/x^2$ is increasing on $[0,1]$.

The derivative of $\alpha(x)/(1-\sqrt{1-x^2})$ is
\[(1-x-\sqrt{1-x^2})\ln(1-x)-(1+x-\sqrt{1-x^2})\ln(1+x)\]
divided by some denominator which is positive on $(0,1)$.
This equals
\[-x\ln(1-x^2)-(1-\sqrt{1-x^2})\ln((1+x)/(1-x)).\]
Note that $\ln(1-x^2)=-x^2-x^4/2-\dots-x^{2i}/i-\dots$
and that $\ln((1+x)/(1-x))
=\ln(1+x)-\ln(1-x)=2x+2x^3/3+\dots+2x^{2i-1}/(2i-1)+\dots$,
so the expression equals
\[x^2\sum_{i=1}^\infty x^{2i-1}/i
-(1-\sqrt{1-x^2})\sum_{i=1}^\infty x^{2i-1}/(i-1/2)
=(\sqrt{1-x^2}-(1-x^2))\sum_{i=1}^\infty -x^{2i-1}/i(2i-1)<0.\]
Hence $\alpha(x)/(1-\sqrt{1-x^2})$ is decreasing on $[0,1]$, as desired.

It is clear that $x^2$ and $1-\sqrt{1-x}$ are convex functions
on $[0,1]$, as their second derivatives are $2>0$
and $(1/4)(1-x)^{-3/2}>0$ (for $x\in(0,1)$) respectively.
\end{proof}

\relations*

\begin{proof}
We use \lem{inequalities}.
The chain $\h^2\le\JS\le\Ess^2\le\Delta$ follows from the
inequalities there, while the inequalities $\Delta^2\le\Ess^2$
and $1-\sqrt{1-\Ess^2}\le\h^2$ follow from Jensen's inequality
combined with the convexity of $x^2$ and $1-\sqrt{1-x}$.

Finally, to show inequality $\JS\le \h^2/\ln 2$ we
only need to compute the limit of $\alpha(x)/(1-\sqrt{1-x^2})$
as $x\to 0$, since this ratio is decreasing with $x$
(where $\alpha(x)$ is defined as in the proof of \lem{inequalities}).
To do that it suffices to use $\alpha(x)=x^2+O(x^4)$
and $1-\sqrt{1-x^2}=x^2/2+O(x^4)$, so the limit is $2$.
Hence the limit of $(1-H((1+x)/2))/(1-\sqrt{1-x^2})$ as $x\to 0$
is $1/\ln 2$, meaning this ratio is always at most $1/\ln 2$.
Similarly, to show the inequality $\Ess^2\le (\ln 4)\JS$,
we only need to compute the limit of $\alpha(x)/x^2$ as $x\to 0$.
Again using $\alpha(x)=x^2+O(x^4)$, the limit is $1$, so
the ratio $(1-H((1+x)/2))/x^2$ is always at least $1/\ln 4$.
\end{proof}

\ampinequality*

\begin{proof}
Set $f(x)\coloneqq 1-(1-x)^k$. Clearly, when $x\in[0,1]$,
we have $f(x)\in[0,1]$, so $f\colon[0,1]\to[0,1]$.
Note $f(0)=0$, $f(1)=1$, and that $f(x)$ is increasing on $[0,1]$.
If $k=1$, we have $f(x)=x$, and the inequalities trivially hold;
therefore, assume $k>1$. Then $f'(x)=k(1-x)^{k-1}$
and $f''(x)=-k(k-1)(1-x)^{k-2}$, meaning that $f(x)$ is concave
on $[0,1]$; we also have $f'(0)=k$ and $f''(0)=-k(k-1)$.
From this we conclude that $f(x)\le kx$, proving the upper bound
(as $f(x)\le 1$ is clear).

For the lower bound, note that $f'''(x)=k(k-1)(k-2)(1-x)^{k-3}$,
which is non-negative on $[0,1]$. This means that
$f''(x)\ge -k(k-1)$ on $[0,1]$, that $f'(x)\ge k-k(k-1)x$
on $[0,1]$, and that $f(x)\ge kx-(k(k-1)/2)x^2=kx(1-(k-1)x/2)$
on $[0,1]$. If $(k-1)x\le 1$, we get $f(x)\ge kx/2$.
If $(k-1)x\ge 1$, we have $f(x)\ge 1-e^{-kx}\ge1-1/e\ge 1/2$.
This completes the proof.
\end{proof}

\hmixture*

\begin{proof}
Note that the squared-Hellinger distance is one minus the fidelity,
that is, $\h^2(\mu_1,\mu_2)=1-F(\mu_1,\mu_2)$ where
$F(\mu_1,\mu_2)=\sum_x\sqrt{\mu_1[x]\mu_2[x]}$
(this is easy to check from the definition of $\h^2$).
Now write
\begin{align*}
\h^2(\nu_0^\mu,\nu_1^\mu)&=1-\sum_{x\in\bigcup_{a} S_a}
    \sqrt{\nu_0^\mu[x]\nu_1^\mu[x]}\\
&=1-\sum_{a\in A}\sum_{x\in S_a}
    \sqrt{\mu[a]\nu_0^a[x]\mu[a]\nu_1^a[x]}\\
&=1-\bE_{a\leftarrow\mu}\left[\sum_{x\in S_a}
    \sqrt{\nu_0^a[x]\nu_1^a[x]}\right]\\
&=\bE_{a\leftarrow\mu}\left[1-\sum_{x\in S_a}
\sqrt{\nu_0^a[x]\nu_1^a[x]}\right]\\
&=\bE_{a\leftarrow\mu}\left[\h^2(\nu_0^a,\nu_1^a)\right].\qedhere
\end{align*}
\end{proof}

\section{Quantum amplitude estimation}\label{app:amplitude}

We show the following strengthening of \thm{amplitude_estimation},
which follows from \cite{BHMT02}.

\begin{theorem}[Amplitude estimation]\label{thm:stronger_amplitude_estimation}
Suppose we have access to a unitary $U$ (representing a quantum
algorithm) which maps $|0\rangle$ to $|\psi\rangle$,
as well as access to a projective measurement $\Pi$,
and we wish to estimate $p\coloneqq\|\Pi|\psi\rangle\|_2^2$
(representing the probability the quantum algorithm accepts).
Fix $\epsilon,\delta\in(0,1/2)$.
Then using at most $(100/\epsilon)\cdot\ln(1/\delta)$
controlled applications of $U$ or $U^\dagger$
and at most that many applications of $I-2\Pi$, we
can output $\tilde{p}\in[0,1]$ such that
$|\tilde{p}-p|\le\epsilon$ with probability
at least $1-\delta$.

Further, this can be tightened to a bound that
depends on $p$, as follows.
For any positive real number $T$,
there is an algorithm which depends on
$\epsilon$, $\delta$, and $T$ (but not on $p$)
which uses at most $T$
applications of the unitaries (as above)
and outputs $\tilde{p}\in[0,1]$ with the following
guarantee: if $T$ is at least
$\lfloor(100/\epsilon)\sqrt{\max\{p,\epsilon\}}\cdot\ln(1/\delta)
\rfloor$,
then
$|\tilde{p}-p|\le \epsilon$ with probability at least
$1-\delta$.
\end{theorem}

\begin{proof}
\cite{BHMT02} showed that an algorithm which makes $M$
controlled calls to the unitary
$U(I-2\ket{0}\bra{0})U^{-1}(I-2\Pi)$
and one additional call to $U$ can output $\tilde{p}$
such that
\[|\tilde{p}-p|
\le \frac{2\pi\sqrt{p(1-p)}}{M}+\frac{\pi^2}{M^2}\]
with probability at least $8/\pi^2\ge 4/5$.
If we pick $M$ such that $M\ge 8/\sqrt{\epsilon}$
and $M\ge 8\sqrt{p}/\epsilon$,
then this is at most $(\pi/4+\pi^2/64)\gamma\le \gamma$.
Note that $M$ must be an integer, and that the number
of applications of $U$ or $U^{-1}$ is $2M+1$.
Hence to get this success probability, it suffices
to have $T\ge 3+(16/\epsilon)\sqrt{\max\{p,\epsilon\}}$,
or $T\ge (19/\epsilon)\sqrt{\max\{p,\epsilon\}}$.

To generalize to other success probabilities, we amplify
this algorithm by repeating $2k+1$ times and returning
the median estimate. The probability that this is
still wrong is the probability that at least $k+1$
out of $2k+1$ of the estimates were wrong, which is
\[\sum_{i=1}^{k+1}\binom{2k+1}{k+1-i}q^{k+i}(1-q)^{k+1-i}
\le q^{k+1}(1-q)^{k}\sum_{i=1}^{k+1}\binom{2k+1}{k+1-i}\]
\[=q^{k+1}(1-q)^{k}2^{2k}
=q(1-(1-2q)^2)^k
\le qe^{-k(1-2q)^2}.\]
Hence to get this below $\delta$, we just need
$k\ge (1/(1-2q)^2)\ln(1/q\delta)$, or
$k\ge 2.6\ln(1/\delta)-4$. Since $k$ must be an integer,
but we can always choose it so that $2k+1$ is at most
$5.2\ln(1/\delta)$. Multiplying this by the bound from
before, we get that it suffices for $T$
to be at most
$(100/\epsilon)\sqrt{\max\{p,\epsilon\}}\cdot\ln(1/\delta)$,
as desired.
\end{proof}

\newpage

\fi 

\phantomsection\addcontentsline{toc}{section}{References} 
\renewcommand{\UrlFont}{\ttfamily\small}
\let\oldpath\path
\renewcommand{\path}[1]{\small\oldpath{#1}}
\emergencystretch=1em 
\printbibliography

\end{document}